\documentclass{amsart}

\usepackage{amsmath,amssymb}
\usepackage{latexsym}
\usepackage{graphicx}
\usepackage{color}
\usepackage[framemethod=TikZ]{mdframed}
\usepackage{subfigure}
\usepackage{multirow}
\usepackage{mathptmx}
\usepackage{bm}
\usepackage{url}

\newcommand{\ud}{\,\mathrm{d}}
\newcommand{\RR}{\mathbb{R}}

\newcommand{\ZZ}{\mathbb{Z}}
\newcommand{\NN}{\mathbb{N}}

\newtheorem{thm}{Theorem}[section]
\newtheorem{cor}[thm]{Corollary}
\newtheorem{lem}[thm]{Lemma}
\newtheorem{defn}[thm]{Definition}
\newtheorem{assumption}[thm]{Assumption}
\newtheorem{example}[thm]{Example}
\newtheorem{rem}[thm]{Remark}

\newcommand{\argmax}{\operatornamewithlimits{argmax}}

\title{Wave-shape function analysis -- when cepstrum meets time-frequency analysis}

\author[C.-Y.~Lin]{Chen-Yun~Lin}
\address{Department of Mathematics, University of Toronto}
\email{cylin@math.toronto.edu}

\author[L.~Su]{Li~Su}
\address{Research Center for Information Technology Innovation, Academia Sinica}
\email{li.sowaterking@gmail.com}

\author[H.-T.~Wu]{Hau-Tieng~Wu}
\address{Department of Mathematics, University of Toronto}
\email{hauwu@math.toronto.edu}

\begin{document}

\begin{abstract}
We propose to combine cepstrum and nonlinear time-frequency (TF) analysis to study {multiple} component oscillatory signals with time-varying frequency and amplitude and with time-varying non-sinusoidal oscillatory pattern.
The concept of cepstrum is applied to eliminate the wave-shape function influence on the TF analysis, and we propose a new algorithm, named {\em de-shape synchrosqueezing transform (de-shape SST)}.
The mathematical model, {\em adaptive non-harmonic model}, is introduced and the de-shape SST algorithm is theoretically analyzed.
In addition to simulated signals, several different physiological, musical and biological signals are analyzed to illustrate the proposed algorithm.
\end{abstract}

\maketitle

\section{Introduction}

Time series is a ubiquitous datatype in our life, ranging from finance, medicine, geology, etc.
It is clear that different problems depend on different interpretation and processing of the observed time series. In some situations, the information can be easily read from the signal, for example, the cardiac arrest could be easily read from the electrocardiogram (ECG) signal; in others, it is less accessible, for example, the heart rate variability (HRV) hidden inside the ECG signal; in yet others, the information might be masked and cannot be read directly from the observed time series. %
This comes from the fact that while the time series encodes the temporal dynamics of the system under observation, most of {the} time the dynamical information we could perceive is masked or deformed due to the observation process and the nature of the physiology.
When the information is masked or deformed but exists in the observed time series, we might need more sophisticated approaches to extract the information relevant to the situation we have interest {in}. In general, inferring the dynamical information from the time series is challenging.

We could view the challenge in two parts. First, we need to choose a model to quantify the recorded signal, which captures the features or information about the underlying dynamical system we have interest {in}. This model could come from the field background knowledge, or in some cases it could be relatively blind. Second, we need to design an associated algorithm to extract the desired features from the recorded signal. With the acquired features, we could proceed to study the dynamical problem we have interest {in}.

We take physiological signals to illustrate the challenge of the modeling issue. Note that the procedure could be applied to other suitable fields. It is well known that how the signal oscillates contains plenty of information about a person's health condition.
Based on the oscillatory behavior and the widely studied Fourier analysis, common features we could discuss are the frequency, which represents how fast the signal oscillates, and the amplitude, which represents how strongly the signal oscillates at that frequency. However, these features have been found limited when the signal is not stationary, which is a property shared by most physiological signals. Indeed, { these signals mostly oscillate} with time-varying frequency and amplitude. To capture this property, we could consider the {\em adaptive harmonic model} encoding the features {\em instantaneous frequency (IF)} and {\em amplitude modulation (AM)} \cite{Daubechies_Lu_Wu:2011,Chen_Cheng_Wu:2014}; that is, the signal is modeled as
\begin{equation}\label{Introduction:AHM}
f_0(t)=A(t)\cos(2\pi\phi(t)),
\end{equation}
where $A$ is a smooth positive function and $\phi$ is a smooth monotonically increasing function. In other words, at time $t$, the signal $f$ repeats itself as a sinusoidal function within about $1/\phi'(t)$ seconds, and the oscillation is modulated by the AM function $A(t)$. These features have been proved useful and could well represent the physiological dynamics and health status, and have been applied to different problems \cite{Lin_Hseu_Yien_Tsao:2011,Lin_Wu_Tsao_Yien_Hseu:2014,Wu_Talmon_Lo:2015}.

There are actually more detailed features embedded in the oscillatory signals that cannot be captured by (\ref{Introduction:AHM}). One particular feature is the non-sinusoidal oscillatory pattern.
For example, respiratory flow signals usually do not oscillate like the sinusoidal function, since the inspiration is normally shorter than the expiration, and this difference is intrinsic to the respiratory system \cite{Benchetrit:2000}.
These observations lead us to consider the following model  \cite{Wu:2013,Yang:2014,Hou_Shi:2016},
\begin{equation}\label{Introduction:ANHM0}
f_1(t)=A(t)s(\phi(t)),
\end{equation}
where $A(t)$ and $\phi(t)$ are the same as those of (\ref{Introduction:AHM}), and $s$ is a real 1-periodic function {with the unitary $L^2$ norm}, that is $s(t+1)=s(t)$ for all $t$, so that the first Fourier coefficient $\hat{s}(1)\neq 0$, which could be different from the cosine function.
We call the periodic signal $s(t)$ the {\em wave-shape function}, $\phi(t)$ the phase function, the derivative $\phi'(t)$ the IF, and $A(t)$ the AM of $f_1(t)$.  Note that when $s$ is smooth enough, (\ref{Introduction:ANHM0}) could be expanded pointwisely by the Fourier series as
\begin{equation}\label{Introduction:ANHM1}
f_1(t)=\sum_{k=0}^\infty A(t)a_k\cos(2\pi k\phi(t)+\alpha_k),
\end{equation}
where $a_k\geq0$, $k\in\NN\cup\{0\}$ are associated with the Fourier coefficients of $s$, $\alpha_0=0$ and $\alpha_k\in[0,2\pi)$, $k\in\NN${, and $a_0^2+2\sum_{k=1}^\infty a_k^2=1$}. Note that we could have two different aspects of the same signal $f_1$. First, we could view it as an oscillatory signal with one oscillatory component with non-sinusoidal oscillation (\ref{Introduction:ANHM0}). Second, we also could view it as an oscillatory signal with multiple oscillatory components with the cosine oscillatory pattern (\ref{Introduction:ANHM1}); in this case, we call the first oscillatory component $A(t)a_1\cos(2\pi \phi(t)+\alpha_1)$ the {\em fundamental component} and $A(t)a_k\cos(2\pi k\phi(t)+\alpha_k)$, $k\geq 2$, the {\em $k$-th multiple} of the fundamental component. Clearly, the IF of the $k$-th multiple is $k$-times that of the fundamental component. Note that $A(t)a_0$ could be viewed as the {trend coming from the DC (direct current) or zero-frequency} term of the wave-shape function. While the second viewpoint (\ref{Introduction:ANHM1}) is better for the theoretical analysis, the first viewpoint (\ref{Introduction:ANHM0}) is more physical in several applications.

{ Let us take} the ECG signal as an example, where the IF, AM and the wave-shape function have their own physiological meanings.
The oscillatory morphology of the ECG signal, the wave-shape function, reflects not only the electrical pathway inside the heart and how the sensor detects the electrophysiological dynamics, but also the respiration as well as the heart anatomy. Several clinical diseases are diagnosed by reading the oscillatory morphology. With these physiological understanding, it is better to consider model (\ref{Introduction:ANHM0}) to study the ECG signal and view IF, AM and wave-shape function as separate features.
As for IF and AM, it is well known that while the rate of the pacemaker is constant, the heart rate generally is not constant. The discrepancy comes from neural and neuro-chemical influences on the pathway from the pacemaker to the ventricle. This non-constant heart beat rate could be modeled as the IF of the ECG signal.
The AM of the ECG signal is directly related to the respiration via the variation of thoracic impedance. Indeed, when the lung is full of air, the thoracic impedance increases and hence and ECG amplitude decreases, and vice versa. Note that IF and AM could be captured by both (\ref{Introduction:ANHM0}) and (\ref{Introduction:ANHM1}).

Several algorithms were proposed to extract IF and AM from a given oscillatory signal in the past decade, like empirical mode decomposition \cite{Huang_Shen_Long_Wu_Shih_Zheng_Yen_Tung_Liu:1998}, reassignment method (RM) \cite{Auger_Flandrin:1995}, synchrosqueezing transform (SST) \cite{Daubechies_Lu_Wu:2011}, concentration of frequency and time \cite{Daubechies_Wang_Wu:2016}, Blaschke decomposition \cite{Coifman_Steinerberger:2015}, iterative filtering \cite{Cicone_Liu_Zhou:2014}, sparsification approach \cite{Hou_Shi:2013a}, approximation approach \cite{Chui_Mhaskar:2016}, convex optimization \cite{Kowalski_Meynard_Wu:2015}, Gabor transform based on different selection criteria \cite{Balazs_Dorfler_Jaillet_Holighaus_Velasco:2011,Ricaud_Stempfel_Torresani:2014}, etc. In general, we could view these methods as a nonlinear time-frequency (TF) analysis.
However, to capture the wave-shape function, an extra step is needed  -- we could fit a non-sinusoidal periodic function to the signal after/while extracting the IF by, for example, applying the functional regression \cite[Section 4.7]{Chui_Lin_Wu:2015}, designing a dictionary \cite{Hou_Shi:2016} or unwrapping the phase \cite{Yang:2014}. The obtained features have been used to study field problems, such as the sleep stage prediction \cite{Wu_Talmon_Lo:2015}, the blood pressure analysis \cite{Wu_Chang_Wu_Wang_Yang_Wu:2015}. See \cite{Daubechies_Wang_Wu:2016} for a review of the applications.

As useful as the above-mentioned model and algorithms { are} to extract dynamical features from time series, there are{, however,} several unsolved limitations.
First, for most physiological signals, the wave-shape function {varies from time to} time. The time-varying wave-shape function might { prevent} the current available methods from extracting the wave-shape function. We will provide physiological details in Section \ref{Section:Limitation}.
Second, there might be more than one oscillatory component in a signal, and each oscillatory component has its own wave-shape function. See Figure \ref{fig:Introduction:Example0sig} for an photoplethysmogram signal (PPG) as an example. In this PPG signal, there are two oscillatory components, hemodynamic rhythm and respiratory rhythm.
Third, although we could obtain reasonable information about IF and AM from the above-mentioned approaches, when the signal has multiple oscillatory components with non-sinusoidal waves, these methods are limited. In particular, the multiples of different fundamental components will interfere with each other.
Furthermore, an automatic determination of the number of oscillatory components becomes more difficult when each component oscillates with a non-sinusoidal wave.
Hence, modifications are needed.

\begin{figure}[h!]
\begin{centering}
\includegraphics[width= \textwidth]{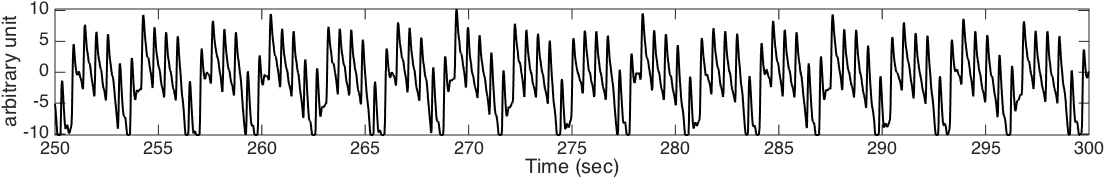}
\end{centering}
\caption{\label{fig:Introduction:Example0sig} A photoplethysmogram signal. It is visually clear that there are two rhythms inside the signal. The faster one is associated with the heart rate, which beats about 100 times per minute; the slower one is associated with the respiration, which is about 18 times per minute.}
\end{figure}

In this paper, we resolve these limitations. We introduce the adaptive non-harmonic model to model oscillatory signals with multiple components and time-varying wave-shape functions. Motivated by cepstrum, we introduce an algorithm called {\em de-shape SST} to alleviate the influence caused by non-sinusoidal wave-shape functions in the TF analysis. Hence, we provide an enhanced TF representation in the following sense -- the de-shape SST would provide a TF representation with only IF and AM information without the influence of non-sinusoidal wave-shape functions.

We illustrate the effectiveness of de-shape SST by showing results on a simulated signal. In this example, the clean signal $f(t)$ is composed of two oscillatory components $f_{1}$ and $f_{2}$, where $f_{1}(t)=A_1(t)(\sum_{k\in\ZZ}\delta_k\star h)(t)\chi_{[0,60]}(t)$, $A_{1}(t)=1.5e^{-(\frac{t-20}{100})^2}$, $h(t)=e^{-18t^2}$, $\delta_k$ is the Dirac delta measure supported at $k$, $\chi_{I}$ is the indicator function supported on $I\subset \RR$ and $f_{2}(t)=A_{2}(t)\text{mod}(\phi_{2}(t),1)$, where $A_{2}(t)>0$ and $\phi_{2}'(t)>0$ are two non-constant smooth function and $\text{mod}(x,1):=x-\lfloor x\rfloor$ for $x\in\RR$ and $\lfloor x\rfloor$ means the largest integer less than or equal to $x$. Clearly, $f_{1}$ oscillates at the fixed frequency $\phi'_1(t)=1$ with a non-sinusoidal wave-shape function -- the wave-shape function of $f_1$ looks like a Gaussian function; $f_{2}$ oscillates with a time-varying frequency with the non-sinusoidal wave, which behaves like a sawtooth wave. This signal is sampled at rate $100$Hz, from $t=0$ to $t=100$ seconds.
Figure \ref{fig:Introduction:Example1sig} shows the two constituents of the total signal $f(t)=f_{1}(t)+f_{2}(t)$, as well as $A_{2}(t)$ and $\phi_{2}'(t)$. Note that $f_1$ ``lives'' during only part of the full time observation time interval.
The panels in Figure \ref{fig:Introduction:Example1} show the results of short-time Fourier transform of $f(t)$, the SST of $f(t)$ and the de-shape SST of $f(t)$. It is clear that compared with the TF representation provided by STFT or SST, the TF representation provided by the de-shape SST contains only the fundamental frequency information of the two oscillatory components, even when the wave-shape function is far from the sinusoidal wave. More discussions will be provided in Section \ref{Section:Numerics}, including how $A_{2}(t)$ and $\phi_{2}(t)$ are generated.

\begin{figure}[h!]
\begin{centering}
\includegraphics[width= \textwidth]{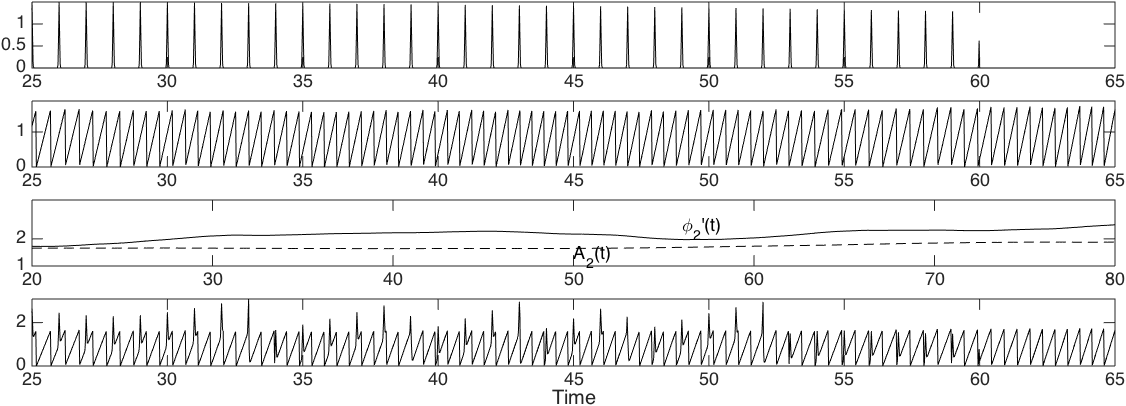}
\end{centering}
\caption{\label{fig:Introduction:Example1sig}Top panel: $f_1(t)$; second panel: $f_2(t)$; third panel: the $A_2(t)$ (dashed curve) and $\phi_2'(t)$ (solid curve) of $f_2(t)$; {bottom panel: $f(t)$.} To enhance the visibility, we only show the signal from the 25-th second to the 65-th second.}
\end{figure}

\begin{figure}[h!]
\begin{centering}
\includegraphics[width= \textwidth]{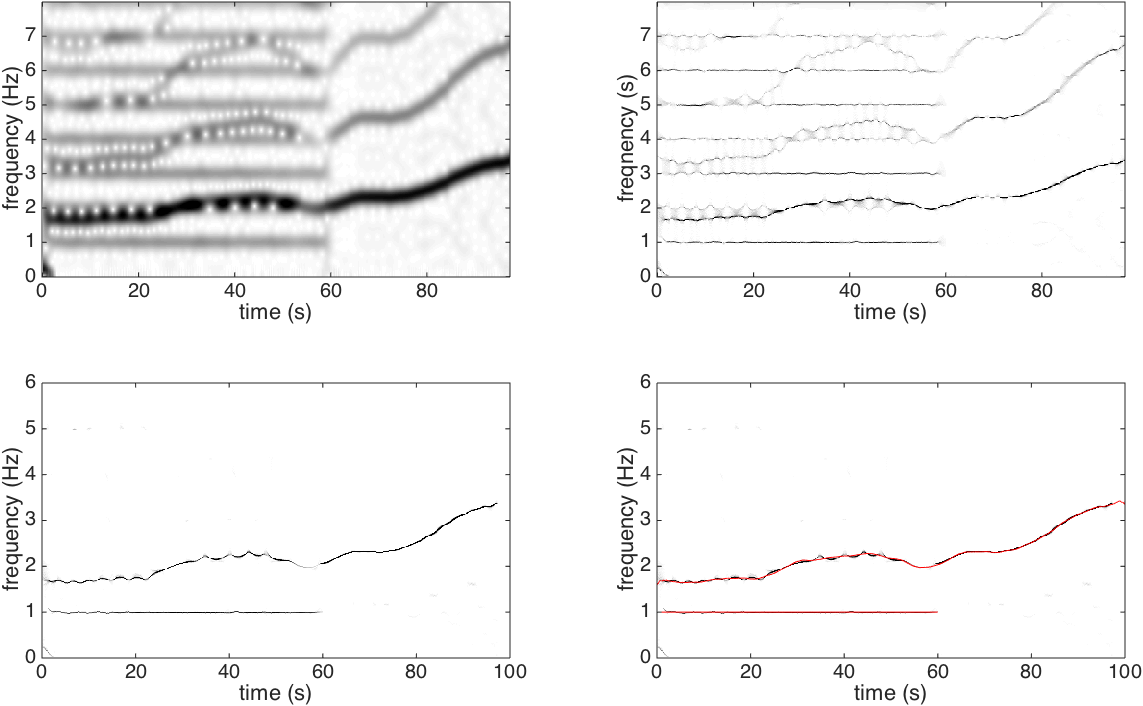}
\end{centering}
\caption{\label{fig:Introduction:Example1}Upper left: the short time Fourier transform (STFT) of $f(t)$; upper right: the synchrosqueezed STFT of $f(t)$; lower left: the de-shape SST of $f(t)$. It is clear that the de-shape SST provides only the fundamental frequency information of the two oscillatory components, even when the wave-shape function is far from the sinusoidal wave; lower right: the de-shape SST of $f(t)$ superimposed with the ground truth IF's of both components in red. To enhance the visibility, we show the de-shape SST only up to 6 Hz in the frequency axis.}
\end{figure}

The paper is organized in the following way. In Section \ref{Section:Model}, we discuss the limitation of model (\ref{Introduction:ANHM0}), and provide a modified model, the adaptive non-harmonic model. In Section \ref{Section:Cepstrum}, the existing cepstrum algorithms in the engineering field are reviewed, and the new algorithm de-shape SST is introduced. The theoretical justification of the de-shape SST is postponed to Appendix \ref{Appendix:Proof}. Section \ref{Section:Numerics} shows the numerical results of de-shape SST on several different simulated, medical, musical, and biological signals. Section \ref{Section:NumericalIssues} discusses numerical issues of the de-shape SST algorithm. Section \ref{Section:Conclusions} summarizes the paper.

\section{Adaptive Non-harmonic Model}\label{Section:Model}

In this section we first review the phenomenological model based on the wave-shape function (\ref{Introduction:ANHM0}) fixed over time. Then, we discuss the relationship between the wave-shape function and several commonly encountered physiological signals, and discuss limitations. This discussion leads us to introduce the adaptive non-harmonic model.

We start from introducing some notations. The Schwartz space is denoted as $\mathcal{S}$; the tempered distribution space, which is the dual space of the Schwartz space, is denoted as $\mathcal{S}'$; $\ell^p$, where $p>0${,} indicates the sequence space including all sequences $x:\NN\to \RR$ so that $\sum_{n\in\NN} |x(n)|^p<\infty$, where $x(n)$ is the $n$-th element of the sequence $x$. For each $k\in\NN\cup\{\infty\}$, $C^k$ indicates the space of continuous functions with all the derivatives continuous, up to the $k$-th derivates, and $C_c^k$ indicates the space of compactly supported continuous functions with all the derivatives continuous, up to the $k$-th derivates. For each $p\in\NN$, $L^p$ includes all measurable functions $f$ so that $\int_{-\infty}^\infty|f(x)|^pdx<\infty$; $L^\infty$ includes all measurable functions which are bounded almost surely. For $f\in\mathcal{S}'$ and $g\in \mathcal{E}'$, where $\mathcal{E}'$ is the set of compactly supported distribution{s}, denote $f\star g$ to be the convolution. We will interchangeably use $\mathcal{F}f$ or $\hat{f}$ to denote the Fourier transform of the function $f\in \mathcal{S}'$. When $f\in L^1(\RR)$, the Fourier transform is equally defined as $f(\xi)=\int_{-\infty}^\infty f(t)e^{-i2\pi \xi t}dt$; when $f\in \mathcal{E}'$, we know that $\hat{f}(\xi)=\langle f,e^{-i2\pi \xi\cdot}\rangle$, where $\langle\cdot,\cdot\rangle$ indicates the evaluation of the {distributions} $f$ at the $C^\infty$ function $e^{-i2\pi \xi t}$.
For a periodic function $s$, denote $\hat{s}(k)$, $k\in\ZZ$, to be its Fourier series coefficients. For each $N\in\NN$, denote the Dirichlet kernel $D_N(x):=\sum_{\ell=-N}^N e^{i2\pi \ell x}$.

\subsection{Review of the wave-shape function}

We continue the discussion of the model (\ref{Introduction:ANHM0})
\begin{equation}\label{model:linearRelationship}
f(t)=A(t)s(\phi(t)),
\end{equation}
where $A\in C^1(\RR)$ is strictly positive, $\phi\in C^2(\RR)$ is strictly monotonically increasing, and $s\in C^{1,\alpha}$, $\alpha>1/2$, is a 1-periodic function {with the unitary $L^2$ norm} so that its Fourier series coefficients satisfy $|\hat{s}(1)|>0$, $|\hat{s}(k)|\leq \delta |\hat{s}(1)|$ for some $\delta\geq0$ and $\sum_{k=N+1}^\infty |k\hat{s}(k)|\leq \theta$ for some $\theta\geq 0$ and $N\in\NN$.
We need more conditions for {the} analysis. Take $0\leq \epsilon\ll1$, we require
$|\phi''(t)|\leq\epsilon \phi'(t)$ and $|A'(t)|\leq \epsilon \phi'(t)$ for all $t$. This means that we allow the IF and AM to vary in time, as long as the variations are slight from one period to the next.

\subsection{Limitations in modeling physiological signals}\label{Section:Limitation}

While many physiological signals are oscillatory and have ``similar'' patterns, at first glance they could be well modeled by (\ref{model:linearRelationship}) and the analysis could proceed. However, it is not always possible to do so. In this section we provide examples to discuss limitations.

\subsubsection{Electrocardiographic signal} \label{Section:ECGdiscussion}

{The ECG signal, which provides information of the electrical activity of the heart, is ubiquitous in healthcare setting now.}
It not only contains a wealth of information {regarding} the cardiac/cardiovascular health but also provides a unique non-invasive portal to physiological dynamical states of the human body, via for example the HRV assessment. While the HRV{, the non-constant heart rate,} could be studied by evaluating the IF of the ECG signal and well estimated by the ``R peak detection algorithm'', in several cases a modern TF analysis could help improve the estimation accuracy \cite{Herry_Frasch_Wu:2015}.

We now discuss the limitation of modeling the ECG signal by (\ref{model:linearRelationship}).
Take the relationship between the RR and QT intervals of the lead II  ECG signal $E_{\texttt{II}}(t)$ as an example\footnote{{The P, Q, R, S, and T are significant landmarks of the ECG signal. The P wave represents atrial depolarization. The Q wave is any downward deflection after the P wave. The R wave follows as an upward deflection, which is spiky, and the S wave is any downward deflection after the R wave. The Q wave, R wave, and S wave form the QRS complex, which corresponds to the ventricular depolarization. The T wave follows the S wave, which represents the ventricular repolarization.} The QT interval (respectively RR interval) is the length of the time interval between the start of the Q wave and the end of the T wave of one heart beat (respectively two R landmarks of two consecutive heart beats). We could view the R peak as a surrogate of the cardiac cycle, and hence the RR interval could be viewed as a surrogate of the inverse of the heart rate. {See Figure \ref{fig:ECG} for an example of the P, Q, R, S, and T landmarks and the RR and QT intervals. For more information about ECG signal, we refer the readers to \cite{ecg}.}}.
The nonlinearity relationship between the QT interval and the RR interval has been well accepted -- for example, the Fridericia's formula (QT interval is proportional to the cubic root of RR interval) \cite{Fridericia:1920} or a fully nonlinear depiction {\cite[Figure 3]{Fossa_Zhou:2010}}.

If we model $E_{\texttt{II}}(t)$ by the model (\ref{model:linearRelationship}), and have
\begin{align}
E_{\texttt{II}}(t)&=A_{\texttt{II}}(t)s_{\texttt{II}}(\phi_{\text{II}}(t))=\sum_{\ell=1}^\infty A_{\texttt{II}}(t)c_{\texttt{II}}(\ell)\cos(2\pi \ell\phi_{\text{II}}(t)+\alpha_{\texttt{II},\ell}),\label{model:ECG0}
\end{align}
where $\alpha_{\texttt{II},0}=0$, $\alpha_{\texttt{II},\ell}\in [0,2\pi)$ when $\ell\in\NN$, $c_{\texttt{II}}(1)>0$ and $c_{\texttt{II}}(\ell)\geq 0$ for $\ell\neq 1$ are related to Fourier series coefficients of the wave-shape function $s_{\texttt{II}}$. Here, $s_{\texttt{II}}$ models the oscillation in the lead II ECG signal, which is non-sinusoidal.
Note that under this model, the QT interval has to be ``almost'' linearly related to the RR interval. To see this, suppose there is a 1-periodic function $s_{\text{II}}$ for the lead II ECG signal, where the R peak happens at time $0$, and a monotonically increasing function $\phi_{\text{II}}(t)$ so that the ECG signal could be modeled as $s_{\text{II}}(\phi_{\text{II}}(t))$; that is, the wave-shape function is fixed all the time. Suppose that the $k$-th R peak happens at time $t_k$, where $k\in\ZZ$; that is, $t_k=\phi_{\text{II}}^{-1}(k)$. By the mean value theorem, in this model we have the following relationship for the ECG signal at time $t\in[t_k,t_{k+1}]$ :
\begin{align}
&s_{\text{II}}(\phi_{\text{II}}(t)){=\sum_{k\in\ZZ}s_{\text{II}}(\phi_{\text{II}}(t))\chi_{[t_k,t_{k+1})}(t)}\nonumber\\
=&\,\sum_{k\in\ZZ}s_{\text{II}}(\phi_{\text{II}}(t_k)+(t-t_k)\phi'_{\text{II}}(\tilde{t}_k))\chi_{[t_k,t_{k+1})}(t)=\sum_{k\in\ZZ}s_{\text{II}}\left(\frac{t-t_k}{1/\phi'_{\text{II}}(\tilde{t}_k)}\right)\chi_{[t_k,t_{k+1})}(t),\label{Model:LinearRelationship}
\end{align}
where $\tilde{t}_k\in[t_k,t_{k+1}]$, $\chi_{[t_k,t_{k+1})}$ is a indicator function defined on $[t_k,t_{k+1})$, and the second equality holds since $\phi_{\text{II}}(t_k)=k$ and $s$ is $1$-periodic. While the RR interval between the $k$-th and the $(k+1)$-th R peaks is proportional to $1/\phi'_{\text{II}}(\tilde{t}_k)$ up to order $O(\epsilon)$ by the slowly varying IF assumption of $\phi_{\text{II}}$, we know that the wave-shape function is approximately linearly dilated according to $1/\phi'_{\text{II}}(\tilde{t}_k)$. If the the wave-shape function is linearly dilated according to the RR interval, then the QT interval should be linearly related with the RR interval and hence the claim. Clearly, this model contradicts the physiological finding that the QT interval should be nonlinearly related to the RR interval, so we need a modified model to better quantify the ECG signal.

Furthermore, note that since the cardiac axis varies from time to time due to respiration, physical activity and so on, even if the RR interval is fixed all the time and we focus on the lead II ECG signal, we cannot find a fixed wave-shape function to exactly model the ECG signal. Note that the wave-shape function variation caused by respiration could be applied to extract the respiratory information from the ECG signal \cite{Chui_Lin_Wu:2015}.

\subsubsection{Respiratory signal}

Oscillation is a typical pattern in breathing in normal subjects. It is well known that there is a rhythmic controller in the Pre-B\"otzinger complex in the brain stem which regularly oscillates. In a normal subject the respiratory period is about 5 seconds per cycle. Note that when we are awake, we could also control our respiration by our will, but to simplify the discussion, we do not take this into account.
The existence of breathing pattern variability has been well known \cite{Benchetrit:2000}. For example, the period of each respiratory cycle for a normal subject under normal status varies according to time. The ratio between the length of inspiration period and the length of expiration period is not linearly related to the instantaneous respiratory rate, and its variability also contains plenty of physiological information \cite{Benchetrit:2000}. In other words, the wave-shape function associated with the respiration is not fixed all the time.
By the same argument as that for the ECG signal, this nonlinear relationship between the instantaneous respiratory rate and the wave-shape function could not be fully captured by (\ref{model:linearRelationship}).

The same argument holds for the other physiological signals, like the photoplethysmography signal that reflects the hemodynamics information, the capnogram signal that monitors inhaled and exhaled concentration or partial pressure of carbon dioxide and is a surrogate of the oscillatory dynamics of the respiratory system, and so on.

\subsubsection{Natural vibration of stiff strings}\label{Section:StiffStringModel}

In this section we discuss the signal commonly encountered in music, in particular the sound generated by the string musical instrument. The acoustic signal generated by the string musical instrument could be well modeled by the transversal vibration behavior of an ideal string. For an ideal string of length $L>0$ placed on  $[0,L]$ with both ends fixed ideally, when the string has {\em stiffness}; that is, there is a restoring force proportional to the displacement (or more generally the bending angle), we could consider the following differential equation for $y\in \RR^+\times [0,L]$ satisfying \cite{fletcher1964normal,fletcher2010physics}
\begin{equation}\label{Equation:Music:Stiffness}
\mu \frac{\partial^2 y}{\partial t^2} = T \frac{\partial^2 y}{\partial x^2} - ESK^2 \frac{\partial^4 y}{\partial x^4}\,,
\end{equation}
where $\mu>0$ is mass per unit length, $T\geq0$ is tension, $E\geq0$ is Young's modulus of the string, $S\geq0$ is the cross-sectional areas of the string, and $K\geq0$ is the radius of gyration, with the initial condition $y(0,x)=0$ for all $x\in[0,L]$ and the boundary condition $y(t,0)=0$ and $y(t,L)=0$ for all time $t\geq 0$.

Consider the case of a {\em pinned string}, that is, $y(t,0)=y(t,L)=0$ and $\frac{\partial^2 y}{\partial x^2}(t,0)=\frac{\partial^2 y}{\partial x^2}(t,L)=0$ for all $t$. The solution $y(t,x)$ is the transversal displacement of the string point $x$ at time $t$ \cite{fletcher1964normal,fletcher2010physics}, which is the linear combination of the {\em normal modes} represented by $y_n(t,x)=\sin(2\pi k_nx)\sin(2\pi {\xi_n} t)$ with $k_n=\frac{n}{2L}$, $\xi_1=\frac{1}{2L}\sqrt{\frac{T}{\mu}}$ and
\begin{equation}
\xi_n = n\xi_1\sqrt{1+\beta n^2}\,,
\end{equation}
where $\beta=\frac{\pi^2ESK^2}{TL^2}$; that is, the $n$-th component with the $n$-th lowest frequency is deviated from $n\xi_1$ in a nonlinear way.
In other words, the sound associated with the solution oscillates with a non-sinusoidal wave and the fundamental frequency is $\frac{1}{2L}\sqrt{\frac{T}{\mu}}$ with several multiples. Clearly, when $E=0$, (\ref{Equation:Music:Stiffness}) is reduced to the wave equation, and the solution is well known.

In music signal processing, this phenomenon is well known as {\em inharmonicity}, which appears in instruments, like piano and guitar. In these instruments, natural vibration appears after the excitation (i.e., plucking or pressing the keyboard) of the modes. For the piano, $\beta$ is in the ranges from around $10^{-4}$ to $10^{-3}$. Obviously, the sound with inharmonicity does not well fit (\ref{model:linearRelationship}).

\subsection{Time-varying wave-shape function}

The above discussions indicate that we need a model with time-varying wave-shape functions.
Thus, we wish to generalize (\ref{model:linearRelationship}). To achieve this goal, we will directly generalize the equivalent expression (\ref{Introduction:ANHM1}) to capture an oscillatory signal with the ``time-varying wave-shape function''.

\begin{defn}[Adaptive non-harmonic function]\label{Definition:ANHFunction}
Take $\epsilon { > 0}$, a non-negative { $\ell^1$} sequence $c =\{c(\ell)\}_{\ell=0}^\infty$, $0<C<\infty$, and $N\in\NN$. The set $\mathcal{D}_{\epsilon}^{c,C,N}\subset C^1(\RR)\cap L^\infty(\RR)$ of {\it adaptive non-harmonic (ANH) functions} is defined as the set consisting of functions
\begin{align}
f(t)=\frac{1}{2}B_0(t)+\sum_{\ell=1}^\infty B_\ell(t)\cos(2\pi \phi_\ell(t))\label{model:nonlinearRelationship}
\end{align}
satisfying the following conditions:
\begin{itemize}

\item the {\em regularity condition} :
\begin{align}
&B_\ell\in C^1(\RR)\cap L^\infty(\RR),\, { \mbox{ for } \ell = 0 ,\ldots \infty,}\\
&\phi_\ell\in C^2(\RR), \, { \mbox{ for } \ell= {1} ,\ldots \infty.}
\end{align}
 For all $t\in\RR$, $B_\ell(t)\geq 0$ for all {$\ell=0,1,2,\ldots,\infty$} and $\phi'_\ell(t)>0$ for all $\ell=1,\ldots,\infty$.

\item the {\em time-varying wave-shape} condition: for all $t\in\RR$,
\begin{equation}\label{Condition:ANH:phi_ell}
\left|\phi'_\ell(t)-\ell\phi'_1(t)\right|\leq \epsilon \phi'_1(t)
\end{equation}
for all $\ell=1,\ldots,\infty$,
\begin{equation}\label{Condition:ANH:B_ell}
{B_\ell(t)\leq c(\ell)B_1(t)}
\end{equation}
for all {$\ell=0,1,\ldots,\infty$, 
\begin{equation}\label{Condition:ANH:B_elltail1}
\sum_{\ell=N+1}^\infty B_\ell(t)\leq \epsilon \sqrt{\frac{1}{4}B_0(t)^2+\frac{1}{2}\sum_{\ell=1}^\infty B_\ell(t)^2},
\end{equation}
and 
\begin{equation}\label{Condition:ANH:B_elltail2}
\sum_{\ell=1}^\infty \ell B_\ell(t)\leq C \sqrt{\frac{1}{4}B_0(t)^2+\frac{1}{2}\sum_{\ell=1}^\infty B_\ell(t)^2}.
\end{equation}}

\item the {\em slowly varying} condition: for all $t\in\RR$,
\begin{align}
\label{def:slow_varying}
&|B_\ell'(t)|\leq \epsilon c(\ell)\phi_1'(t),\, { \mbox{ for } \ell = 0 ,\ldots \infty,}\\
&|\phi_\ell''(t)|\leq \epsilon \ell\phi_1'(t),\, { \mbox{ for } \ell= {1} ,\ldots \infty,}
\end{align}
 and $\|\phi_1'(t)\|_{L^\infty}<\infty$.

\end{itemize}

\end{defn}

The adjective {\em adaptive} in ANH function indicates that the frequency and amplitude are time-varying, and the adjective {\em non-harmonic} indicates that the oscillation might be non-sinusoidal. When $\frac{B_{\ell}(t)}{B_1(t)}$ are constants for all $\ell=0,1,\ldots,\infty$ and $\phi'_\ell(t)=\ell\phi_1'(t)+\alpha_\ell$ for some $\alpha_\ell\in\RR$ for all $\ell=1,\ldots,\infty$, (\ref{model:nonlinearRelationship}) is reduced to (\ref{Introduction:ANHM1});
when the other conditions for the wave-shape function in (\ref{model:linearRelationship}) are further satisfied, (\ref{model:nonlinearRelationship}) is reduced to (\ref{model:linearRelationship}). 
Thus, (\ref{model:nonlinearRelationship}) is a direct generalization of (\ref{Introduction:ANHM1}) by allowing $c_\ell$ and $\alpha_\ell$ in (\ref{Introduction:ANHM1}) to vary, which quantifies the time-varying wave-shape function.
We call $B_1(t)\cos(2\pi \phi_1(t))$ the {\it fundamental component} and $\phi'_1$ the {\it fundamental IF} (or {\em pitch} in the music signal analysis) of the signal $f(t)$. {Note that the condition $|\phi_1''(t)|\leq \epsilon \phi_1'(t)$ says that locally the fundamental IF is nearly constant, but it does not imply that the fundamental IF is nearly constant globally.} By a slight abuse of terminology, for $\ell>1$, we call $B_\ell\cos(2\pi\phi_\ell(t))$ the {\it $\ell$-th multiple}, although $\phi_\ell$ might not be proportional to $\phi_1$. {Note that we can ``view'' $\sqrt{\frac{1}{4}B_0(t)^2+\frac{1}{2}\sum_{\ell=1}^\infty B_\ell(t)^2}$ as the AM of $f(t)$. This comes from the fact that in (\ref{Introduction:ANHM1}), $A(t)^2a_0^2+2\sum_{\ell=1}^\infty A(t)^2a^2_\ell=A^2(t)$, and $B_k(t)$ is the the generalization of $A(t)a_k$ in (\ref{Introduction:ANHM1}) for $k=0,1,\ldots,\infty$. In this definition, however, we do not control how large $\sqrt{\frac{1}{4}B_0(t)^2+\frac{1}{2}\sum_{\ell=1}^\infty B_\ell(t)^2}$ should be. Also note that the series $\big(\frac{B_0(t)/2}{\sqrt{\frac{1}{4}B_{0}(t)^2+\frac{1}{2}\sum_{\ell=1}^\infty B_{\ell}(t)^2}},\,\frac{B_1(t)/\sqrt{2}}{\sqrt{\frac{1}{4}B_{0}(t)^2+\frac{1}{2}\sum_{\ell=1}^\infty B_{\ell}(t)^2}},\ldots,\big)$ has the unitary $\ell^2$ norm, which is a generalization of the assumption that the wave-shape function has the unitary $L^2$ norm. The condition (\ref{Condition:ANH:B_elltail1}) says that only the first $N$ multiples are significant. The condition (\ref{Condition:ANH:B_elltail2}) is a direct generalization of the $C^{1,\alpha}$ condition of the wave-shape function in (\ref{model:linearRelationship}).}

To see how the wave-shape function varies according to time, denote {$t_k:=\phi_1^{-1}(k)$}. Clearly, for signals in $I_k:=[t_k,t_{k+1})$, we could not find a single 1-periodic function $s(t)$ so that $\frac{f}{B_1(t)}|_{I_k}$ is the composition of $s$ and $\phi_1(t)$.
Thus, we could view the model (\ref{model:nonlinearRelationship}) either as an adaptive non-harmonic model with one oscillatory component with the {\it time-varying wave-shape function}, or as an adaptive harmonic model with many oscillatory components with the sinusoidal wave pattern.

\begin{defn}[Adaptive non-harmonic model]\label{DefBClassMultipleTimeSeries}
Take $\epsilon { > 0}$ and $d>0$.
The set $\mathcal{D}_{\epsilon,d}\subset C^1(\RR)\cap L^\infty(\RR)$ consists of \textit{superposition of ANH functions}, that is
\begin{equation}
f(t)=\sum_{k=1}^{K}f_{k}(t)
\end{equation}
for some finite $K>0$ and \
\begin{equation*}
f_{k}(t)={\frac{1}{2}B_{k,0}(t) +}\sum_{\ell={1}}^\infty B_{k,\ell}(t)\cos(2\pi \phi_{k,\ell}(t))\in \mathcal{D}_{\epsilon_k}^{c_{k},C_k,N_k}
\end{equation*}
for some $0\leq \epsilon_k\leq \epsilon$, non-negative sequence $c_{k}=\{c_k(\ell)\}_{\ell={0}}^\infty$, $0<C_k<\infty$ and $N_k\in\NN$, where  {for all $t\in\RR$,} the fundamental IF's of all ANH functions satisfy 
\begin{itemize}
\item the {\em frequency separation} condition:
\begin{equation}\label{condition_Cepsilon_d}
\phi'_{k,1}(t)-\phi'_{k-1,1}(t)\geq d
\end{equation}
for $k=2,\ldots,K$
\item the {\em non-multiple} condition: for each $k=2,\ldots,K$, $\phi'_{k,1}(t)/\phi'_{\ell,1}(t)$ is not an integer for $\ell=1,\ldots,k-1$.
\end{itemize}
\end{defn}

\section{De-shape SST}\label{Section:Cepstrum}

In this section, we propose an algorithm, de-shape SST, to study a given oscillatory signal. De-shape SST provides a TF representation which contains essentially the IF and AM information of the fundamental component of each ANH function and removes the influence caused by the non-trivial wave-shape function. In Section \ref{Section:ReviewCepstrum}, we provide a review of how cepstrum is applied in engineering. 
{In Section \ref{Section:TFcepstrum}, the short time cepstral transform (STCT) is introduced with a theoretical justification in Theorem \ref{Theorem:TimeVaryingCepstrum} to generalize cepstrum to the time-frequency analysis.}
The proof of the theorem is postponed to the Appendix.
{In Section \ref{Section:deShape}, we introduce the inverse STCT that will be used in the de-shape algorithm. }
In Sections {\ref{Section:deshapeSTFT}-\ref{Section:Synchrosqueezing}}, the de-shape STFT and de-shape SST are discussed.

\subsection{A quick review of cepstrum}\label{Section:ReviewCepstrum}

Cepstrum is a commonly applied signal processing technique \cite{Oppenheim_Schafer:2009}.
One motivation of introducing cepstrum is the {\em pitch detection} problem in music (recall that pitch means the fundamental frequency).
It is closely related to the {\em homomorphic signal processing}, which aims at converting signals structured by complicated algebraic systems into simple ones.
Since its invention in 1963 \cite{bogert1963quefrency}, the cepstrum has been applied in various discrete-time signal processing problems, such as detecting the echo delay, deconvolution, feature representations for speech recognition like the Mel-Frequency Cepstral Coefficients (MFCC), and estimating the pitch of an audio signal. A thorough review of the cepstrum can be found in \cite{oppenheim2004frequency,Oppenheim_Schafer:2009}.

We start from recalling the {\em complex cepstrum}. For a suitable chosen signal $f(t)\in\RR$, the {\em cepstrum}, denoted as $\tilde{f}^{C}(q)$, where $q\in\RR$ is called {\em quefrency}\footnote{The term ``cepstrum'' is invented by interchanging the consonants of the first part of the word ``spectrum'' in order to signify their difference. Similarly, the word ``{quefrency}'' is the inversion of the first part of ``frequency''. By definition, the quefrency has the same unit as time.}, is defined as the inverse Fourier transform of the logarithm of the Fourier transform \cite{Oppenheim_Schafer:2009}:
\begin{equation}
\tilde{f}^{C}(q) := \int \log \hat{f}(\xi) e^{2\pi iq\xi } d\xi \,,\label{eq:cepstrum_basic}
\end{equation}
whenever the inverse Fourier transform of $\log \hat{f}(\xi)$ makes sense, where $\log$ is defined on a chosen branch. We call the domain of $\tilde{f}^C$ the {\em quefrency domain}. Numerically, since the computation of the complex cepstrum requires the phase unwrapping process, it causes instability. Therefore, we could also consider the {\em real cepstrum}, denoted as $\tilde{f}^{R}$, which is represented as
\begin{equation}
\tilde{f}^{R}(q) = \int \log |\hat{f}(\xi)| e^{-2\pi iq\xi }d\xi\,,\label{eq:real_cepstrum}
\end{equation}
whenever the Fourier transform of $\log |\hat{f}(\xi)|$ makes sense. Note that there is no difference to take the Fourier transform or inverse Fourier transform since the signal is in general real, so we take the Fourier transform instead of the inverse Fourier transform.
In audio signal analysis, the logarithm operation on the magnitude spectrum can be interpreted to be an approximation of the perceptual scale of sound intensity, thus it is conventionally measured in dB.
Intuitively, the cepstrum measures ``the rate of the harmonic peaks per Hz'', namely the {\em period} of the signal, where the period is the inverse of the frequency; that is, the prominent peaks in the cepstrum indicate the periods and their multiples in the signal.
Besides periodicity detection, this method has also been used in a wide variety of fields which requires deconvolution of a {\em source-filter model}.

The main idea behind cepstrum is to find ``the spectral distribution of the spectrum'', which contains the period information of the signal. It is effective since it could transform the ``slow-varying envelope'' of the spectrum to the low-quefrency range, separated from the fast-varying counterpart of the spectrum, which is transformed to the high-quefrency range and represents the period information of the signal.

\begin{example}
We consider an acoustic signal to demonstrate how the overall idea beyond cepstrum or homomorphic signal processing could help in signal processing when the signal comes from a complicated combination of two components. A human voice $f\in C^\infty$ could be modeled by the glottal vibration, which is a pulse sequence $g=p\sum_{k\in\ZZ}\delta_{T_0k}\in \mathcal{S}'$, where $p\in \mathcal{S}$ and $T_0>0$, convolved with the impulse response of the vocal tract $h\in\mathcal{S}$ so that $\hat{h}$ is a non-negative function, i.e., $f(t)=(g\star h)(t)$. A mission of common interest is to separate these two components.

First, the Fourier transform converts the convolution into multiplication in the frequency domain $\hat{f}(\xi)=\hat{h}(\xi)\hat{g}(\xi)$, where $\hat{g}=\frac{1}{T_0}\hat{p}\star \sum_{k\in\ZZ}\delta_{k/T_0}\in C^\infty$ by the Poisson summation formula. Second, the logarithm converts multiplication into addition, but we have to be careful when we take the logarithm. To simplify the discussion, we assume that supports of both $\hat{g}$ and $\hat{h}$ are positive-valued. Thus, $\log(\hat{f}(\xi)) = \log(\hat{g}(\xi))+\log(\hat{h}(\xi))$. Thus, the convolution operator in the time domain becomes the addition operator. Although under our simplified assumption, $\hat{g}\in C^\infty\cap L^\infty$ and $\hat{h}\in\mathcal{S}$, after taking logarithm we might not be able to define the Fourier transform. So we further assume that $\log(\hat{g}(\xi)),\log(\hat{h}(\xi))\in\mathcal{S}'$ so that we could apply the Fourier transform. For example, if $h$ is a Gaussian function, $\log(\hat{h}(\xi))$ is a quadratic polynomial function. We call the domain where $\mathcal{F}\log(\hat{f})$ is defined the quefrency domain.

In summary, the periodic glottal excitation is modeled as a series of harmonic peaks in the frequency domain by the Poisson summation formula (contributing to pitch), while the frequency response of the vocal tract, $\hat{h}(\xi)$, contributes to the amplitude of the spectrum. { Let us further assume} that after taking Fourier transform on $\log(\hat{f}(\xi))$ the glottal excitation lies in the high quefrency region while the vocal tract in the low quefrency region\footnote{In the music processing, the high-quefrency part in the cepstrum is related to the pitch while the low-quefrency part to timbre (i.e., sound color). }, then a simple high pass filtering, which is called the {\em liftering} (again an interchange of the consonants of ``filtering'') process, can separate the two components. One simple example of $h$ is that when $h$ is a Gaussian function, the Fourier transform of $\log(\hat{h}(\xi))$ is proportional to the second distributional derivative of the Dirac measure supported at $0$. These two components could then be reconstructed by reversing the procedure -- apply the Fourier transform, take the exponential and apply the inverse Fourier transform. The whole process is called the {\em homomorphic deconvolution}.

\end{example}

Although the real cepstrum avoids phase unwrapping, it is still limited by evaluating the logarithm, which is prone to numerical instability either in synthetic data or real-world data. To address this issue, it has been proposed in the literature to replace the logarithm by the {\em generalized logarithm function} \cite{kobayashi1984spectral,tokuda1994mel,taxt1997comparison},
\begin{equation}
L_\gamma(x):=\frac{|x|^\gamma-1}{\gamma},
\end{equation}
where $\gamma>0$, or the {\em root function} \cite{lim1979spectral,alexandre1993root,taxt1997comparison}, defined as
\begin{equation}
g_\gamma(x):=|x|^\gamma,
\end{equation}
where $\gamma>0$. Note that $L_\gamma$ approximates the logarithm function as $\gamma\to 0$. As $g_\gamma$ and $L_\gamma$ are related by a constant and a dilation, there is no practical difference which relaxation we choose. Thus, although we could also consider the generalized logarithm function $L_\gamma$ \cite{kobayashi1984spectral,tokuda1994mel}, to simplify the discussion, in this paper we relax the real cepstrum by the root function $g_\gamma$, and we call the resulting ``cepstrum'' the {\em $\gamma$-generalized cepstrum} (In the literature it is also called the root cepstrum):
\begin{equation}
\tilde{f}_\gamma(q) := \int g_\gamma(\hat{f}(\xi)) e^{{-}2\pi iq\xi }d\xi\,.\label{eq:root_cepstrum}
\end{equation}
There are several proposals for the choice of $\gamma$.
First, when $\gamma=2$, the formulation is equivalent to the autocorrelation function of $f$, which is a basic feature for {\em single} pitch detection but has been found unfeasible for multipitch estimation (MPE). To deal with the issue of multipitch, we should consider $\gamma<2$ \cite{kobayashi1984spectral,tolonen2000computationally}.
$\gamma=0.67$ is suggested by the nonlinear relationship between the sound intensity and perceived loudness determined by experiment, known as a case of Stevens' power law, which states that  the sound intensity $x$ and the perceived loudness $y$ are related by $y\propto|x|^{0.67}$ \cite{stevens1957psychophysical,hermansky1990perceptual,tolonen2000computationally}.
Previous researches also suggested $\gamma$ to be $0.6$ \cite{kraft2014polyphonic}, $0.25$ \cite{indefrey1985design} and $0.1$ \cite{klapuri2008multipitch}. In short, the $\gamma$-generalized cepstrum has been shown more robust to noise than the real or complex cepstrum in the literature of speech processing  \cite{lim1979spectral,alexandre1993root}.
In addition to the robustness, the $\gamma$-generalized cepstrum has been found useful in various problems like speech recognition \cite{hermansky1990perceptual}, speaker identification \cite{zhao2013analyzing}, especially in multiple pitch estimation \cite{indefrey1985design,tolonen2000computationally,klapuri2008multipitch,kraft2014polyphonic,su2015combining,su2016exploiting}. Due to its usefulness and for the sake of simplification, in the paper we focus on the $\gamma$-generalized cepstrum.

\subsection{Combining cepstrum and time-frequency analysis -- short time cepstral transform (STCT)}\label{Section:TFcepstrum}

As useful as the Fourier transform is for many practical problem, however, it has been well known that when the IF or AM is not constant, Fourier transform might not perform correctly. Indeed, for the ANH functions, since IF and AM are time-varying, the momentary behavior of oscillation is mixed up by the Fourier transform, and hence the cepstrum approach discussed in the previous section fails. To study this kind of dynamical signal, we need a replacement for the Fourier transform. A lot of efforts have been made in the past few decades to achieve this goal. TF analysis based on different principals \cite{Flandrin:1999} has attracted a lot of attention in the field and many variations are available. Well known examples include short time Fourier transform (STFT), continuous wavelet transform (CWT), Wigner-Ville distribution (WVD), etc. We refer the reader to \cite{Daubechies_Wang_Wu:2016} for a summary of the current progress of TF analysis. In this paper, we consider STFT, since it is a direct and intuitive generalization of the Fourier transform. A generalization of cepstrum to other TF analyses will be studied in {future works}.

Recall the definition of STFT. For a chosen window function $h\in\mathcal{S}$, the STFT of $f\in \mathcal{S}'$ is defined by
\begin{equation}
V^{(h)}_f(t, \xi) = \int f(\tau) h(\tau-t)e^{-i2\pi \xi (\tau-t)} \ud \tau\,,\label{eq: stft1}
\end{equation}
where $t\in\RR$ indicates time and $\xi\in\RR$ indicates frequency\footnote{The phase factor $e^{i2\pi\xi t}$ in this definition is not always present in the literature, leading to the name {\em modified STFT} for this particular form. To slightly abuse the notation, we still call it STFT.}. We call $V^{(h)}_f(t, \xi)$ the TF representation of the signal $f$.
Since STFT could capture the spectrum or local oscillatory behavior of a signal, we could combine the ideas of STFT and cepstrum, which leads to the {\em short time cepstral transform (STCT)}:
\begin{defn}
Fix $\gamma>0$. For $f\in\mathcal{S}'$ and $h\in\mathcal{S}$, we have the {\em short time cepstral transform (STCT)}:
\begin{equation}
C^{(h,\gamma)}_f(t, q) := \int g_\gamma(V^{(h)}_f(t, \xi)) e^{-i2\pi q \xi} \ud \xi,
\label{eq: rceps1}
\end{equation}
where $q\in\RR$.
\end{defn}

$q$ in $C^{(h,\gamma)}_f(t, q)$ is called the {quefrency}, and its unit is second or any feasible unit in the time domain. Clearly, $C^{(h,\gamma)}_f(t, \cdot)$ is the $\gamma$-generalized cepstrum of the signal $f(\cdot) h(\cdot-t)$ and in general $C^{(h,\gamma)}_f(t, q)$ is not positive. To show the well-definedness of STCT, note that while $f\in\mathcal{S}'$ and $h\in\mathcal{S}$, $V^{(h)}_f(t, \xi)\in C^\infty$ is smooth and slowly increasing on both time and frequency axes. By a slowly increasing $C^\infty$ function $f$, we mean that $f$ and all its derivatives have at most polynomial growth at infinity. Thus, we know that $g_\gamma(V^{(h)}_f(t, \cdot))$ is continuous and slowly increasing. Hence its Fourier transform can be well-defined in the distribution sense since a continuous slowly increasing function is a tempered distribution. In the special case that $f\in C^\infty\cap L^\infty$, $g_\gamma(V^{(h)}_f(t, \xi))$ is a continuous function vanishing at infinity faster than any power of $|\xi|$, and hence $C^{(h,\gamma)}_f(t, q)$ is a well-defined continuous function in the {quefrency} axis.

As discussed above, since the cepstrum provides the information about periodicity, we call $C^{(h,\gamma)}_f(t,q)$ the {\em time-periodicity} (TP) representation of the signal $f$. Before proceeding, we consider the following example to demonstrate how the STCT works.

\begin{example}
Consider the Dirac comb $f(t)=\frac{1}{\xi_0}\sum_{k\in\ZZ}\delta_{k/\xi_0}$, where $\xi_0>0$. This is the typical periodic distribution, and we could view it as an ANH function with $K=1$, the delta measure as the shape function, the constant fundamental frequency $\xi_0$ Hz and the constant fundamental period $1/\xi_0$, although the wave-shape function is more general than { what} we consider in the ANH model; {it is more general than the ANH model since with the delta measure $f(t)$ does not satisfy the ANH model.} By the Poisson's summation formula, $f(t)=\sum_{k\in\ZZ}e^{i2\pi k\xi_0 t}$, where the summation holds in the distribution sense.
Choose a smooth window function $h\in\mathcal{S}$ so that $\hat{h}$ is supported on $[-\Delta,\Delta]$, where $0<\Delta<\xi_0/2$. By a direct calculation, the STFT of $f$  is
\begin{equation}
V_f^{(h)}(t,\xi)=\sum_{k\in\ZZ}\hat{h}(\xi-k\xi_0),
\end{equation}
and since $\Delta<\xi_0/2$,
\begin{equation}
|V_f^{(h)}(t,\xi)|=\sum_{k\in\ZZ}|\hat{h}(\xi-k\xi_0)|.
\end{equation}
To evaluate the STCT, where $\gamma>0$, we need to evaluate $|V_f^{(h)}(t,\xi)|^\gamma$. Under our assumption, it is trivial and we have
\begin{equation}
|V_f^{(h)}(t,\xi)|^\gamma=\sum_{\ell\in\ZZ}|\hat{h}(\xi-\ell\xi_0)|^\gamma= (\sum_{\ell\in\ZZ}\delta_{\ell\xi_0}\star |\hat{h}|^\gamma)(\xi),
\end{equation}
where $|\hat{h}|^\gamma\in C_c^0(\RR)$. Note that the convolution is well-defined since $\sum_{\ell\in\ZZ}\delta_{\ell\xi_0}$ is a tempered distribution and $|\hat{h}|^\gamma$ is a compactly supported distribution. By taking Fourier transform of $|V_f^{(h)}(t,\xi)|^\gamma$ and applying the Poisson summation formula, the STCT of $f$ is
\begin{equation}
C_f^{(h,\gamma)}(t,q)=\frac{\widehat{|\hat{h}|^\gamma}(q)}{\xi_0}\sum_{\ell\in\ZZ}\delta_{\ell/\xi_0}{(q)}\,,
\end{equation}
which provides the period information.

\end{example}

This example indicates the overall behavior of STCT when there is only one periodic function with a non-sinusoidal wave. In general, when there are more than one oscillatory functions with non-sinusoidal waves and different fundamental frequencies, the calculation is no longer direct since the multiples of different oscillatory functions may collide. Moreover, since the frequency and amplitude are time-varying, the calculation is more intricate. For the signals in { the set $\mathcal{D}_{\epsilon,d}$ defined in Definition \ref{DefBClassMultipleTimeSeries}}, however, we have the following Theorem showing how STCT works.

Before stating the theorem, we make the following general assumption about the window function.

\begin{assumption}\label{Assumption:GeneralAssumption}
Fix $\epsilon {>0}$ and $d>0$. Take $f(t)=\sum_{k=1}^Kf_k(t)\in \mathcal{D}_{\epsilon,d}$ for some $K\geq 1$. Suppose the fundamental frequency satisfies
\begin{equation}
\inf_{t\in\RR} \phi_{1,1}'(t) > 0.
\end{equation}
Fix a window function $h\in \mathcal{S}$, which is chosen so that $\hat{h}$ is compactly supported and $\text{supp}(\hat{h})\subset[-\Delta,\Delta]$, where $\Delta>0$. Also assume that $\Delta$ is small enough so that
 \begin{equation}
 0<\Delta<\min\{\inf_{t\in\RR}\phi'_{1,1}(t)/4,\,d/4\}.
 \end{equation}
\end{assumption}

For a chosen window $h\in \mathcal{S}$, denote
\begin{equation}
I_k:=\int |h(x)||x|^kdx,
\end{equation}
where $k\in\{0\}\cup \NN$.
We mention that a more general window could be considered with more error terms showing up in the proof. Since these extra efforts do not provide more insight about the theory, we choose to work with this setup.

\begin{defn} {Let $\phi_{k,\ell}(t)$ and $B_{k,\ell}(t)$ for $k=1,\cdots, K$ and $\ell \in \mathbb{N}$ be defined as in Definition \ref{DefBClassMultipleTimeSeries}.}
Under Assumption \ref{Assumption:GeneralAssumption}, define
\begin{equation*}
\phi_{k,-\ell}(t):=-\phi_{k,\ell}(t)\mbox{ and }B_{k,-\ell}(t):=B_{k,\ell}(t)
\end{equation*}
for $\ell\in \NN$, 
\[
\phi_{k,0}=0\,,
\] 
and define a set of intervals
\begin{align}
Z_{k,\ell}(t)&=[\ell\phi'_{k,1}(t)-\Delta,\ell\phi'_{k,1}(t)+\Delta]\subset \RR\,,\label{Definition:Zkl}
\end{align}
associated with $f$, where $k\in\{1,\ldots,K\}$ and $\ell\in\ZZ$.
\end{defn}

The following Theorem describes the behavior of STCT when the signal is in $\mathcal{D}_{\epsilon,d}$. The proof of Theorem \ref{Theorem:TimeVaryingCepstrum} is postponed to Appendix \ref{Appendix:Proof}.

\begin{thm}\label{Theorem:TimeVaryingCepstrum}
Suppose Assumption \ref{Assumption:GeneralAssumption} holds.
The STFT of $f$ at time $t\in\RR$ is
\begin{align}
V^{(h)}_f(t,\xi)=
\frac{1}{2}\sum_{k=1}^K\sum_{\ell=-N_k}^{N_k} B_{k,\ell}(t)\hat{h}(\xi-\phi'_{k,\ell}(t))e^{i2\pi\phi_{k,\ell}(t)}+\epsilon(t,\xi),\label{MainTheorem:STFTExpansion}
\end{align}
where $\xi\in\RR$ and $\epsilon(t,\xi)$ is defined in (\ref{Lemma:Proof:Definition:epsilon}). Furthermore, $\epsilon(t,\xi)$ is of order $\epsilon$ and decays at the rate of $|\xi|^{-1}$ as $|\xi|\to \infty$.

Take $0<\gamma\leq 1$.
For each $k\in\{1,\ldots,K\}$, denote a series $b_k\in \ell^1$, where $b_k(j)=0$ for all $|j|>N_k$ and $b_k(j)=B^\gamma_{k,j}(t)$ for all $|j|\leq N_k$. Then, for $q>0$, we have
\begin{align}
C^{(h,\gamma)}_f(t,q)=\frac{\widehat{|\hat{h}|^\gamma}(q) }{2^\gamma}\sum_{k=1}^K  \hat{b}_k(q)  +E_1+ E_2,\label{Theorem:STCT:MainStatement}
\end{align}
where $\hat{b}_k$ is {is the discrete-time Fourier transform of $b_k$}, $E_1$ is defined in (\ref{Theorem:Statement:STCT:E1}), {which is the Fourier transform of $\delta_3$ defined in (\ref{Proof:Lemma:Definition:delta3part2}),} and $E_2$ is defined in (\ref{Theorem:Statement:STCT:E2}), which {is the Fourier transform of $\epsilon_3$ defined in (\ref{Theorem:Statement:STCT:epsilon4}) and} in general is a distribution. {When $K=1$, $E_1=0$, and when $K>1$ it} satisfies
\begin{equation*}
|E_1|\leq {2}\Delta I_0^\gamma \sum_{k=2}^K  B^\gamma_{k,1}(t)\|c_k^\gamma\|_{\ell^\infty}N_k\sum_{\ell=1}^{k-1}\Big[\frac{4\Delta}{\phi'_{\ell,1}(t)}+E^{(\ell)}(N_k)\Big],
\end{equation*}
where $E^{(\ell)}(N_k) { =o(N_k)}$ is defined in (\ref{Theorem:Statement:STCT:Eell}). $E_2$ satisfies $|E_2(\psi)|\leq \|{\epsilon_3}(t,\cdot)\|_{L^\infty} \|\hat{\psi}\|_{L^1}$ for all $\psi\in \mathcal{S}$, {and $\epsilon_3$} is of order {$\epsilon^\gamma$}.
\end{thm}

{
The equation (\ref{Theorem:STCT:MainStatement}) does not indicate the relationship between the relationship among $\frac{\widehat{|\hat{h}|^\gamma}(q) }{2^\gamma}\sum_{k=1}^K  \hat{b}_k(q)$, $E_1$, and $E_2$, so we could not conclude that we could obtain the inverse of the IF from the STCT. We need more conditions to obtain what we are after. The following corollary is immediate from Theorem \ref{Theorem:TimeVaryingCepstrum}.

\begin{cor}\label{Corollary:TimeVaryingCepstrum}
Fix $\epsilon>0$ and $d>0$. Take $f(t)=\sum_{k=1}^Kf_k(t)\in \mathcal{D}_{\epsilon,d}$ for some $K\geq 1$. In addition to Assumption \ref{Assumption:GeneralAssumption}, suppose 
\begin{equation}\label{Assumption:ANH:B_ell}
\frac{B_{k,\ell}(t)}{\sqrt{\frac{1}{4}B^2_{k,0}(t)+\frac{1}{2}\sum_{\ell=1}^\infty B^2_{k,\ell}(t)}} > \epsilon^{1/2}
\end{equation}
for all $k=1,\ldots,K$ and $\ell=0,1,\ldots,N_k$. 
Then, when $\epsilon<1$ is sufficiently small, $\Delta N_k$ is sufficiently small and $\sqrt{\frac{1}{4}B^2_{k,0}(t)+\frac{1}{2}\sum_{\ell=1}^\infty B^2_{k,\ell}(t)}$ is sufficiently large for $k=1,\ldots,K$, the term $\frac{\widehat{|\hat{h}|^\gamma}(q) }{2^\gamma}\sum_{k=1}^K  \hat{b}_k(q)$ in (\ref{Theorem:STCT:MainStatement}) dominates and $\hat{b}_k$ is a real, continuous and periodic function of period $1/\phi_{k,1}'(t)$ for $k=1,\ldots,K$.
\end{cor}

The assumption that $\sqrt{\frac{1}{4}B^2_{k,0}(t)+\frac{1}{2}\sum_{\ell=1}^\infty B^2_{k,\ell}(t)}$ is sufficiently large for $k=1,\ldots,K$ means that the ANH functions we have interest in have large enough AMs. 
Condition (\ref{Condition:ANH:B_elltail1}) and Assumption (\ref{Assumption:ANH:B_ell}) together mean that the first $N_k$ multiples of the fundamental component of the $k$-th ANH function are strong enough, while the remaining multiples are not significant. When $\gamma$ is chosen small enough, this assumption leads to the fact that $b_k(j)=B_{k,j}^\gamma(t)$ is close to $1$ for $|j|\leq N_k$, and ``small'' otherwise. Thus, the Fourier transform of the $\ell^1$ series $b_k$ reflects faithfully the inverse of the IF. We could call $1/\phi_{k,1}'(t)$ the {\em instantaneous period (IP)} of the $k$-th ANH function, which is the inverse of its fundamental frequency.

Note that the assumption (\ref{Assumption:ANH:B_ell}) can be generalized, but more conditions are needed to guarantee that we obtain the IP. For example, if the condition (\ref{Assumption:ANH:B_ell}) is failed for $\ell=2j$ so that $B_{k,2j}(t)=0$ for $j=1,\ldots,\lfloor N_k/2\rfloor$, then the $\ell^1$ series $b_k$ has an oscillation of period $2\phi'(t)$, and hence its Fourier transform is dominant in $\frac{1}{2\phi'(t)}$ instead of $\frac{1}{\phi'(t)}$. This will lead to an incorrect conclusion about the IF in the end. In this paper, to simplify the discussion, we focus on this assumption. See more discussions in the Discussion section.  

The bounds for $E_1$ and $E_2$ need some discussions.}
\begin{enumerate}
\item The {bound for} $E_1$ comes from controlling the possible overlaps between the multiples of different ANH components in the STFT $V^{(h)}_f(t,\xi)$. When $K=1$, there is no danger of overlapping, so $E_1=0$. When $K>1$, the term $N_k\sum_{\ell=1}^{k-1}[\frac{4\Delta}{\phi'_{\ell,1}(t)}+E^{(\ell)}(N_k)]$ is the upper bound of all possible overlaps between the $k$-th component and all $\ell$-th component, where $\ell\in\{1,\ldots,k-1\}$. {
The origin of this upper bound is the fundamental Erd\"{o}s-Tur\'{a}n inequality, which gives a quantitative form of Weyl's criterion for equidistribution, and the convergence rate of $E^{(\ell)}(N_k)\to 0$ when $N_k\to\infty$ depends on the algebraic nature of the ratio $\phi_{k,1}'(t)/\phi_{\ell,1}'(t)$. Note that even when the IF's of all oscillatory components are constant, if $K>1$, the $E_1$ term still exists due to the fundamental equidistribution property.

\item When $K>1$, the bound for $E_1$ is the worst bound. Since we could not control the locations of the overlaps between those multiples of different ANH components in the STFT, when we evaluate the STCT by the Fourier transform, the discrepancy caused by the overlaps, denoted as $\delta_3$ in (\ref{Proof:Lemma:Definition:delta3part2}), is bounded simply by the Riemann-Lebesgue theorem. The bound is shown in (\ref{Proof:Lemma:Definition:delta3Fourier}). See Remark \ref{Remark:BoundE1} for more discussions.
The constant could be improved. However, since the focus here is showing how the result is influenced by the fundamental limitation of the number of overlapped multiples, no effort has been made to optimize it. 

\item Note that the bound of $E_1$ blows up when $N_k\to \infty$. Thus, the bound of $E_1$ is not useful when $N_k$ is ``huge''. In practice, however, most non-sinusoidal oscillatory signals have $N_k$ less than 20. The most extreme case we have encountered up to now is the ECG signal, which has $N_k$ about 40. Thus, in practice, we could choose a small $\Delta$ so that $E_1$ is well controlled for a ``reasonable'' $N_k$, and this is the condition ``$\Delta N_k$ is sufficiently small'' in Corollary \ref{Corollary:TimeVaryingCepstrum}. 
However, $\Delta$ cannot be chosen arbitrarily small. Note that the smaller the $\Delta$ is, the longer the window will be, and the larger the absolute moments $I_k$ will be. Thus, the smaller the $\Delta$ is, the worse the bound of $E_2$ is. 
In sum, when $K>1$, except for special non-sinusoidal oscillations with huge $N_k$, the bound for $E_1$ could be well controlled for practical applications. }

\item The {term} $E_2$ comes from the non-constant AM and IF of each ANH component. When the IF and AM are constant, this term becomes zero. Note that when $\gamma$ is chosen small, $b_k$ becomes more like a constant sequence and $\hat{b}_k(q)$ behaves more like a Dirichlet kernel. On the other hand, $E_2$ becomes large when $\gamma$ is small. 
\end{enumerate}

{The theorem and the corollary say} that the STCT encodes the IF information in the format of IP via a periodic function.
To better understand periodic functions $\hat{b}_k$, we take a look at the following example.

\begin{example}\label{Example:WithNontrivialShape}
Consider a signal $f(t)=s(\xi_0 t)$, where $\xi_0>0$ and $s$ is real, smooth and 1-periodic. This special case has only one oscillatory component, $K=1$, with the fixed wave-shape function and a constant IF. Thus we do not worry about the error terms $E_1$ and $E_2$ in Theorem \ref{Theorem:TimeVaryingCepstrum}. By a direct expansion, $f(t)=\sum_{k=0}^Nc(k)\cos(2\pi k\xi_0t+\alpha_k)$, where $N\in\NN\cup \{\infty\}$, $\alpha_0=0$, $c(1)>0$ and $\alpha_k\in[0,2\pi)$ and $c(k)\geq0$ for all $k\neq 1$. To simplify the calculation, we choose a smooth window function $h$ so that $\hat{h}$ is supported on $[-\Delta,\Delta]$, where $0<\Delta<\xi_0/2$. By the Plancherel identity and a direct calculation, we have
\begin{equation}
V_f^{(h)}(t,\xi)=\sum_{k=-N}^Nc(k)\hat{h}(\xi-k\xi_0)e^{i(2\pi k\xi_0 t+\alpha_k)},
\end{equation}
where we denote $c(-k)=c(k)$ for all $k\in\NN$. Since $\Delta<\xi_0/2$, for $0<\gamma\ll 1$, we have
\begin{equation}
|V_f^{(h)}(t,\xi)|^\gamma=\sum_{k=-N}^Nc(k)^\gamma|\hat{h}(\xi-k\xi_0)|^\gamma=\big[\sum_{k=-N}^Nc(k)^\gamma\delta_{k\xi_0}\star |\hat{h}|^\gamma\big](\xi).
\end{equation}
The evaluation of the Fourier transform of $[\sum_{k=-N}^Nc(k)^\gamma\delta_{k\xi_0}\star |\hat{h}|^\gamma](\xi)$ is straightforward and we have for $q>0$,
\begin{align}
C_f^{(h,\gamma)}(t,q)\,= \mathcal{F}[|V_f^{(h)}(t,\cdot)|^\gamma](q)=\widehat{|\hat{h}|^\gamma}(q)\sum_{k=-N}^Nc(k)^\gamma e^{-i2\pi k\xi_0q}\,=\widehat{|\hat{h}|^\gamma}(q)S^{{(\gamma)}}_{1/\xi_0}(q),
\end{align}
where $S^{{(\gamma)}}_{1/\xi_0}(q)$ is a periodic distribution with the period of length $1/\xi_0$ so that $\widehat{S^{{(\gamma)}}_{1/\xi_0}}(k)=c(k)^\gamma$ for $k\in\{-N,-N+1,\ldots,N-1,N\}$ and $\widehat{S^{{(\gamma)}}_{1/\xi_0}}(k)=0$ otherwise.

We could take a look at a special case to have a better picture of what we get eventually. Suppose $N$ is finite and $c(k)=1$ for $k\in\{-N,\ldots,N\}$. In this case, we have
\begin{equation*}
C_f^{(h,\gamma)}(t,q)\,=\widehat{|\hat{h}|^\gamma}(q)D_N(\xi_0q),
\end{equation*}
where $D_N(\xi_0q)$ is the Dirichlet kernel, which is periodic with the period $1/\xi_0$ since $D_N(\xi_0q)=\frac{\sin(\pi(2N+1)\xi_0q)}{\sin(\pi\xi_0q)}$. Also, it becomes more and more spiky at $\ell/\xi_0$ and eventually the Delta comb supported on $\ell/\xi_0$, $\ell\in\ZZ$, when $N\to\infty$. On the other hand, when $N$ is finite and small, the STCT could be oscillatory but still contains information we need. For example, when $N=1$, $D_1(\xi_0q)=\frac{\sin(\pi 3\xi_0q)}{\sin(\pi\xi_0q)}$ and $D_1(\xi_0q)$ still has dominant values at $q=\ell/\xi_0$ for $\ell\in\ZZ$ .
\end{example}

\subsection{inverse STCT}\label{Section:deShape}

Based on Theorem \ref{Theorem:TimeVaryingCepstrum} and a careful observation, we see that to determine the fundamental frequency for an ANH signal $f(t)$, a candidate frequency should have the saliency of its multiples in the TF representation $V_f^{(h)}(t,\xi)$, and the associated period and its multiples in the TP representation $C^{(h,\gamma)}_f(t,q)$. In \cite{su2015combining,su2016exploiting}, this observation is summarized as a practical {principle} called the {\em constraint of harmonicity}, which is described as follows: at a specific time $t_0$, a pitch candidate, $\xi_1>0$, is determined to be the true pitch when there exists $M_v, M_u \in \NN$ such that there are
\begin{enumerate}
\item A sequence of ``peaks'' found around $V_f^{(h)}(t_0,\xi_1)$, $V_f^{(h)}(t_0,2\xi_1)$, $\ldots$, $V_f^{(h)}(t_0,M_v \xi_1)$;
\item A sequence of ``peaks'' found around $C_f^{(h,\gamma)}(t_0,q_1)$, $C_f^{(h,\gamma)}(t_0,2q_1)$, $\ldots$, $C_f^{(h,\gamma)}(t_0,M_u q_1)$;
\item $\xi_1=1/q_1$.
\end{enumerate}
The sequence $\{\xi_1,2\xi_1,\ldots,M_v\xi_1\}$ is commonly called {\em harmonic series} associated with multiples of the pitch $\xi_1$.
The constraint of harmonicity {principle} leads to the following consideration. If we ``invert'' the quefrency axis of the TP representation by the operator $\mathcal{I}$,
\begin{equation}
\mathcal{I}:q\mapsto 1/q,
\end{equation}
when $q>0$, then by the relationship that the period is the inverse of the frequency, we could obtain information about the frequency in $C_f^{(h,\gamma)}(t,\mathcal{I}q)$. Note that $\mathcal{I}$ is open from $(0,\infty)$ to $(0,\infty)$ and the differentiation of $\mathcal{I}$ is surjective on $(0,\infty)$, so for a distribution $T$ defined on $(0,\infty)$, we could well-define the composition $T\circ \mathcal{I}$, or the pull-back of $T$ via $\mathcal{I}$ \cite[Theorem 6.1.2]{Hormander:1990}. Since in general $C_f^{(h,\gamma)}(t,\cdot)$ is a tempered distribution, we could consider the following definition to extract the frequency information for $f$:

\begin{defn}
For a function $f\in\mathcal{S}'$, window $h\in\mathcal{S}$ and $\gamma>0$, the {\em inverse short time cepstral transform (iSTCT)} is defined on $\RR\times \RR^+$ as
\begin{equation}
U_f^{(h,\gamma)}(t,\xi):=C_f^{(h,\gamma)}(t,\mathcal{I}\xi),
\end{equation}
where $\xi>0$ and $U_f^{(h,\gamma)}(t,\cdot)$ is in general a distribution.
\end{defn}

The unit of $\xi$ in $U_f^{(h,\gamma)}(t,\xi)$ is Hz or any feasible unit in the frequency domain. We mention that in the special case that $f\in C^\infty\cap L^\infty$, $U^{(h,\gamma)}_f(t, \xi)$ is a well-defined continuous function in the frequency axis. Also, if $C_f^{(h,\gamma)}(t,\cdot)$ is integrable and we want to preserve the integrability, we could weight $C_f^{(h,\gamma)}$ by the Jacobian of $\mathcal{I}$. However, since the integrability is not the main interest here, we do not consider it.
We view $U_f^{(h,\gamma)}(t,\xi)$ as a TF representation determined by a nonlinear transform composed of several transforms.
While this operator looks natural at the first glance, it is actually not stable.
See the following example for the source of the instability.

\begin{example}
{ Let us continue} the discussion of Example \ref{Example:WithNontrivialShape}.
Suppose $N$ is finite and hence $C_f^{(h,\gamma)}(t,q)=\widehat{|\hat{h}|^\gamma}(q)D_N(\xi_0q)$ for $q>0$.
Thus, by inverting the axis by $\xi\mapsto 1/\xi$ when $\xi\neq 0$, the iSTCT becomes $U^{(h,\gamma)}_f(t,\xi)= \widehat{|\hat{h}|^\gamma}(1/\xi) D_N(\xi_0/\xi)$, where $\xi>0$.
Clearly, due to the oscillatory nature of the Dirichlet kernel, the non-zero region of $D_N(\xi_0q)$ around $q=0$ would be flipped to the high frequency region, which amplifies the unwanted information in the low frequency and represents it in the high frequency region. To be more precise, since $D_N(\xi_0q)$ decays monotonically from $2N$ to about $-0.43N$ as $q$ goes from $0$ to $x_1\in (\frac{1}{(2N+1)\xi_0},\frac{2}{(2N+1)\xi_0})$, where $x_1$ is the local extremal point, in iSTCT, $U_f^{(h,\gamma)}(t,\xi)$ increases from about $-0.43N$ to $2N$ as $\eta$ goes from $1/x_1$ to $\infty$. This indicates that $|U_f^{(h,\gamma)}(t,\xi)|>N$ for all $\xi>\Xi$ for some $\Xi>1/x_1$.

\end{example}

Motivated by the above example, in practice, we need to apply a filtering process on the STCT to stabilize the algorithm. Here is the main idea. Since our interest is to capture the IF's of the signal, we have to effectively remove components unrelated to IF's in the STCT. In practice, the irrelevant components lie in the low quefrency region. Therefore, we need to apply a {\em long-pass lifter} on $U^{(h,\gamma)}_f(t,\xi)$, where the {\em lifter} refers to a ``filter'' processed in the cepstral domain, again by inverting the first four letters of ``filter'', to distinguish it from the {\em filter} processed in the spectral domain \cite{bogert1963quefrency,oppenheim2004frequency}. Moreover, since the quefrency is measured in the unit of time, a lifter is identified as a {\em short-pass} or {\em long-pass} one rather than a low-pass or a high-pass one \cite{bogert1963quefrency,oppenheim2004frequency}. In short, a long-pass lifter passes mainly the component of high quefrency (long period) while rejects mainly the component of low quefrency (short period).

\subsection{de-shape STFT}\label{Section:deshapeSTFT}

Take the music signal as an example to examine the iSTCT. The constraint of harmonicity {principle} tells us that while at a fixed time $t$ we could find a harmonic series associated with multiples of the pitch $\xi_0$ in the TF representation $V^{(h)}_f(t,\xi)$, we should find a  sequence of peaks in the TF representation $U^{(h,\gamma)}_f(t,\xi)$, denoted as $\{\xi_1, \xi_1/2, \ldots, \xi_1/M_u\}$ and this sequence is called the {\em sub-harmonic series} associated with the fundamental frequency $\xi_1$ in the literature.
This observation motivates a combination of the STFT and iSTCT to extract the pitch information; that is, we consider the following combination of the TF representation and TP representation via the iSTCT, which we coined the name {\em de-shape STFT}:

\begin{defn}
For a function $f\in\mathcal{S}'$, window $h\in\mathcal{S}$ and $\gamma>0$, the de-shape STFT is defined on $\RR\times \RR^+$ as
\begin{equation}
\label{eq:W}
W^{(h,\gamma)}_f(t, \xi) := V^{(h)}_f(t,\xi)U^{(h,\gamma)}_f(t, \xi),
\end{equation}
where $\xi>0$ is interpreted as frequency.
\end{defn}

In general, since $V^{(h)}_f(t,\xi)$ is a $C^\infty$ function in the frequency axis and $U^{(h,\gamma)}_f(t, \xi)$ is a distribution in the frequency axis, the de-shape STFT is well-defined as a distribution. Again, in the special case that $f\in C^\infty\cap L^\infty$, $W^{(h,\gamma)}_f(t, \xi)$ is a well-defined continuous function in the frequency axis.

The motivation beyond the nomination ``de-shape'' is intuitive -- since the harmonic series associated with multiples of the fundamental frequency $\xi_0$ in $V^{(h)}_f(t,\xi)$ overlaps with the sub-harmonic series associated with multiples of the fundamental frequency $\xi_0$ in $U^{(h,\gamma)}_f(t, \xi)$ only at $\xi_0$, by multiplying $V^{(h)}_f(t,\xi)$ and $U^{(h,\gamma)}_f(t, \xi)$, only the information associated with the pitch is left in the result.
Thus, the influence caused by the non-trivial wave-shape function in the TF representation is removed, and hence we could view the de-shape process as an {\em adaptive and nonlinear filtering technique} for the STFT.
Since $\xi>0$ in $W^{(h)}_f(t, \xi)$ is interpreted as frequency, the de-shape STFT provides a TF representation.

We mention that in the music field, a similar idea called the {\em combined temporal and spectral representations} has been applied to the single pitch detection problem \cite{peeters2006music,emiya2007parametric}. With our notation, the proposed idea of detecting the pitch at time $t$, denoted as $\xi_0(t)$, is simply by $\xi_0 (t)=\argmax_{\xi>0} |W^{(h,\gamma)}_f(t,\xi)|$ \cite{peeters2006music,emiya2007parametric}.
In the last section of \cite{peeters2006music}, the authors showed a figure of polyphonic music and slightly addressed the ``potential'' of this idea in multiple pitch estimation problems. But this idea was not noticed until \cite{su2015combining,su2016exploiting}, which gives an explicit methodology, systematic investigation, and evaluation of using this idea in multiple pitch estimation.

\subsection{Sharpen de-shape STFT by the synchrosqueezing transform -- de-shape SST}\label{Section:Synchrosqueezing}

While the de-shape STFT could alleviate the influence of the wave-shape function, it again suffers from the Heisenberg-Gabor uncertainty principle and tends to be blurred in the TF representation \cite{Flandrin:1999}.
One approach to sharpen a TF representation is by applying the SST, and we propose to combine SST to obtain a sharp TF representation without the influence of the wave-shape function. SST is a nonlinear TF analysis technique, which is special case of the more general RM method \cite{Auger_Flandrin:1995}. In summary, it aims at moving the spectral-leakage terms caused by Heisenberg-Gabor uncertainty principle to the correct location, and therefore sharpens the TF representation with high concentration \cite{Daubechies_Lu_Wu:2011,Chen_Cheng_Wu:2014,Thakur:2015,Oberlin_Meignen_Perrier:2015}. The main step in SST is estimating the {\em frequency reassignment vector}, which guides how the TF representation should be nonlinearly deformed. The resulting TF representation has been applied to several fields. For example, in the physiological signal processing, SST leads to a a better estimation of IF and AM, which is applied to study sleep dynamics \cite{Wu_Talmon_Lo:2015}, coupling \cite{Iatsenko_Bernjak_Stankovski_Shiogai_Owen_Clarkson_McClintock_Stefanovska:2013} and others, or a better spectral analysis, which is applied to study the noxious stimulation problem \cite{Lin_Wu:2016};
in the mechanical engineering, it has been applied to estimate speed of rotating machinery \cite{Xi_Cao_Chen_Zhang_Jin:2015} and others;
in finance, it is applied to detect the non-stationary dynamics in the financial system \cite{Guharay_Thakur_Goodman_Rosen_Houser:2013};
in the music processing, such an approach can better discriminate closely-located components, and applications have be found in chord recognition \cite{khadkevich2011time}, sinusoidal synthesis \cite{peeters1999sinola} and others.

The frequency reassignment vector associated with a function $f\in\mathcal{S}'$ is determined by
\begin{equation}
\Omega^{(h,\upsilon)}_f(t,\xi):=
\left\{
\begin{array}{ll}
-\Im\frac{V_f^{(\mathcal{D}h)}(t,\xi)}{2\pi V_f^{(h)}(t,\xi)}&\mbox{ when }|V_f^{(h)}(t,\xi)|> \upsilon\\
-\infty&\mbox{ when }|V_f^{(h)}(t,\xi)|\leq \upsilon
\end{array}
\right.
\,,\label{RM:omega}
\end{equation}
where $\mathcal{D}h(t)$ is the derivative of the chosen window function $h\in\mathcal{S}$, $\Im$ means the imaginary part and $\upsilon>0$ gives a threshold so as to avoid instability in computation when $|V^{(h)}_f(t,\xi)|$ is small. The theoretical analysis of the frequency reassignment vector has been studied in several papers \cite{Daubechies_Lu_Wu:2011,Wu:2013,Chen_Cheng_Wu:2014}, and we refer the reader with interest to these papers. {In general, we could consider variations of the reassignment vectors for different purposes. For example, the reassignment vectors used in the second order SST \cite{Oberlin_Meignen_Perrier:2015}. To keep the discussion simple, we focus on the original SST.}

The SST of $V^{(h)}_f(t,\xi)$ is therefore defined as
\begin{equation}\label{definition:SSTV}
SV^{(h,\upsilon)}_{f}(t,\xi)=\int_{\mathfrak{N}_\upsilon(t)} V^{(h)}_f(t,\eta) \frac{1}{\alpha}g\left(\frac{|\xi-\Omega^{(h,\upsilon)}_f(t,\eta)|}{\alpha}\right)\ud \eta.
\end{equation}
where $\xi\geq0$, $\alpha>0$, $g\in\mathcal{S}$ and $\frac{1}{\alpha}g(\frac{\cdot}{\alpha})$ converges weakly to the Dirac measure supported at $0$ when $\alpha\to 0$, $\mathfrak{N}_\upsilon(t):=\{\xi\geq 0|\, |V^{(h)}_f(t,\xi)|>\upsilon\}$; similarly, we have the {\em de-shape SST} defined as
\begin{equation}\label{definition:SSTW}
SW^{(h,\gamma,\upsilon)}_{f}(t,\xi)=\int_{\mathfrak{N}_\upsilon(t)} W^{(h,\gamma)}_f(t,\eta) \frac{1}{\alpha}g\left(\frac{|\xi-\Omega^{(h,\upsilon)}_f(t,\eta)|}{\alpha}\right)\ud \eta\,,
\end{equation}
where $\xi\geq0$. Numerically, $g$ could be chosen to be the Gaussian function with $\alpha>0$ or as a direct discretization of the Dirac measure when $\alpha\ll 1$. For numerical implementation details and the stability results of SST, we refer the reader with interest to \cite{Chen_Cheng_Wu:2014}.

With the de-shape STFT, the wave-shape information is decoupled from the IF and AM in the TF representation; with the de-shape SST, the TF representation is further sharpened. We could continue to do the analysis to, for example, carry out the wave-shape reconstruction, count the oscillatory components, etc. Furthermore, we could combine the de-shape SST information and current wave-shape analysis algorithms, including the functional regression \cite[Section 4.7]{Chui_Lin_Wu:2015}, designing a dictionary \cite{Hou_Shi:2016} or unwrapping the phase \cite{Yang:2014}, to study the oscillatory signal with time-varying wave-shape function. The work of estimating the time-varying wave-shape function with applications will be explored systematically in {a} coming work.

\section{Numerical Results}\label{Section:Numerics}

In this section we demonstrate how the de-shape SST performs in various kinds of signals with multiple ANH components with non-trivial time-varying shape function. We consider a wide range of physiological, biological, audio and mechanical signals, which are generated in different dynamical system and recorded by different sensors. The signals are: (1) abdominal fetal ECG signal, (2) different photoplethysmography signals under different challenges -- respiratory and heartbeat, motion and heartbeat, and non-contact PPG signal, (3) music and bioacoustic signals including the violin sonata, choir and wolves sound. The code of SST and de-shape SST and test datasets are available via request.

For a fair comparison, the parameters for computing the de-shape SST are set to be the same for all signals throughout the paper: $\gamma=0.3$ for the STCT and $\upsilon=10^{-4}$\% of the root mean square energy of the signal under analysis for the de-shape SST.

\subsection{Simulated signal}

We continue the example shown in the Introduction section, make clear how $f_{2}$ is generated, and consider a more complicated example.
Take $W$ to be the standard Brownian motion defined on $\RR$ and define random processes $\Phi_{A,\sigma,a}:=\frac{(|W|+1)\star K_\sigma}{\|(|W|+1)\star K_\sigma\|_{L^\infty[0,100]}}+a$ and $\Phi_{\phi,\sigma,b,c}:=\frac{b|W|}{\|W\|_{L^\infty[0,100]}}\star K_{\sigma}+c$, where $a,b,c\in \RR$, $K_{\sigma}$ is the Gaussian function with the standard deviation $\sigma>0$. $A_2(t)$ is a realization of $\Phi_{A,10,0.9}$, $A_3(t)$ is a realization of $\Phi_{A,10,0.9}$, $\phi_2(t)$ is a realization of $\Phi_{\phi,5,2,\pi/2}$, and $A_3$ is a realization of $\Phi_{\phi,5,1,4}$ on $[0,100]$. Here all realizations are independent.

The signal $f_2$ is generated by $A_2(t)\text{mod}(\phi_2,1)$. To generate $f_3$, denote $t_{k}=\phi_3^{-1}(k)$. The signal $f_3$ is $A_3(t)\chi_{[30,100]}(\sum_{k}\delta_{t_k}\star \chi_{[-3/100,3/100]})$.

Consider a clean signal $f(t)=f_1(t)+f_2(t)+f_3(t)$ from $t=0$ to $t=100$ sampled at 100Hz. Clearly, while $f_1$, $f_2$ and $f_3$ are oscillatory, the wave-shape functions are all non-trivial and the wave-shape functions of $f_2$ and $f_3$ are time-varying, and $f_2$ and $f_3$ exist for only part of the full time observation time. To further challenge the algorithm, we add a white noise $\xi(t)$ to $g$ by considering $Y(t)=f(t)+\xi(t)$, where for all $t$, $\xi(t)$ is a student t4 random variable with the standard deviation $0.5$. The signal-to-noise (SNR) ratio of $Y(t)$ is 1.8dB, where SNR is defined as $20\log\frac{\text{std}(f)}{\text{std}(\xi)}$ and std means the standard deviation.
The signal $f_1$ and $f_2$ are shown in Figure \ref{fig:Introduction:Example1sig}, and the signal of $f_3(t)$, $g(t)$ and $Y(t)$ are shown in Figure \ref{fig:simulation}. The results of STCT, iSTCT, de-shape STFT and de-shape SST of $g(t)$ and $Y(t)$ are shown in Figure \ref{fig:simulationResult}. The ground truths are superimposed for the comparison.

There are several findings.
Note that even when the signal is clean, we could see several interferences in either STFT or SST. For example, we could see the ``bubbling pattern'' in these TF representations around 2Hz from 0 to 60 seconds (indicated by red arrows), which comes from the interference of the 2nd-multiple of $f_1$ and the fundamental component of $f_2$. These interferences are eliminated in de-shape SST, since the wave-shape is ``decomposed'' in the analysis. Second, when the signal is clean, we could see a ``curve'' starting from about 3.4Hz at 0 second and climbing up to 4Hz at 40s in STFT and SST (indicated by green arrows). Certainly this is not a true component but an artifact, which comes from the incidental appearance of different multiples of different ANH functions. This might mislead us and conclude that there is an extra component. Note that this possible artifact is eliminated in the de-shape SST.
Third, around 85s, the IF's of $f_2$ and $f_3$ cross over (indicated by blue arrows). How to directly decouple signals with this kind of cross-over IF's with TF analysis technique is still an open question.
Last but not the least, while the SNR is low, the de-shape SST could still be able to provide a reasonable IF information regarding the components.
This comes from the robustness of the frequency reassignment vector, which has been discussed in \cite{Chen_Cheng_Wu:2014}.
We mention that we could further stabilize the TF representation determined by the de-shape SST by the currently proposed multi-taper technique called concentration of frequency and time (ConceFT) \cite{Daubechies_Wang_Wu:2016}. We refer the reader with interest to \cite{Daubechies_Wang_Wu:2016} for a detailed discussion of ConceFT.

\begin{figure}
\centering
\includegraphics[width=0.9\columnwidth]{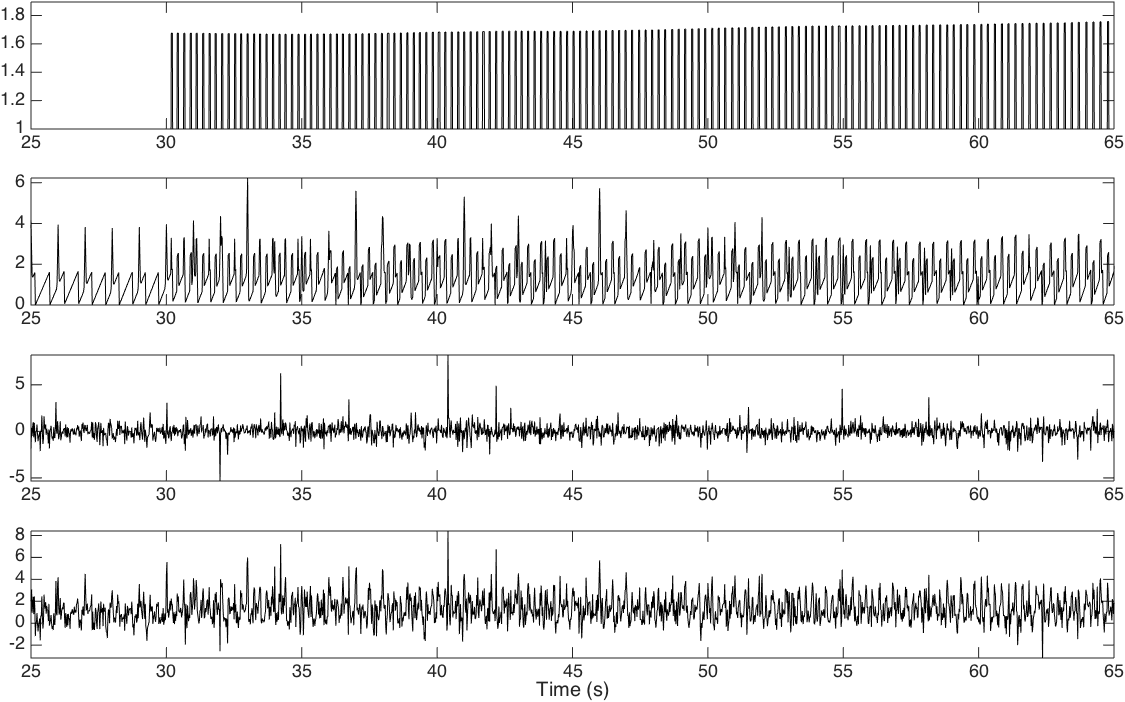}\\
\caption{Top: the simulated signal $f_3$; top middle: $f=f_1+f_2+f_3$; top bottom: $\xi(t)$, bottom: $Y=f(t)+\xi(t)$. To enhance visibility, we only show the signal over $[25,65]$.}
\label{fig:simulation}
\end{figure}

\begin{figure}
\centering
\includegraphics[width=\columnwidth]{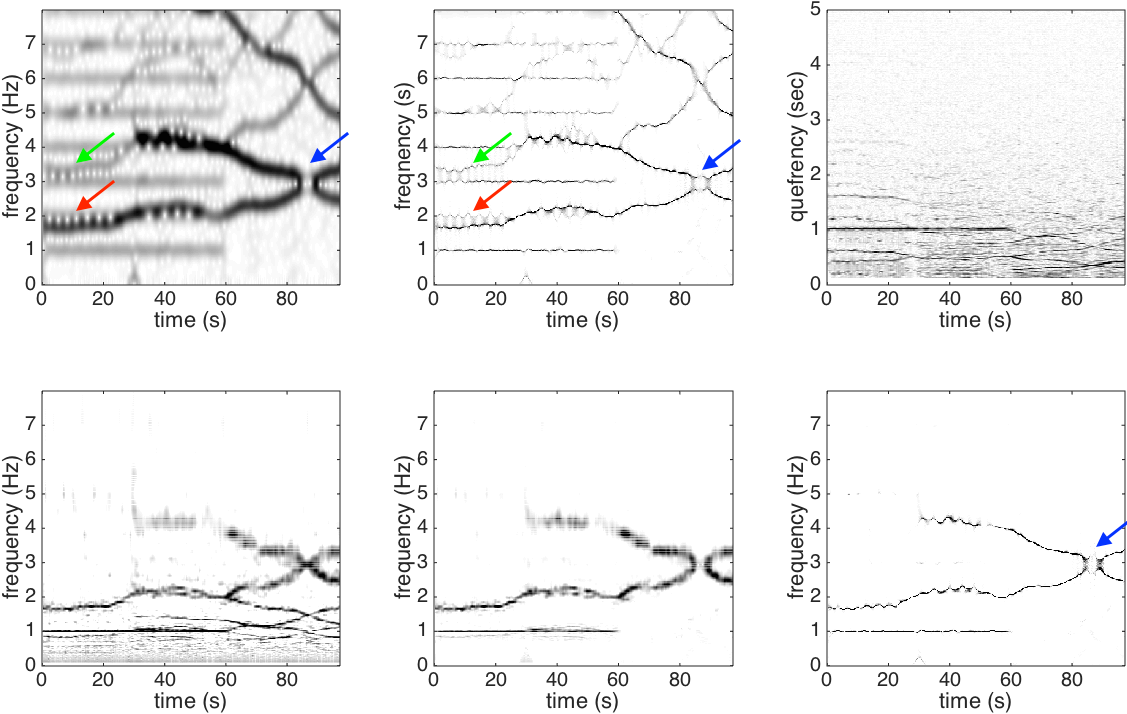}\\
\includegraphics[width=\columnwidth]{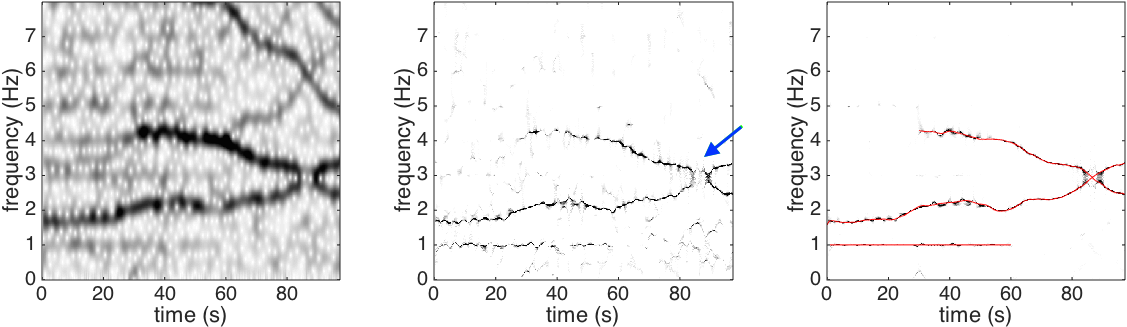}\\
\caption{Top left: the STFT of the clean signal $f$, $|V^{(h)}_f(t,\xi)|$; top middle: the SST of $f$, $|SV^{(h)}_f(t,\xi)|$. The colored arrows indicates three findings mentioned in the main context; top right: the STCT of $f$, $|C^{(h)}_f(t,\xi)|$; middle left: the inverse STCT of $f$, $|U^{(h,\gamma)}_f(t,\xi)|$; middle middle: de-shape STFT of $f$, $|W^{(h,\gamma)}_f(t,\xi)|$; middle right: the de-shape SST of $f$, $|SW^{(h,\gamma)}_f(t,\xi)|$;
bottom left: the STFT of the noisy signal $Y$, $|V^{(h)}_Y(t,\xi)|$;
bottom middle: the de-shape SST of the noisy signal $Y$, $|SW^{(h,\gamma)}_Y(t,\xi)|$; bottom right: the de-shape SST of the clean signal $f$, $|SW^{(h,\gamma)}_f(t,\xi)|$, superimposed with $\phi'_1(t)$, $\phi'_2(t)$ and $\phi'_3(t)$ in red.}
\label{fig:simulationResult}
\end{figure}

\subsection{ECG signal}

As discussed in Section \ref{Section:ECGdiscussion}, we need the modified wave-shape function to better capture the features in the ECG signal. We now show that by the de-shape SST, we could obtain a TF representation without the influence of the time-varying wave-shape function. For the ECG signal, we follow the standard median filter technique to remove the baseline wandering \cite{Clifford_Azuaje_McSharry:2006}, and the sliding window is chosen to be 0.1 second.

\subsubsection{Normal ECG signal}

The lead II ECG signal $f(t)$ is recorded from a normal subject for 85 seconds, which is sampled at 1000Hz. The average heart rate of the subject is about 70 times per minute; that is, the IF is about 1.2 Hz. By reading Figure \ref{fig:ECG}, it is clear that the ECG signal is oscillatory with a non-trivial wave-shape function, and the wave-shape function is time varying, as is discussed in Section \ref{Section:ECGdiscussion}.

Figure \ref{fig:ECG} shows the analysis result. We could see a dominant curve in the STCT, which shows the period information of the oscillation and it is about 0.9 second per wave. The iSTCT flips the period information back to frequency information, and hence we see a dominant curve around 1.2 Hz. Eventually, the multiples associated with the ECG wave-form are well eliminated by the de-shape STFT $SW^{(h)}_f$ and the TF representation is sharp.
Thus we conclude that the de-shape SST provides a more faithful TF representation and decouples the IF, AM and the wave-shape function information. Moreover, the dominant curve around 1.2 Hz fits the ground truth instantaneous heart rate (IHR), which indicates the potential of the de-shape SST in the ECG signal analysis.

\begin{figure}
\centering
\includegraphics[width=0.9\columnwidth]{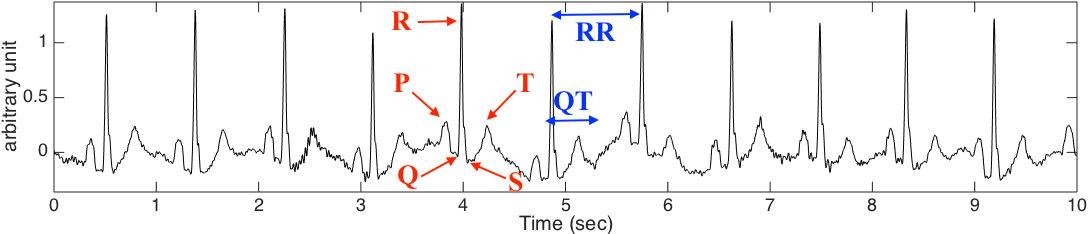}\\
\includegraphics[width=\columnwidth]{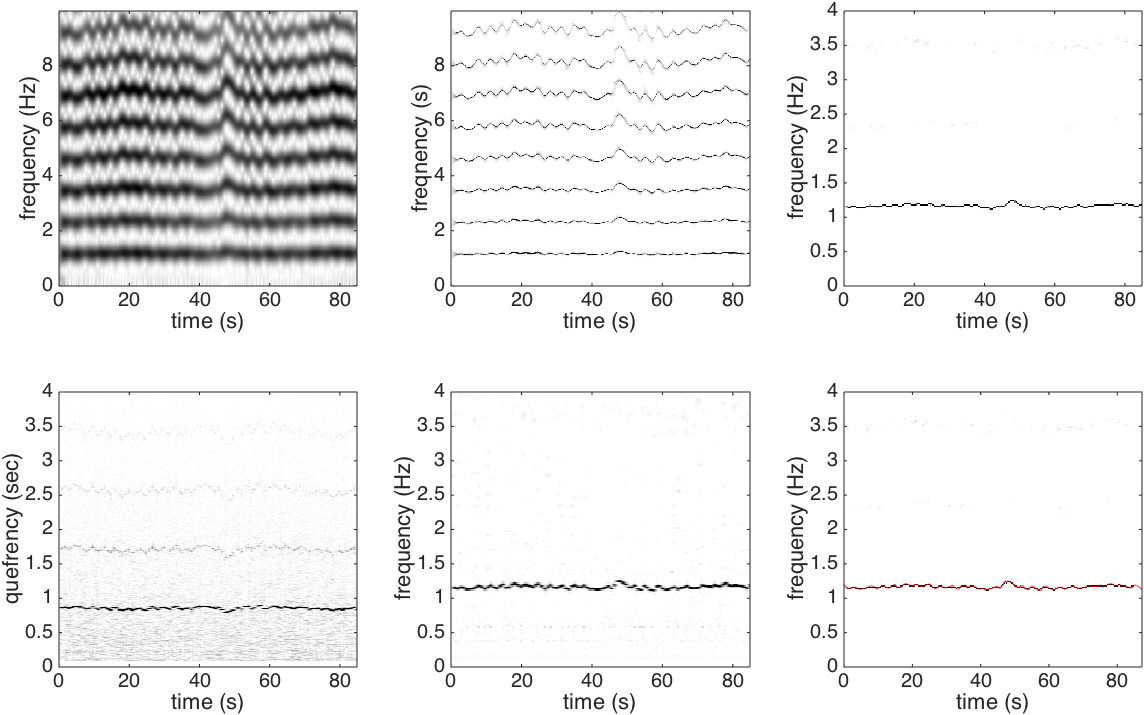}\\
\caption{Top: the ECG signal $f$ recorded for 85 seconds from a normal subject. 
{The basic landmarks of the ECG signal, P, Q, R, S, and T,  and the QT and RR intervals are shown. Note that the QT interval (respectively RR interval) is the length of the time interval between the start of the Q wave and the end of the T wave of one heart beat (respectively two R landmarks of two consecutive heart beats).}
To enhance the visibility, we only show the first {10} seconds. Second row, left panel: $|V^{(h)}_f(t,\xi)|$; middle panel: $|SV^{(h)}_f(t,\xi)|$; right panel: $|SW^{(h,\gamma)}_f(t,\xi)|$. Third row, left panel: $|C^{(h,\gamma)}_f(t,\xi)|$; middle panel: $|U^{(h,\gamma)}_f(t,\xi)|$; right panel: $|SW^{(h,\gamma)}_f(t,\xi)|$ superimposed with the instantaneous heart rate. To enhance the visibility, we show $|C^{(h,\gamma)}_f(t,\xi)|$, $|U^{(h,\gamma)}_f(t,\xi)|$ and $|SW^{(h,\gamma)}_f(t,\xi)|$ only up to 4Hz in the frequency axis.}
\label{fig:ECG}
\end{figure}

\subsubsection{Abdominal fetal ECG}

The fetal ECG could provide critical information for physicians to make clinical decision. While several methods are available to obtain the fetal ECG, the abdominal fetal ECG signal is probably the most convenient and cheap one. We take the abdominal fetal ECG signal with the annotation provided by a group of cardiologists from PhysioNet \cite{Goldberger_Amaral_Glass_Hausdorff_Ivanov_Mark_Mietus_Moody_Peng_Stanley:2000}. In this database, four electrodes are placed around the navel, a reference electrode is placed above the pubic symphysis and a common mode reference electrode is placed on the left leg, which leads to four channels of abdominal ECG signal. The signal is recorded at 1000Hz for 300 seconds. In this example we show the result with the third abdominal ECG signal. Note that while the signal is carefully collected, the signal to noise ratio of the abdominal fetal ECG is relatively low.
We refer the reader with interest to \url{https://www.physionet.org/physiobank/database/adfecgdb/} for more details.

The results of different TF analyses, including de-shape SST, are shown in Figure \ref{fig:fetalECG}. In the STFT and SST, we could see a light curve around 2Hz, which coincides with the fetal IHR we have interest {in}. However, this information is masked by the multiples of the maternal ECG signal. In the de-shape SST, the wave-shape influence is removed and the fetal IHR is better extracted, and the estimated fetal IHR coincides well with the annotation provided by the physician. The curve around 1.5Hz is the IHR associated with the maternal heart beats. The potential of applying de-shape SST to study fetal ECG will be explored and reported in {future works}.

\begin{figure}
\centering
\includegraphics[width=0.9\columnwidth]{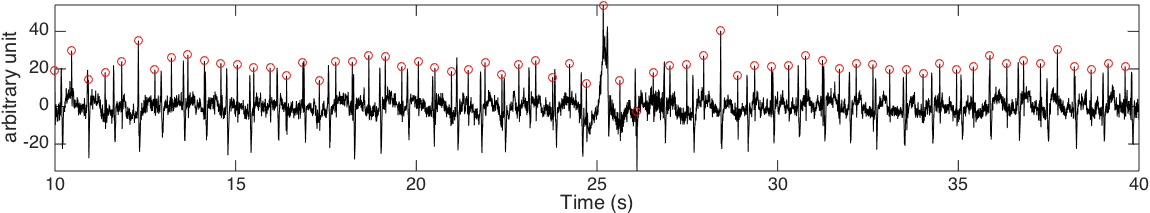}\\
\includegraphics[width=\columnwidth]{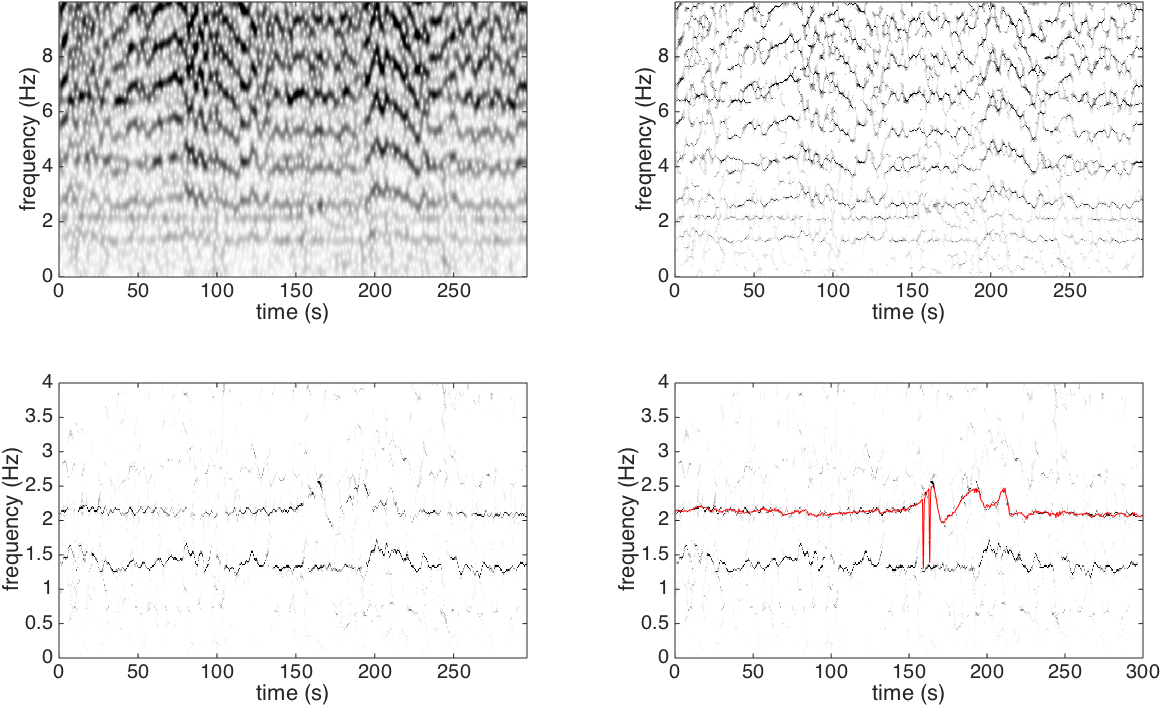}\\
\caption{Top: the abdominal fetal ECG signal $f$ recorded for 300 seconds from a recorded in labor, between 38 to 41 weeks gestation. To enhance the visibility, we only show the signal of 30 seconds long. The annotation of the fetal heart beats are marked in red. Middle left panel: $|V^{(h)}_f(t,\xi)|$; middle right panel: $|SV^{(h)}_f(t,\xi)|$; bottom left panel: $|SW^{(h,\gamma)}_f(t,\xi)|$; bottom right panel: $|SW^{(h,\gamma)}_f(t,\xi)|$ superimposed with the fetal instantaneous heart rate (IHR) determined by a group of cardiologists. {To enhance the visibility, we show $|SW^{(h,\gamma)}_f(t,\xi)|$} only up to 4Hz in the frequency axis. The curve around 1.5Hz is the IHR associated with the maternal heart beats.}
\label{fig:fetalECG}
\end{figure}

\subsection{PPG signal}

Pulse waves represent the hemodynamics, and it can be monitored via plethysmographic technologies in different regions of the body. These technologies often use photo sensors usually placed on the earlobe or finger, by illuminating the tissue and simultaneously measuring the transmitted or the reflected light using a specific wavelength. More recently, noncontact techniques such as video signals (e.g., PhysioCam \cite{Davila:2012Thesis}) have been used to monitor the pulse wave from the face at a distance. Collectively, the application of photosensors to monitor pulse wave are known as photoplethysmography (PPG). See, for example, \cite{Davila:2012Thesis} for a review of the PPG technique. In addition to acquire the hemodynamical information, it also contains the respiration information. Indeed, mechanically, inspiration leads to a reduction in tissue blood volume, which leads to a lower amplitude of the PPG signal. Since none of the pulse wave or the respiration-induce variation oscillates like a sinusoidal wave, the signal should be modeled by the ANH model.

\subsubsection{PPG signal with respiration}

Figure \ref{fig:PPG2comp} shows a PPG signal from the Capnobase dataset\footnote{\url{http://www.capnobase.org}} and its analysis result with the de-shape STFT.  The PPG signal, the capnogram signal and the ECG signal are simulateneously recorded from a subject without any motion at 300 Hz for 480 seconds. By a visual inspection, it is clear that there are two oscillations inside the PPG signal -- the faster (respectively slower) oscillations are associated with the heartbeat (respectively respiration). Clearly, the non-sinusoidal oscillatory waves complicate STFT $V^{(h)}_f$ and SST $SV^{(h)}_f$, while these multiples are elliminated in the de-shape STFT and de-shape SST. Also, we could see that the estimated IHR and instantaneous respiratory rate (IRR) estimated from the PPG signal fit the IHR and IRR derived directly from the ECG signal and the capnogram signal. This indicates the potential of simultaneously obtaining IHR and IRR from the PPG signal.

We mention that when $\gamma$ is chosen to be $2$, the heartbeat component is missed (the result is not shown). This coincides with the general knowledge that $\gamma=2$ is not a good periodicity detector when there exists multiple periodicity in the signal.

\begin{figure}
\centering
\includegraphics[width=0.9\columnwidth]{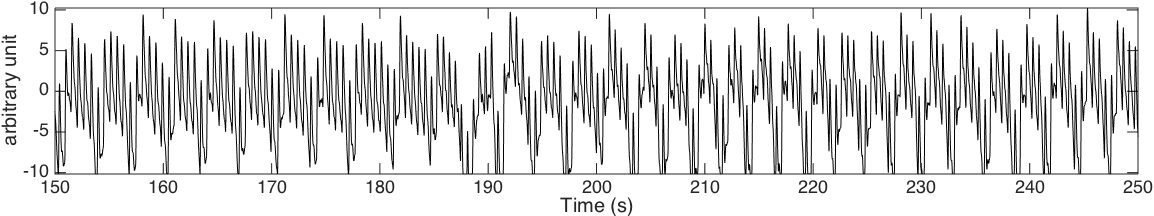}\\
\includegraphics[width=\columnwidth]{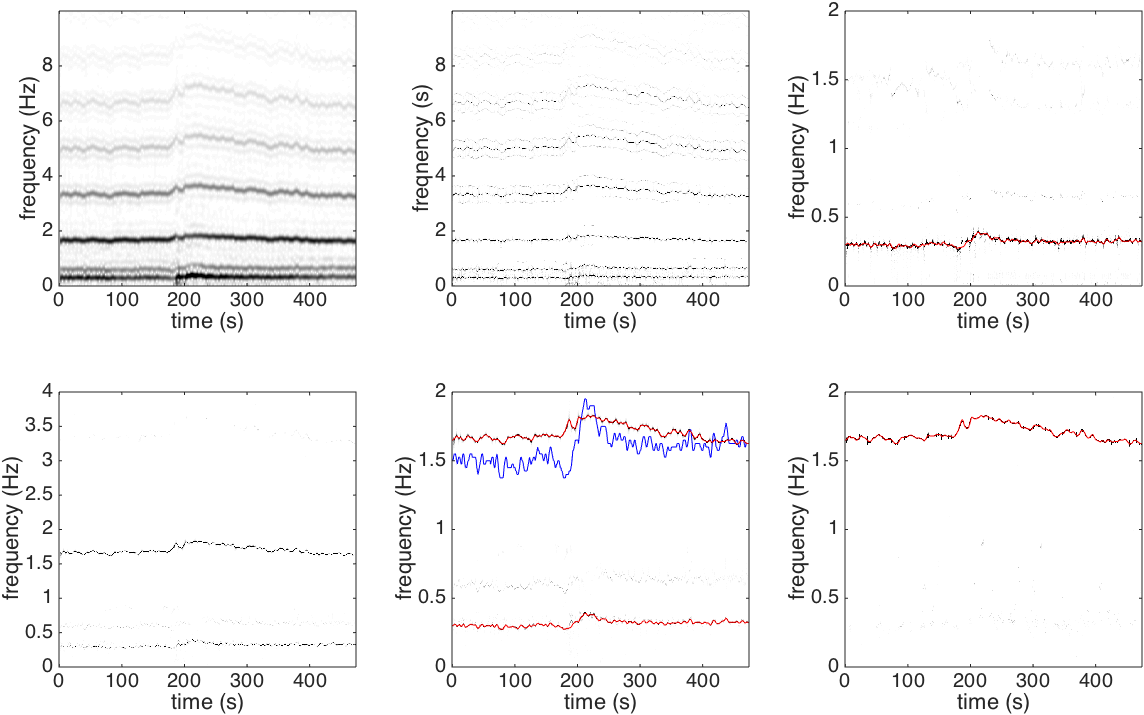}\\
\caption{Top row: the photoplethysmography signal $f$ recorded from a normal subject for 480 seconds. To enhance the visibility, we only show the segment between the 150-th second and 250-th second.
Second row: left: $|V^{(h)}_f(t,\xi)|$; middle: $|SV^{(h)}_f(t,\xi)|$; right: the de-shape SST of the capnogram signal superimposed with the instantaneous respiratory rate (IRR) in red, which is estimated from the PPG signal. Here, only up to 2Hz in the frequency axis is shown to enhance the visibility.
Bottom row: left: $|SW^{(h,\gamma)}_f(t,\xi)|$. To enhance the visibility, we show $|SW^{(h,\gamma)}_f(t,\xi)|$ only up to 4Hz in the frequency axis. Clearly, the multiples of each component are eliminated; middle: $|SW^{(h,\gamma)}_f(t,\xi)|$ superimposed with the estimated instantaneous heart rate (IHR) and IRR. The red curve around 0.3Hz is associated with the IRR and the red curve around 1.6 Hz is associated with the IHR. {The blue curve is five times the IRR curve}, which indicates that the component with the higher frequency is not a multiple of the component with lower frequency; right: the de-shape SST of the electrocardiographic signal superimposed with the IHR in red, which is estimated from the PPG signal. Here, only up to 2Hz in the frequency axis is shown to enhance the visibility.
}
\label{fig:PPG2comp}
\end{figure}

\subsubsection{PPG signal with motion}

Figure \ref{fig:MotionPPG} shows the result of one PPG sample used in the training dataset of ICASSP 2015 signal processing cup\footnote{\url{http://www.zhilinzhang.com/spcup2015/}}. The sample is a 5-minute PPG signal sampled at 125Hz when the subject runs with changing speeds, scheduled as: rest (30s) $\to$ 8km/h (1min) $\to$ 15km/h (1min) $\to$ 8km/h (1min) $\to$ 15km/h (1min) $\to$ rest (30s). {From the recorded signal} it is not easy to see how the motion and heartbeat vary. The heartbeat component starts from around 1.7 Hz at 50 seconds, to 2.2 Hz from 150 to 170 seconds, when the subject has just finished the 15km/h running section. Then, the heartbeat goes lower in the 8km/h section and higher in the final 15km/h section.

Note that the IF of the heartbeats (marked by the red arrow) lies between two other components, supposedly contributed by motion. The higher frequency component associated with motion has IF about twice the IF of the lower one.
We conjecture that the higher one is contributed by the movement of body while the lower is contributed by the movement of arms and legs. The body finishes a period by just one step, while the leg finished a period by two steps (one leg needs to finish a forward and backward movement). This is very similar to the ``octave'' detection problem in music signal processing (see Section \ref{Section:Numerical:ViolinSonata}) and it is quite natural to catch two {components} here as they are indeed (at least) two different oscillatory signals, where the one has IF almost twice from the other one. An extensive study of this signal is needed to fully understand how the body motion influences the physiological signal and will be reported in {a} future work.

\begin{figure}
\centering
\includegraphics[width=0.9\columnwidth]{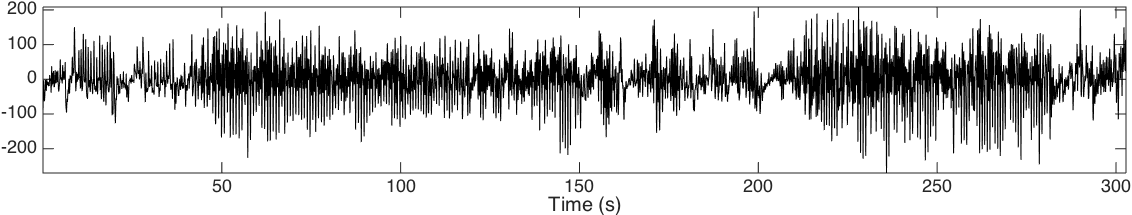}\\
\includegraphics[width=0.9\columnwidth]{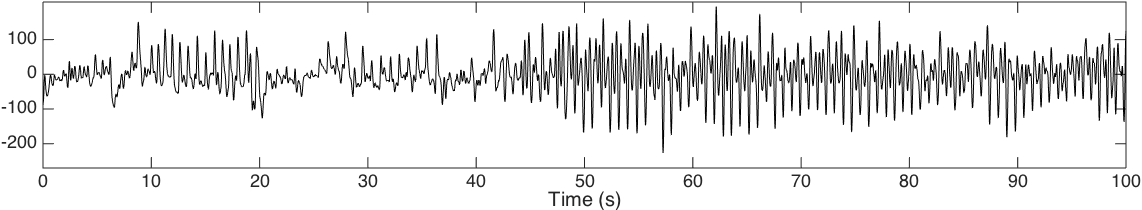}\\
\includegraphics[width=\columnwidth]{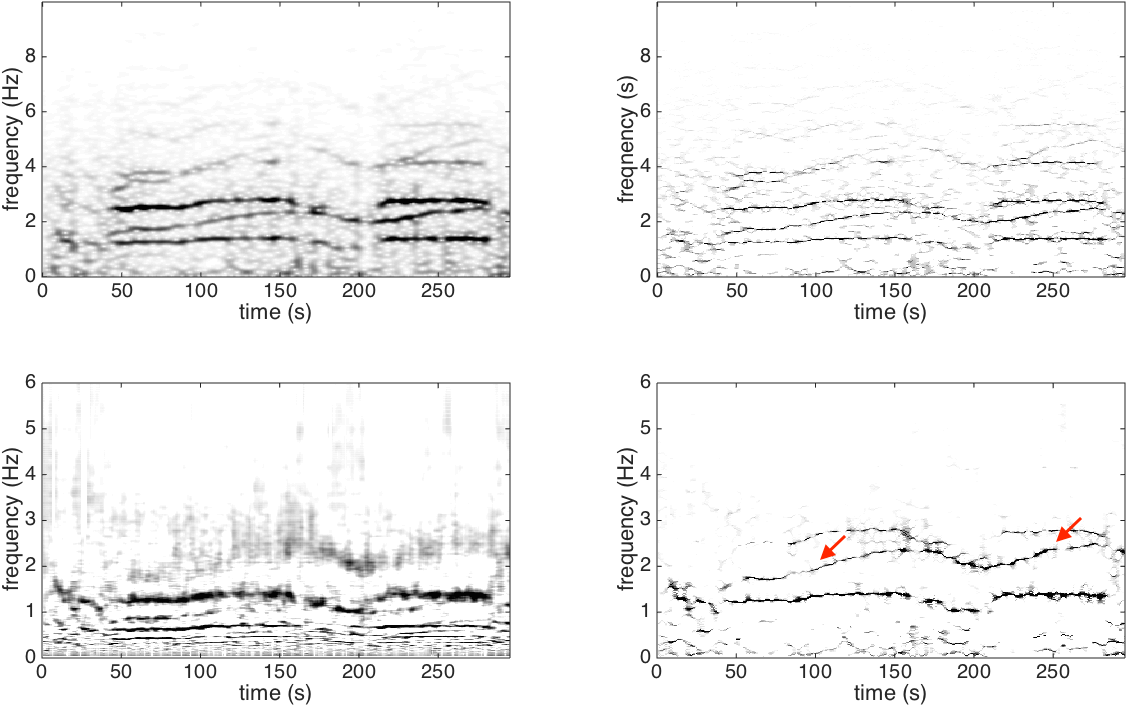}\\
\caption{Top row: the photoplethysmography signal $f$ recorded from a normal subject, who is scheduled to run at different speeds. Second row: the first 100 seconds photoplethysmography signal $f$. It is clear that the signal is composed of several components with complicated dynamics. Third row: $|V^{(h)}_f(t,\xi)|$ is shown on the left and $|SV^{(h)}_f(t,\xi)|$ is shown on the right. Bottom row: $|U^{(h,\gamma)}_f(t,\xi)|$ is shown on the left and $|SW^{(h,\gamma)}_f(t,\xi)|$ is shown on the right. The heartbeat component is marked by red arrows. To enhance the visibility, we shown $|U^{(h,\gamma)}_f(t,\xi)|$ and $|SW^{(h,\gamma)}_f(t,\xi)|$ only up to 6Hz in the frequency axis.}
\label{fig:MotionPPG}
\end{figure}

\subsubsection{Non-contact PPG signal}

Figure \ref{fig:NonContactPPG3} shows the non-contact PPG signal recorded from a normal subject when he is walking on the treadmill at 0.6 Hz. The sampling rate is 100Hz. The non-contact PPG is collected with the PhysioCam technology, and we refer the reader with interest to \cite{Davila:2012Thesis} for details. The ECG signal is simultaneously recorded from the subject at the sampling rate 1000Hz, so we have the true IHR for comparison. Clearly the signal is noisy and contains the walking rhythm; that is, the non-contact PPG signal is composed of two oscillatory signals -- one is associated with the hemodynamics and one is associated with the walking rhythm. Despite the heavy corruption terms in the low frequency, which comes from the ``trend'' inside the signal, we could see that the de-shape STFT successfully extracts the walking rhythm around 0.6 Hz and the IF around 2Hz, which coincides with the IHR determined from the ECG signal. A systematic study of this kind of signal, including the associated de-trend technique, is critical for practical applications and will be reported in {a} future work.

\begin{figure}
\centering
\includegraphics[width=0.9\columnwidth]{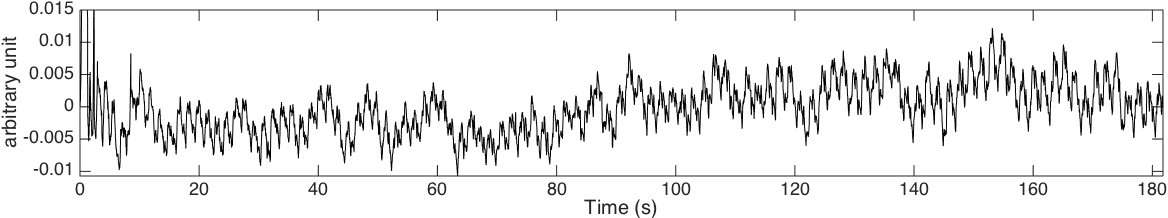}\\
\includegraphics[width=\columnwidth]{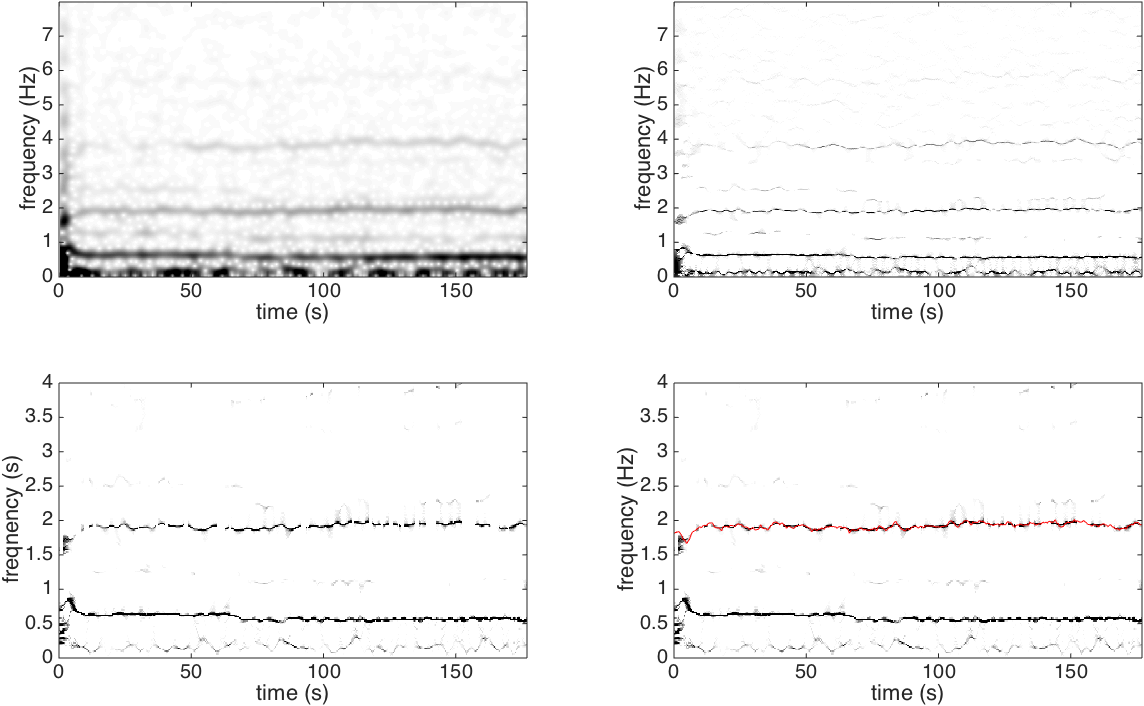}\\
\caption{Top: the non-contact PPG signal $f$ recorded from a normal subject walking on the treadmill at a fixed speed. Middle left: $|V^{(h)}_f(t,\xi)|$; middle right: $|SV^{(h,\gamma)}_f(t,\xi)|$; bottom left: $|SW^{(h,\gamma)}_f(t,\xi)|$; bottom right: $|SW^{(h,\gamma)}_f(t,\xi)|$ superimposed with the instantaneous heart rate. To enhance the visibility, we show $|U^{(h,\gamma)}_f(t,\xi)|$ and $|SW^{(h,\gamma)}_f(t,\xi)|$ only up to 5Hz in the frequency axis.}
\label{fig:NonContactPPG3}
\end{figure}

\subsection{Music and bioacoustic sounds}

The idea of de-shape STFT has been applied in the task called {\em automatic music transcription} (AMT) \cite{peeters2006music,emiya2007parametric,su2015combining,su2016exploiting}, and this approach has been shown competitive in comparison to the state-of-the-art AMT methods in the MIREX-MF0 challenge, an annual competition in the field of music information retrieval (MIR).\footnote{\url{http://www.music-ir.org/mirex/wiki/MIREX_HOME}} AMT is still a technology under active development by now, where one big challenge is how to correctly identify the pitches of the notes played at the same time. In this subsection, we show the potential of applying the de-shape SST to the AMT problem.

\subsubsection{Violin sonata}\label{Section:Numerical:ViolinSonata}

Figure \ref{fig:Mozart} shows a 6-second segment from {\em Mozart's Violin Sonata in E minor, K.304}, where the annotations are provided by musicians. The sampling rate of the signal is 44.1 kHz. This segment contains the sounds of two instruments, violin (melody) and piano (accompaniment). The number of concurrent pitches of this signal at every timestamp varies from 1 to 4, where the violin is played in single pitch and the piano in multiple pitches. The patterns of the two instruments are different, which can be seen from reading the TF representations of STFT and SST. The violin sound exhibits a clear {\em vibrato} (i.e., periodic variation of the IF) together with a strong and frequency-dependent AM effect. See the red arrows in Figure \ref{fig:Mozart} for an example. It is to say that the spectral envelope of the sound varies strongly during one cycle of vibrato \cite{fletcher2010physics}. On the other hand, piano notes have stable IF's, strong attack and long decay of AMs, and, as mentioned in Section \ref{Section:StiffStringModel}, the {\em inharmonicity} makes the high-order harmonic peaks deviate from the integral multiple of the fundamental frequency $f_1$.
The notes of this segment are with pitches ranging from E2 (the fundamental frequency is $82.4$ Hz) to G5 (the fundamental frequency is $784.0$ Hz), and they are shown in the red lines in Figure \ref{fig:Mozart}. The resolution of the labels formatted in Musical Instrument Digital Interface (MIDI) is one semitone.

We indicate one specific tricky problem commonly encountered in this kind of signal. Take the signal from 0.76 to 1.14 seconds as an example. The highest note of piano, B3 (the fundamental frequency is $246.9$ Hz), is just one half of the violin note, B4 (the fundamental frequency is $493.88$ Hz). It is to say, all multiples of violin note are (nearly) overlapped with the piano note, thereby violates the frequency separation condition in Definition \ref{DefBClassMultipleTimeSeries}. The problem of detecting these ``overlaps'' is commonly understood as the {\em octave detection} \cite{su2014resolving}. A systematic study of this specific problem is out of the scope of this paper, and it will be discussed in {a} future work.

From the result of the de-shape SST, we see that the multiples are distinguished from the IF's and are eliminated. All the notes of both violin and piano are well captured. For violin we can even obtain the vibrato rate and vibrato depth of the notes, which are not recorded in the MIDI ground truth. We could also see that the octave problem mentioned above is well resolved.
However, we can still see some false detections in the ``inner part'' of the music. For example, there is a component appearing at around 330 Hz from 1.46 to 1.8 seconds, but there is no note played here. To explain this, notice that the fake component has frequency twice of a piano note while at the same time one half of the violin note. This causes an issue called the {\em stacked harmonics} ambiguity, which is caused by double or even more octave ambiguities. This open problem has also been raised in \cite{su2015combining,su2016exploiting}. Again, a systematic study of this specific problem is out of the scope of this paper, and it will be discussed in future {works}.

\begin{figure}
\centering
\includegraphics[width=0.9\columnwidth]{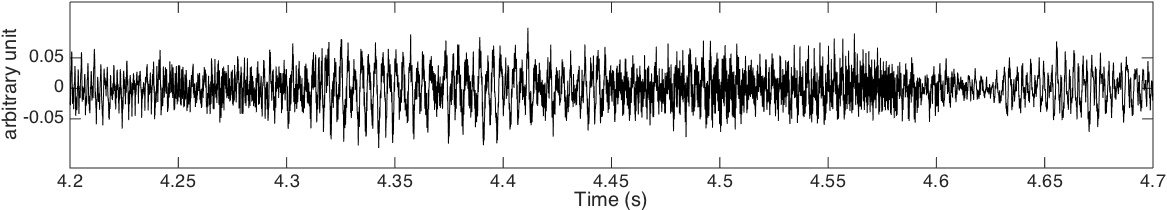}\\
\includegraphics[width=\columnwidth]{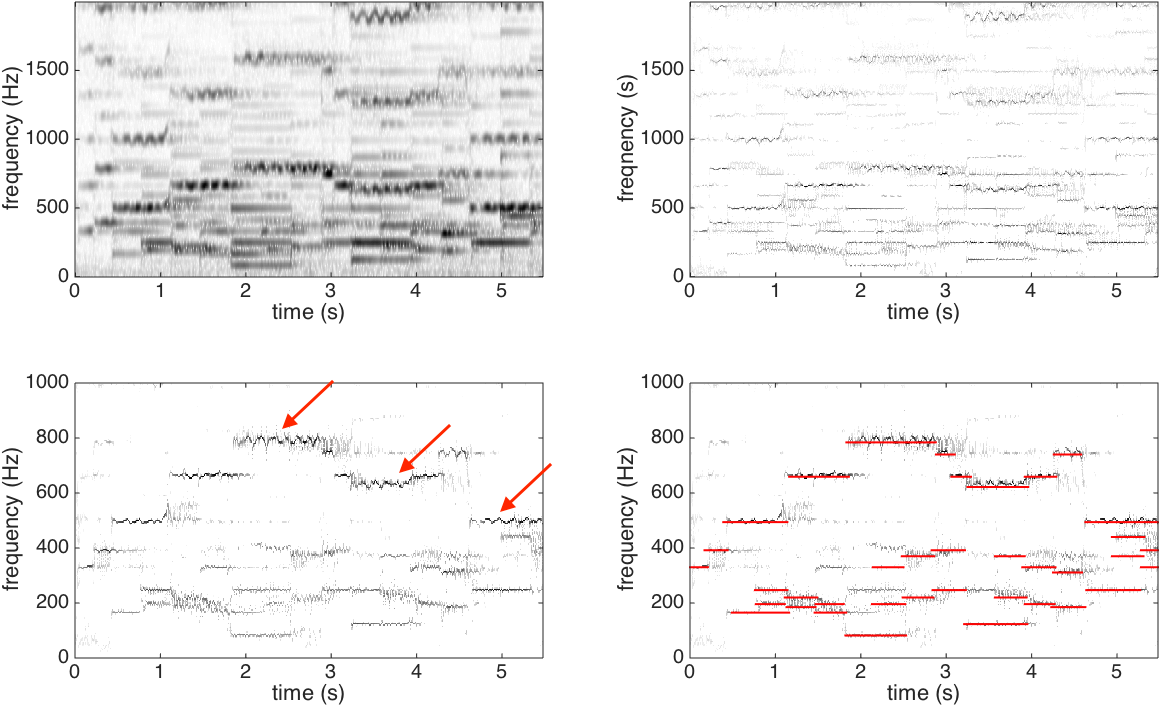}\\
\caption{Top: the Mozart violin sonata signal $f$, which is zoomed in to the period from 4.2 second to 4.7 second to enhance the visualization. Since there are several oscillatory components with complicated wave-shape function, it is not clear what information is hidden inside, even it is hard to identify oscillations. Middle left: $|V^{(h)}_f(t,\xi)|$; middle right: $|SV^{(h)}_f(t,\xi)|$; bottom left: $|SW^{(h,\gamma)}_f(t,\xi)|$, where the red arrows indicate the violin sound with vibrato; bottom right: $|SW^{(h,\gamma)}_f(t,\xi)|$ with the annotations superimposed in red. To enhance the visibility, we show $|SW^{(h,\gamma)}_f(t,\xi)|$ only up to 1000Hz in the frequency axis.}
\label{fig:Mozart}
\end{figure}

\subsubsection{Choir}

Figure \ref{fig:Choir} shows the analysis result of a recorded choir music with the annotation provided by experts. Similar to the above example, the choir music also has multiple components and usually in consonant intervals. Moreover, in the choir music, every perceived individual note is typically sung {\em in unison} by more than one performer. However, since there is always some small and independent variation of the IF among performers, the resulting sound would have wider mainlobe in the STFT than the other music sung by a single performer. Such a phenomenon, called {\em pitch scattering}  \cite{ternstrom1993perceptual}, usually appears in choir and symphony music, as a challenge in correctly estimate the pitch of every note.

This example is a 3-part choir (first soprano, second soprano and alto), with pitches ranging from B3 (the fundamental frequency is $246.9$ Hz) to E5 (the fundamental frequency is $659.3$ Hz). We could see in Figure \ref{fig:Choir} that the pitch scattering issue can be partially addressed by the SST. However, we can still find some intertwined components, like the component at around 920 Hz from 2.2 to 3.5 seconds, which might be contributed by more than two notes with different vibrato behaviors. By using the de-shaped SST, this wide-spread terms are correctly identified as the multiples and removed. All labeled notes are captured and there are few false alarm terms.

Although we have shown the usefulness of de-shape SST in both physiological and musical signals, we need to emphasize some differences between them. In comparison to physiological data, musical signals can have a much larger number of components (e.g., more than 10 components in a symphony), which complicate the patterns of the multiples. Besides, most of the musical works are composed following the {\em theory of harmony}, which holds a principle that a sound is {\em consonant} when the ratio of the IF's are in simple ratios. This implies that the spectra of the components are highly overlapped. Moreover, the octave is very often seen in music composition. Therefore, musical signals usually violate Definition \ref{DefBClassMultipleTimeSeries} and make the problem of AMT ill-posed. To reduce the ambiguities of octaves and other consonant intervals, we may impose more strict constraints when we analyze the signal, like the constraint of harmonicity discussed in Section \ref{Section:deShape}.
For more information of this approach in AMT, readers could refer to \cite{su2015combining,su2016exploiting}.

\begin{figure}
\centering
\includegraphics[width=0.9\columnwidth]{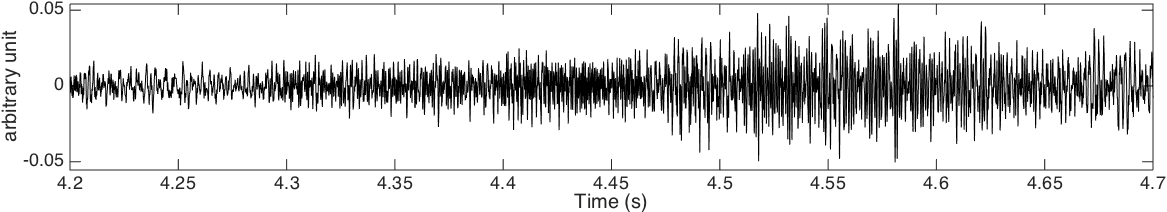}\\
\includegraphics[width=\columnwidth]{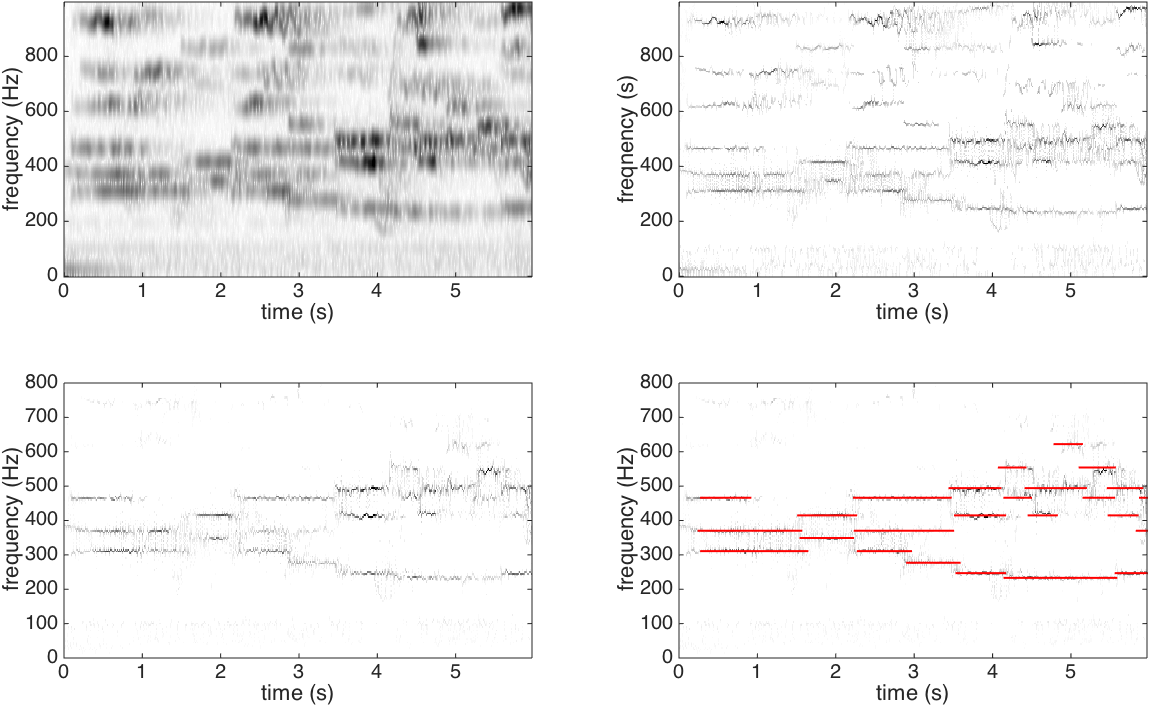}\\
\caption{Top: the choir signal $f$, which is zoomed in to the period from 4.2 second to 4.7 second to enhance the visualization. Middle left: $|V^{(h)}_f(t,\xi)|$; middle right: $|SV^{(h)}_f(t,\xi)|$; bottom left: $|SW^{(h,\gamma)}_f(t,\xi)|$; bottom right: $|SW^{(h,\gamma)}_f(t,\xi)|$ with the annotations superimposed in red. To enhance the visibility, we show $|SW^{(h,\gamma)}_f(t,\xi)|$ only up to 800Hz in the frequency axis.}
\label{fig:Choir}
\end{figure}

\subsubsection{Wolf howling}

An important topic in conservation biology is monitoring the number of wolves in the field \cite{Passilongo_Mattioli_Bassi_Szabo_Apollonio:2015}. Analyzing the wolf howling signal is an efficient approach to evaluate how many wolves are there in the field under survey. In this final example we show the analysis result with a field signal recorded The sound is downloaded from Wolf Park website\footnote{\url{http://www.wolfpark.org/Images/Resources/Howls/Chorus_1.wav}}. The signal is sampled at 11.025 kHz for 40 seconds. In Figure \ref{fig:Wolf} we could directly see that while TF representations provided by STFT and SST are complicated by the multiples caused by the non-trivial wave-shape, the TF representation provided by de-shape SST contains only the fundamental components. By reading the de-shape SST, we could suggest that there are at least three wolves in the field, since during the recording period, there are at most three dominant curves at a fixed time. However, the ground truth for this database is not provided, and identifying each single wolf needs field experts, so this conclusion is not confirmed, and a further collaborative exploration with biologists is needed. To sum up, this suggests that the de-shape SST {has} potential to provide {an} audio visualization for this kind of application.

\begin{figure}
\centering
\includegraphics[width=0.9\columnwidth]{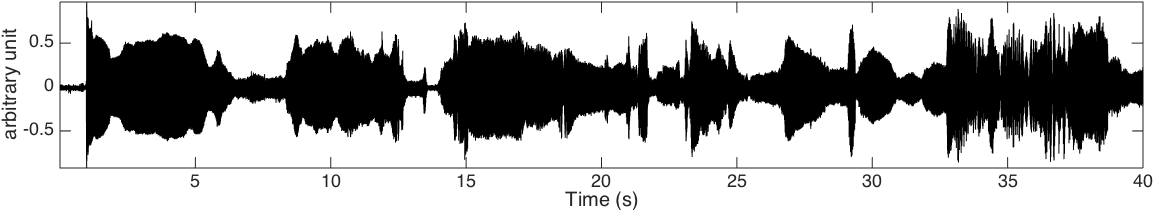}\\
\includegraphics[width=\columnwidth]{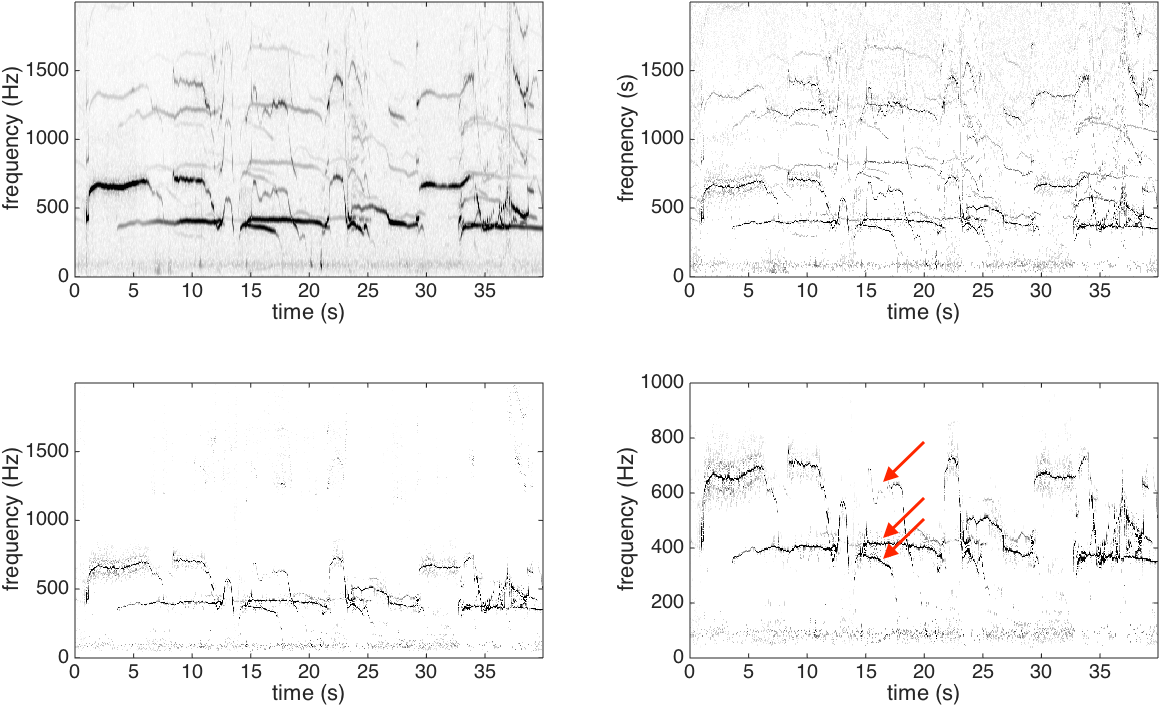}\\
\caption{Top: the wolf howling signal $f$. Since each component inside the signal oscillates at the frequency at least 400Hz, the oscillation could not be visualized except the overall amplitude modulation. Middle left: $|V^{(h)}_f(t,\xi)|$; middle right: $|SV^{(h)}_f(t,\xi)|$; bottom left: $|SW^{(h,\gamma)}_f(t,\xi)|$; bottom right: the zoom-in $|SW^{(h,\gamma)}_f(t,\xi)|$ only up to 1000Hz in the frequency axis. The three red arrows at around 18 seconds indicate that there are at least three wolves.}
\label{fig:Wolf}
\end{figure}

\section{Numerical issues}\label{Section:NumericalIssues}

While the numerical implementation of STCT, iSTCT, de-shape STFT and de-shape SST are straightforward, we should pay an attention to evaluate {iSTCT}. In particular, the map from $C_f^{(h,\gamma)}(t,\eta)$ to $U_f^{(h,\gamma)}(t,\xi)$ depends on the inverse map $\mathcal{I}$, which is numerically unstable. To stabilize it, there are two critical process: (1) long-pass lifter; (2) discretize $\mathcal{I}$ by a suitable weighting, for example, by the Jacobian of $\mathcal{I}$, so that the iSTCT is defined on the uniform frequency grid. Let the sampling frequency of the signal $f(t)$ be $\zeta>0$ and we sample $N\in\NN$ points from $f$. Then, for the $N$-point STFT, the frequency axis is discretized into $\eta_n=n\zeta/N$, where $n = 0,1,\ldots,N-1$, $\eta_n$ is the $n$-th index in the frequency axis, and the frequency resolution is $\Delta\zeta:=\zeta/N$. Similarly, the quefrency axis in STCT is discretized into $q_n=n/\zeta$, where $n=0,1,\ldots, N-1$, $q_n$ is the $n$-th index in the quefrency axis and the quefrency resolution is $\Delta q:=1/\zeta$. We discretize the frequency axis of iSTCT in the same way as that of STFT; that is, the frequency axis of iSTCT is discretized into $\eta_n$, where $n = 0,1,\ldots,N-1$, and the frequency resolution is $\Delta\zeta$.

To implement the long-pass lifter mentioned in Section \ref{Section:deShape}, we consider a simple but effective hard threshold approach by choosing a cutoff quefrency $q_c$, where $c\in\NN$ is chosen by the user; that is, all entries with index less than $c$ are set to zero and the other entries are not changed. While it depends on the characteristic of the signal, in practice we suggest to choose the cutoff quefrency in the range of $10\leq c\leq 20$ and numerically it performs well.

One main issue of the mapping $\mathcal{I}$ is that it maps uniform grid to a non-uniform grid and hence there are insufficient low-quefrency elements in $C_f^{(h,\gamma)}(t,\mathcal{I}\eta)$, which could be directed implemented by inverting the quefrency axis index of $C_f^{(h,\gamma)}(t,\eta)$, to represent the high-frequency content in $U_f^{(h,\gamma)}(t,\xi)$. For example, we have only about $\lfloor 0.1/\Delta q\rfloor=\lfloor 0.1\zeta\rfloor$ entries on the quefrency interval $[0.1,0.2]$ in $C_f^{(h,\gamma)}(t,\eta)$, while we have $\lfloor\frac{5N}{\zeta}\rfloor$ entries on the frequency interval $[5,10]$ in $U_f^{(h,\gamma)}(t,\xi)$. On the other hand, there are too many high-quefrency elements to represent the low-frequency content. Therefore, we suggest to do interpolation over the {quefrency} axis in the STCT to alleviate this issue. Denote the finer grid in the quefrency axis as $\tilde{q}_j$, $j=1,\ldots,M$ and $M>N$. Further, if we want to preserve the integrability of the function after the mapping $\mathcal{I}$, we should weight the entries by the Jacobian of $\mathcal{I}$. To sum up, after obtaining $C_f^{(h,\gamma)}$ with a finer resolution in the {quefrency} axis, the elements in $C_f^{(h,\gamma)}$ are weighted and summed up to the closest frequency bin corresponding to it; that is, we implement iSTCT by
\begin{equation}
U_f^{(h,\gamma)}(t,\eta_n)=\sum_{j \in \mathfrak{P}(\eta_n)}C^{(h,\gamma)}_f\left(t, \tilde{q}_j\right)\tilde{q}_j,
\end{equation}
where $\mathfrak{P}(\eta_n):=\{j: 1/(\eta_n+0.5\Delta\zeta)<\tilde{q}_j\leq 1/(\eta_n-0.5\Delta\zeta)\}$ for each $n=0,1,\ldots,N-1$.

\section{Conclusions}\label{Section:Conclusions}

To handle oscillatory signals in the real world, we provide a model capturing oscillatory features, including IF, AM and time-varying wave-shape function. To alleviate the limitation of TF analysis caused by the existence of non-trivial wave-shape function, we consider the idea of cepstrum and introduce the STCT, de-shape STFT and de-shape SST. A theoretical proof is provided to study how STCT works. When {the STCT and its theoretical proof is} combined with the previous study of SST, we have a theoretical understanding of the efficiency of de-shape SST. In addition to the simulated signal, several real datasets are studied and confirm the potential of the proposed algorithms. The proposed algorithm could be easily combined with several other algorithms to study a given database. For example, we could apply ConceFT \cite{Daubechies_Wang_Wu:2016} to stabilize the influence of the noise, the RM technique \cite{Auger_Flandrin:1995} could be applied to further sharpen the TF representation if causality is not an issue, we could apply the adaptive local iterative filtering \cite{Cicone_Liu_Zhou:2014} to reconstruct each oscillatory component, we could consider the template fitting scheme by designing a good dictionary based on the available information from the de-shape SST \cite{Hou_Shi:2013a}, to name but a few.
{However,} there are several problems left unanswered in this paper. We summarize them below.

To facilitate the discussion, we could call the sequence $\{B_{k,\ell}(t)\}_{\ell=-N_k}^{N_k}$ in (\ref{MainTheorem:STFTExpansion}) the {\em spectral envelope} of the $k$-th ANH model. The assumption in Theorem \ref{Theorem:TimeVaryingCepstrum} says that the spectral envelope of an ANH function should be ``far away'' from $0$.
In the ideal case, we would expect that the spectral envelope is ``slow-varying'' in comparison to the harmonic series in the spectrum, so that the cepstrum can well extract the periodicity-related elements from the filter-like elements. This ideal case is satisfied by the assumption in Theorem \ref{Theorem:TimeVaryingCepstrum} in the sense that the IP information is recovered in the STCT. However, this is not always true for real-world signals; in some challenging cases we could see non-trivial patterns in the spectral envelope, which breaks the assumptions in Theorem \ref{Theorem:TimeVaryingCepstrum}. This contaminates the information associated with the IP information we have interest {in}, and hence causes fake detection of periodicity.
Here we discuss two real scenarios when the spectral envelope has a non-trivial pattern.

The first scenario could be observed in the ECG signal with the fundamental frequency around $\xi_1>0$. For example, in some cases, we could find relatively stronger peaks around $3\xi_1$, $6\xi_1$ and $9\xi_1$ in comparison to other peaks in the spectrum. Therefore, in the cepstrum we can find not only a prominent peak at $q_1=1/\xi_1$ but also a small bump around $q_1/3$. To take a closer look at this phenomenon, we recall that it has been well known that the 12-lead ECG signals, denoted as $E(t)\in \RR^{12}$, are the projection of the representative dipole current, denoted as $d(t)\in\RR^3$, where $t\in\RR$, of the electrophysiological cardiac activity on different directions. Physiologically, for a normal subject $d(t)$ is oscillatory with the period about $1$ second. If we could record $d(t)$, the recorded signal is called the vectocardiogram signal. For the $\ell$-th ECG channel, where $\ell=1,\ldots,12$, there is an associated projection direction $v_\ell\in \RR^3$. The $\ell$-th ECG channel is thus the projection of $d(t)$ on $v_\ell$; that is, $E_\ell(t)=v_\ell^Td(t)$ or $E(t)=v^Td(t)$, where $v=[v_1 v_2 \ldots v_{12}]\in \RR^{3\times 12}$. In general, $v$ changes according to time due to the cardiac axis deviation caused by the respiratory activity and other physical movements. To simplify the discussion, we ignore this facts. Thus, since $d(t)$ is oscillatory, it is clear that $E_\ell(t)$ is also oscillatory. In some cases, this complicated procedure leads to an {oscillation} in the spectral envelop, and hence the first scenario.

{The second example is the sound of clarinets.}
Clarinet is one kind of woodwind instrument which makes air resonating in a cylindrical tube with one ended closed. Because of such a physical structure, the even-numbered harmonics including $2\xi_1$ and $4\xi_1$ are highly suppressed\footnote{The absence of even harmonics is (part of) what is responsible for the ``warm'' or ``dark'' sound of a clarinet compared to the ``bright'' sound of a saxophone.}, which breaks the assumption of Corollary \ref{Corollary:TimeVaryingCepstrum}{, and is discussed after Corollary \ref{Corollary:TimeVaryingCepstrum}}. {But, in many real cases,} the cepstrum of the clarinet note do have a peak at {$q_1$} and its multiples {because the even-number harmonics are not totally eliminated}. 
In several real examples, including the {Clarinet and ECG} examples, the unwanted terms in the above situations can be simply eliminated by hard-thresholding{; however,} it is not that easy {to achieve this naive idea}, and a systematic study of this challenge is needed, where we might incorporate more background knowledge into the analysis.

Next, we discuss another scenario when the proposed method works on multi-component signals. Consider {an {\em octave} signal mentioned in Section \ref{Section:Numerical:ViolinSonata}, where one of the two components has the fundamental frequency $\xi_{1,1}$ and the one of the other is higher than it by one octave, thereby with the fundamental frequency $2\xi_{1,1}$. Suppose the phases of these two components match in a way so that} the spectrum of the multi-component signal has stronger peaks at even-order harmonics, especially {$\xi_{1,2}=2\xi_{1,1}$, $\xi_{1,4}=4\xi_{1,1}$, $\xi_{1,6}=6\xi_{1,1}$, $\xi_{1,8}=8\xi_{1,1}$, etc. In this special case,} the spectral envelope oscillates and we may recall the IF's of both components in the de-shape SST. { Note that this special case contradicts the assumption of the ANH model so that we could not model it as a composition of two ANH functions, and the proposed method may or may not work. In the real-world music, the condition of phase matching does not always happen, and it makes the octave detection problem even harder, as is discussed in Section \ref{Section:Numerical:ViolinSonata}. We mention that while it is a difficult job in signal processing, human beings could identify the difference via learning the oscillatory pattern of the signal or the ``timbre''. For example, the timbre of the note C4, which has the fundamental frequency 262 Hz, is different from that of the combined note (or called ``an interval'' in music) C4+C5, where the fundamental frequency of C5 is 524 Hz. By learning the timbre, we could tell the difference.} The above scenarios all have their own interest but {are} out of the scope of this paper. We will report a systematic study in {a} future work.

We mention that there are several challenging cases in processing real-world multi-pitch signals, like {\em missing fundamental} or {\em stacked harmonics}, both of which have been discussed in \cite{su2015combining,su2016exploiting}. These could be treated as exceptional cases of the proposed model and a modification of the model and algorithm is needed to better handle these signals.

{Last but not the least, from the data analysis viewpoint, in general we cannot decide the model parameters, like the $\epsilon$, and the sequence $c$, a priori. This is an estimation problem in nature, and has been open for a while. However, for most problems we face in practice, we have some background knowledge that could guide us to ``guess the model''. For example, for the fetal ECG extraction problem, the heart rates of the mother and the fetus have a well-known range guided by the physiological background, and this is the information we could use to determine the parameters. But for a randomly given dataset without any background knowledge, at this moment, there is still no ideal way to determine the model parameters directly from the data itself. This fundamental estimation problem will be explored in a future work.}

\appendix
\section{Proof of Theorem \ref{Theorem:TimeVaryingCepstrum}}\label{Appendix:Proof}

In this section, we provide an analysis of STCT in Theorem \ref{Theorem:TimeVaryingCepstrum} step by step
\begin{itemize}
\item first step: approximate the ANH function by a ``harmonized'' function by Taylor's expansion and evaluate its STFT;
\item second step: evaluate the $\gamma$ power of the absolute value of STFT. Since in general there will be more than one ANH component in the ANH function, we have to handle the possible interference between different ANH components.  We will apply the Erd\"os-Tur\'an inequality to control the interference;
\item third step: find the Fourier transform of the $\gamma$ power of the absolute value of STFT and finish the proof.
\end{itemize}

We start from the first Lemma, which allows us to locally approximate an ANH function by a sinusoidal function.

\begin{lem}\label{Lemma:TaylorExpansionAMFM}
Take $\epsilon {>0}$, a sequence $c \in \ell^1$, $N\in\NN$ and $0\leq C<\infty$. For $f(t)=\sum_{\ell=0}^\infty B_\ell(t)\cos(2\pi \phi_\ell(t))\in \mathcal{D}^{c,C,N}_\epsilon$, for each $\ell\in\{0\}\cup\NN$ we have
\begin{align}
|B_\ell(t+s)-B_\ell(t)|\leq &\epsilon c(\ell) |s|(\phi'_1(t)+\frac{1}{2}\|\phi''_{1}\|_{L^\infty}|s|),\label{Lemma1:Proof:Bound1}\\
|\phi'_\ell(t+s)-\phi'_\ell(t)|\leq &\epsilon\ell |s|(\phi'_1(t)+\frac{1}{2}\|\phi''_{1}\|_{L^\infty}|s|).\label{Lemma1:Proof:Bound2}
\end{align}
\end{lem}

\begin{proof}
Assume that $s>0$. The proof for $s\leq 0$ is the same. By the assumption of $B_\ell(t)$, we have
\begin{align}
|B_\ell(t+s)-B_\ell(t)| &=\,\left|\int_0^sB_\ell'(t+u)d u\right|  \nonumber\\
&{  \leq \epsilon c(\ell) \int_0^s\phi_1'(t+u) d u \quad \mbox{by the slowly varying condition (\ref{def:slow_varying}), }}\nonumber\\
&\leq \,\epsilon c(\ell) \int_0^s \left(\phi_1'(t)+\int_0^u \phi_1''({t+}y)dy\right) d u\leq \epsilon c(\ell) \left(\phi_1'(t)s+\frac{1}{2}\|\phi''_1\|_{L^\infty}s^2\right). \nonumber
\end{align}
The proof of (\ref{Lemma1:Proof:Bound2}) follows by the same argument.
{
\begin{align}
|\phi'_{\ell}(t+s) - \phi'_{\ell}(t)| &= \left|\int_0^s \phi''_{\ell}(t+u) du \right| \nonumber \\
&\leq \epsilon \ell \int_0^s \phi_1'(t+u)du \quad \mbox{by the slowly varying condition (\ref{def:slow_varying}), } \nonumber\\
&=\epsilon \ell \int_0^s \left( \phi_1'(t) + \int_0^u \phi_1''(t+y)dy\right)du
\leq \epsilon \ell \left( \phi_1'(t) s + \frac{1}{2} \| \phi_1''\|_{L^{\infty}} s^2 \right). \nonumber
\end{align}
}
\end{proof}

The following Lemma leads to the first part of the Theorem,  (\ref{MainTheorem:STFTExpansion}), regarding the STFT. In short, for the superposition of ANH functions in $\mathcal{D}_{\epsilon,d}$, at each time $t$ the function behaves like a sinusoidal function and the STFT could be approximately explicitly.

\begin{lem}\label{Lemma:STFTexpansion}
Fix $\epsilon {>0}$ and $d>0$. Take $f(t)=\sum_{k=1}^Kf_k(t)\in \mathcal{D}_{\epsilon,d}$. Then, the STFT of $f$ at $t\in\RR$ is
\begin{align}
V^{(h)}_f(t,\xi)=
\frac{1}{2}\sum_{k=1}^K\sum_{\ell=-N_k}^{N_k} B_{k,\ell}(t)\hat{h}(\xi-\phi'_{k,\ell}(t))e^{i2\pi\phi_{k,\ell}(t)}+\epsilon_0(t,\xi),\label{Proof:Lemma:STFTExpansionFormula0}
\end{align}
where $\xi\in\RR$ and $\epsilon_0(t,\xi)$ is defined in (\ref{Lemma:Proof:Definition:epsilon}). Furthermore, $|{\epsilon_0}(t,\xi)|$ is of order $\epsilon$ and decays at the rate of $|\xi|^{-1}$ as $|\xi|\to \infty$.
\end{lem}

\begin{proof}
Since $f\in L^\infty\cap C^1\subset \mathcal{S}'$ and $h\in \mathcal{S}$, by the linearity of the STFT, we have
\begin{align}
V^{(h)}_f(t,\xi)=\sum_{k=1}^K\sum_{\ell=0}^\infty V^{(h)}_{f_{k,\ell}}(t,\xi)\,,
\end{align}
where $f_{k,\ell}(\cdot):=B_{k,\ell}(\cdot)\cos(2\pi \phi_{k,\ell}(\cdot))$ for $\ell=0,1,\ldots$. 
{
Denote
\begin{align}
\tilde{V}^{(h)}_{f_{k,\ell}}(t,\xi) :=  \int B_{k,\ell}(t) \cos(2\pi(\phi_{k,\ell}(t) + \phi'_{k,\ell}(t)(x-t))) h(x-t) e^{-i2\pi \xi (x-t)}dx
\end{align}
where $k=1,\cdots, K$ and $\ell = 0,\cdots, \infty$.}
Next, fix $k\in\{1,\ldots,K\}$, we evaluate the difference between $V^{(h)}_{f_{k,\ell}}(t,\xi)$ and { $\tilde{V}^{(h)}_{f_{k,\ell}}(t,\xi)$}. For each $\ell\in\NN\cup\{0\}$, denote
\begin{equation}
\epsilon_{k,\ell}(t,\xi):=V^{(h)}_{f_{k,\ell}}(t,\xi)-{ \tilde{V}^{(h)}_{f_{k,\ell}}(t,\xi)}\,.\label{Lemma:Proof:Definition:epsilonl}
\end{equation}
We show that $|\epsilon_{k,\ell}(t,\xi)|$ is of order $\epsilon$ and linearly dependent on $c_{k}(\ell)$ for all $t,\xi\in\RR$.
{
First, note that
\begin{align}
\left| \epsilon_{k,\ell}(t,\xi) \right|
&\leq \int \left|B_{k,\ell}(x) -  B_{k,\ell}(t)\right| \left|h(x-t) \right|dx  \nonumber\\
& \quad + B_{k,\ell}(t) \int \left|\cos(2\pi \phi_{k,\ell}(x)) - \cos( 2\pi (\phi_{k,\ell}(t) - \phi'_{k,\ell}(t)(x-t) ))\right| \left|h(x-t) \right| dx
\end{align}
and that
\begin{align}
&\left| \cos(2\pi \phi_{k,\ell}(x)) - \cos( 2\pi (\phi_{k,\ell}(t) - \phi'_{k,\ell}(t)(x-t) )\right| \nonumber \\
\leq\,& 2\pi \left| \phi_{k,\ell}(x)-\phi_{k,\ell}(t)-\phi_{k,\ell}'(t)(x-t) \right|
\leq 2\pi  \int_0^{x-t} \left|\phi'_{k,\ell}(t+u)-\phi'_{k,\ell}(t)\right|du
\end{align}
}
Denote
\begin{equation*}
M_k:=\|\phi'_{k,1}\|_{L^\infty}.
\end{equation*}
Clearly, $\|\phi_{k,1}''\|_{L^\infty}\leq \epsilon M_k$. {  Combining the above inequalities} and Lemma \ref{Lemma:TaylorExpansionAMFM}, we have
\begin{align}
|\epsilon_{k,\ell}(t,\xi)|\leq &\, \int |B_{k,\ell}(x)-B_{k,\ell}(t)||h(x-t)|dx\nonumber\\
&+2\pi B_{k,\ell}(t)\int \int_0^{x-t}|\phi'_{k,\ell}(t+u)-\phi'_{k,\ell}(t)|du |h(x-t)|dx\nonumber\\
\leq &\, \epsilon \big[c_{k}(\ell)\big(\phi_{k,1}'(t)I_1+\frac{1}{2} \epsilon M_kI_2\big)+\pi B_{k,\ell}(t)\ell(\phi_{k,1}'(t)I_2+\frac{1}{3}\epsilon M_kI_3)\big]\nonumber
\end{align}
which is of order $\epsilon$ since $\phi_{k,1}'(t)$ and $B_{k,1}(t)$ are bounded. Note that {$\epsilon_{k,0}(t,\xi)\leq \epsilon c_{k}(\ell)(\phi'_{k,1}(t) + \epsilon M_kI_2/2)$} since the phase { $\phi_{k,0}=0$}. Furthermore, note that $|\epsilon_{k,\ell}(t,\xi)|$ decays at the rate of $|\xi|^{-1}$ as $|\xi|\to \infty$ since 
{
\begin{align}\label{Proof:C1LeadToDistribution}
B_{k,\ell}(x) \cos(2\pi \phi_{k,\ell}(x)) - B_{k,\ell}(t) \cos(2\pi ( \phi_{k,\ell}(t)) + \phi'_{k,\ell}(t)(x-t) ) \in C^1.
\end{align}
 }
Denote
\begin{equation*}
E^{(1)}_{k}(t,\xi):=\sum_{\ell=0}^\infty \epsilon_{k,\ell}(t,\xi),
\end{equation*}
which converges by {(\ref{Condition:ANH:B_elltail2}) that $\sum_{\ell=1}^\infty \ell B_{k,\ell}(t)\leq C_k \sqrt{\frac{1}{4}B^2_{k,0}(t)+\frac{1}{2}\sum_{\ell=1}^\infty B^2_{k,\ell}(t)}$}, and hence
\begin{align}
|E^{(1)}_{k}(t,\xi)|\leq \epsilon& \Big(\|c_k\|_{\ell^1}\big[\phi_{k,1}'(t)I_1+\frac{1}{2}\epsilon M_kI_2 {\big]}\nonumber\\
&+\pi C_k {\sqrt{\frac{1}{4}B^2_{k,0}(t)+\frac{1}{2}\sum_{\ell=1}^\infty B^2_{k,\ell}(t)}} (\phi_{k,1}'(t)I_2+\frac{1}{3}\epsilon M_kI_3) \Big).\label{Proof:Bound:E1k}
\end{align}
Thus, $E^{(1)}_{k}(t,\xi)$ is of order $\epsilon$.

Finally, for each $k\in\{1,\ldots,K\}$, denote
\begin{equation*}
E^{(2)}_{k}(t,\xi):= {  \sum_{\ell =N_k +1}^{\infty} \tilde{V}^{(h)}_{f_{k,\ell}}(t,\xi).}
\end{equation*}
{
By the Plancherel identity, we have
\begin{align}
\tilde{V}^{(h)}_{f_{k,\ell}}(t,\xi)
=\frac{1}{2}B_{k,\ell}(t)[\hat{h}(\xi-\phi'_{k,\ell}(t))e^{i2\pi\phi_{k,\ell}(t)}+\hat{h}(\xi+\phi'_{k,\ell}(t))e^{-i2\pi\phi_{k,\ell}(t)}]\,.
\end{align}
Thus, by} the assumption that {(\ref{Condition:ANH:B_elltail1}) that $\sum_{\ell=N_k+1}^\infty B_{k,\ell}(t)\leq \epsilon \sqrt{\frac{1}{4}B^2_{k,0}(t)+\frac{1}{2}\sum_{\ell=1}^\infty B^2_{k,\ell}(t)}$}, we have 
\begin{equation}\label{Proof:Bound:E2k}
\Big|E^{(2)}_{k}(t,\xi)\Big|\leq \frac{1}{2}\sum_{\ell\in \ZZ\backslash {\{ -N_k, \cdots, N_k\} }} B_{k,\ell}(t)|\hat{h}(\xi-\phi'_{k,\ell}(t))|  \leq \epsilon I_0 {\sqrt{\frac{1}{4}B^2_{k,0}(t)+\frac{1}{2}\sum_{\ell=1}^\infty B^2_{k,\ell}(t)}},
\end{equation}
where the last inequality holds since $\|\hat{h}\|_{L^\infty}\leq I_0$ by a direct bound. Thus, we have
\begin{equation}
{ \sum_{k=1}^K \sum_{\ell =0}^{\infty} \tilde{V}^{(h)}_{f_{k,\ell}}}(t,\xi)=\frac{1}{2}\sum_{k=1}^K\sum_{\ell=-N_k}^{N_k} B_{k,\ell}(t)\hat{h}(\xi-\phi'_{k,\ell}(t))e^{i2\pi\phi_{k,\ell}(t)}+\sum_{k=1}^KE^{(2)}_{k}(t,\xi),\label{Proof:lemma2:epsilonBound2}
\end{equation}
where $|E^{(2)}_{k}(t,\xi)|$ is of order $\epsilon$. Furthermore, $|E^{(2)}_{k}(t,\xi)|$ decays faster than $|\xi|^{-1}$ as $|\xi|\to \infty$ since $\sum_{\ell=1}^\infty B_{k,\ell}(t)<\infty$ and {$\sum_{k=1}^K \sum_{\ell =0}^{N_k} \tilde{V}^{(h)}_{f_{k,\ell}}(t,\xi)$ }decays faster than $|\xi|^{-1}$ as $|\xi|\to \infty$.

We thus have
\begin{equation}
V^{(h)}_{\tilde{f}}(t,\xi)=\sum_{k=1}^K\sum_{\ell=0}^\infty V^{(h)}_{\tilde{f}_{k,\ell}}(t,\xi)=\frac{1}{2}\sum_{k=1}^K\sum_{\ell\in\ZZ} B_{k,\ell}(t)\hat{h}(\xi-\phi'_{k,\ell}(t))e^{i2\pi\phi_{k,\ell}(t)}.
\end{equation}
Putting {(\ref{Lemma:Proof:Definition:epsilonl})} and (\ref{Proof:lemma2:epsilonBound2}) together, we have
\begin{align}
V^{(h)}_f(t,\xi)&\,
=\sum_{k=1}^K\sum_{\ell=0}^\infty V^{(h)}_{f_{k,\ell}}(t,\xi)
=\sum_{k=1}^K\sum_{\ell=0}^\infty [{\tilde{V}^{(h)}_{f_{k,\ell}}}(t,\xi)+\epsilon_{k,\ell}(t,\xi)]\nonumber\\
&\,=\sum_{k=1}^K[\frac{1}{2}\sum_{\ell\in\ZZ} B_{k,\ell}(t)\hat{h}(\xi-\phi'_{k,\ell}(t))e^{i2\pi\phi_{k,\ell}(t)}+E^{(1)}_{k}(t,\xi)]\,,\nonumber\\
&\,=\sum_{k=1}^K\big[\frac{1}{2}\sum_{\ell=-N_k}^{N_k} B_{k,\ell}(t)\hat{h}(\xi-\phi'_{k,\ell}(t))e^{i2\pi\phi_{k,\ell}(t)}+E^{(1)}_{k}(t,\xi)+E^{(2)}_{k}(t,\xi)\big]\,.\nonumber
\end{align}
Denote
\begin{equation}
\epsilon_0(t,\xi):=\sum_{k=1}^K[E^{(1)}_{k}(t,\xi)+E^{(2)}_{k}(t,\xi)], \label{Lemma:Proof:Definition:epsilon}
\end{equation}
which is of order $\epsilon$ and $|\epsilon_0(t,\xi)|$ decays at the rate of $|\xi|^{-1}$ as $|\xi|\to \infty$. We thus have the proof.
\end{proof}

\begin{lem}\label{Lemma:STFTexpansion2}
Fix $\epsilon {>0}$ and $d>0$. Take $f(t)=\sum_{k=1}^Kf_k(t)\in \mathcal{D}_{\epsilon,d}$. Fix a window function $h\in \mathcal{S}$.
For each $t\in\RR$ and $\xi\in\RR$, we have
\begin{align}
&\sum_{k=1}^K\sum_{\ell=-N_k}^{N_k} B_{k,\ell}(t)\hat{h}(\xi-\phi'_{k,\ell}(t))e^{i2\pi\phi_{k,\ell}(t)}\label{Proof:Lemma:STFTexpansion2:FirstBound}\\
=\,&\sum_{k=1}^K\sum_{\ell=-N_k}^{N_k} B_{k,\ell}(t)\hat{h}(\xi-\ell\phi'_{k,1}(t))e^{i2\pi\phi_{k,\ell}(t)}+\epsilon_1(t,\xi),\nonumber
\end{align}
where $\epsilon_1(t,\xi)$ is defined in (\ref{Proof:Lemma:epsilon1}) satisfying
\begin{equation}\label{Proof:Bound:E2}
|\epsilon_1(t,\xi)|\leq \epsilon 2\pi I_1 \sum_{k=1}^K \phi_{k,1}'(t)\sum_{\ell=-N_k}^{N_k} B_{k,\ell}(t)\chi_{\tilde{Z}_{k,\ell}}(\xi),
\end{equation}
where { $\tilde{Z}_{k,\ell}(t) := [(\ell - \epsilon) \phi'_{k,1}(t) - \Delta, (\ell + \epsilon) \phi'_{k,1} + \Delta]$.}
Note that the support of $\epsilon_1(t,\xi)$ is inside $[-\max_{k}((N_k+\epsilon)\phi_{k,1}'(t))-\Delta,\,\max_{k}((N_k+\epsilon)\phi_{k,1}'(t))+\Delta]$. In particular, we have
\begin{align}
V^{(h)}_f(t,\xi)=
\frac{1}{2}\sum_{k=1}^K\sum_{\ell=-N_k}^{N_k} B_{k,\ell}(t)\hat{h}(\xi-\ell\phi'_{k,1}(t))e^{i2\pi\phi_{k,\ell}(t)}+\epsilon_2(t,\xi),\label{Proof:Lemma:STFTExpansionFormula}
\end{align}
where $\epsilon_2(t,\xi)=\epsilon_0(t,\xi)+\epsilon_1(t,\xi)$, which is of order $\epsilon$ and $|\epsilon_2(t,\xi)|$ decays at the rate of $|\xi|^{-1}$ as $|\xi|\to \infty$.
\end{lem}

\begin{proof}
The proof is straightforward by the smoothness assumption of $h$ and Taylor's expansion. Indeed, by the assumption that $\left|\frac{\phi'_{k,\ell}(t)}{\phi'_{k,1}(t)}-\ell\right|\leq \epsilon$, we know that $|\phi'_{k,\ell}(t)-\ell\phi'_{k,1}(t)|\leq \epsilon \phi_{k,1}'(t)$ for all $\ell=1,\ldots$.
Thus, since $\hat{h}$ is compactly supported on $[-\Delta,\Delta]$, we have that for $\xi\in  \tilde{Z}_{k,\ell}$,
\begin{equation}
|\hat{h}(\xi-\phi'_{k,\ell}(t))-\hat{h}(\xi-\ell\phi'_{k,1}(t))|\leq \epsilon\phi_{k,1}'(t)\|\hat{h}'\|_{L^\infty}\leq 2\pi\epsilon \phi_{k,1}'(t)I_1,
\end{equation}
where we use the bound $\|\hat{h}'\|_{L^\infty}\leq 2\pi I_1$; for $\xi\notin \tilde{Z}_{k,\ell}$,
\begin{equation*}
|\hat{h}(\xi-\phi'_{k,\ell}(t))-\hat{h}(\xi-\ell\phi'_{k,1}(t))|=0.
\end{equation*}
Denote
\begin{equation}
\epsilon_1(t,\xi):=\sum_{k=1}^K\sum_{\ell=-N_k}^{N_k} B_{k,\ell}(t)(\hat{h}(\xi-\phi'_{k,\ell}(t))-\hat{h}(\xi-\ell\phi'_{k,1}(t)))e^{i2\pi\phi_{k,\ell}(t)}. \label{Proof:Lemma:epsilon1}
\end{equation}
By a direct bound, we have
\begin{align}
|\epsilon_1(t,\xi)|=&|\sum_{k=1}^K\sum_{\ell=-N_k}^{N_k} B_{k,\ell}(t)(\hat{h}(\xi-\phi'_{k,\ell}(t))-\hat{h}(\xi-\ell\phi'_{k,1}(t)))e^{i2\pi\phi_{k,\ell}(t)}|\\
\leq&\, \sum_{k=1}^K\sum_{\ell=-N_k}^{N_k} B_{k,\ell}(t)|\hat{h}(\xi-\phi'_{k,\ell}(t))-\hat{h}(\xi-\ell\phi'_{k,1}(t))|\nonumber\\
\leq\,& \epsilon 2\pi I_1 \sum_{k=1}^K \phi_{k,1}'(t)\sum_{\ell=-N_k}^{N_k} B_{k,\ell}(t)\chi_{\tilde{Z}_{k,\ell}},\nonumber
\end{align}
which leads to the claim. The proof of (\ref{Proof:Lemma:STFTExpansionFormula}) comes from a direct combination of (\ref{Proof:Lemma:STFTExpansionFormula0}) and (\ref{Proof:Lemma:STFTexpansion2:FirstBound}).

\end{proof}

By the assumption that $0<\Delta\leq \phi_{1,1}'(t)/4$, we know that for a fixed $k\in \{1,\ldots,K\}$, {$Z_{k,i}(t)\cap Z_{k,j}(t)=\emptyset$} for all $i\neq j$, where {$Z_{k,\ell}$} is defined in (\ref{Definition:Zkl}).
Thus, when $K=1$, we know that for any $\gamma>0$, the $\gamma$ power of the absolute value of the major term in (\ref{Proof:Lemma:STFTExpansionFormula}) becomes
\begin{equation*}
\left|\sum_{\ell=-N_1}^{N_1} B_{1,\ell}(t)\hat{h}(\xi-\phi'_{1,\ell}(t))e^{i2\pi\phi_{1,\ell}(t)}\right|^\gamma=\sum_{\ell=-N_1}^{N_1} B^\gamma_{1,\ell}(t)|\hat{h}(\xi-\phi'_{1,\ell}(t))|^\gamma
\end{equation*}
since the supports of $\hat{h}(\xi-\phi'_{1,i}(t))$ and $\hat{h}(\xi-\phi'_{1,j}(t))$ do not overlap, when $i\neq j$.
However, when $K>1$, although $Z_{k,1}(t)\cap Z_{\ell,1}(t)=\emptyset$ when $k\neq \ell$ since $\Delta<d/4$, there is no guarantee that $Z_{k,i}(t)\cap Z_{\ell,j}(t)=\emptyset$ when $k\neq \ell$ and $i\neq j$. So, when $K>1$, we need to be careful when we take the power.

\begin{defn}
Fix $\epsilon { >0}$ and $d>0$. Take $f(t)=\sum_{k=1}^Kf_k(t)\in \mathcal{D}_{\epsilon,d}$. Define ${S_1(t)}=\emptyset$, and for each $k\in\{{2},\ldots,K\}$, define
\begin{equation}
S_{k}(t):=\{i,-i|\, 1\leq i\leq N_k,\,Z_{k,i}(t)\cap Z_{\ell,j}(t)\neq \emptyset,\,j\in { \{1,\ldots,N_\ell \}}\backslash S_\ell(t),\,\ell=1,\ldots,k-1 \} { \cup \{0\}}.
\end{equation}
Furthermore, define
\begin{align}
Y_{\text{no-OL}}(t)&:=\cup_{k=1}^K\cup_{i\in { \{0, \pm1,\ldots,\pm N_k\}} \backslash S_k}{Z}_{k,i}(t)\subset\RR\\
Y_{\text{with-OL}}(t)&:=\cup_{k=1}^K\cup_{i\in S_k} {Z}_{k,i}(t)\subset \RR.\nonumber
\end{align}
\end{defn}

The set $S_k(t)$ indicates the multiples of the $k$-th ANH function that have the danger of overlapping with the other ANH functions. To be more precise,
for $k\in\{2,\ldots,K\}$ and $\ell\in\{1,\ldots,k-1\}$, the supports of {$\hat{h}(\xi-i\phi_{k}'(t))$ and $\hat{h}(\xi-j \phi_{\ell}'(t))$}, where $i\in {\{0,\pm1,\ldots,\pm N_k\}}\backslash S_k$ and $j\in { \{0,\pm1,\ldots,\pm N_\ell\}} \backslash S_\ell$ do not overlap. The sets $Y_{\text{no-OL}}(t)$ and $Y_{\text{with-OL}}(t)$ are used to control the overlapping of {multiples associated with} different ANH components.
Note that the supports of all summands in $\sum_{k=1}^K\sum_{\ell\in \{{0,}\pm1,\ldots,\pm N_k\}\backslash S_k}B_{k,\ell}(t)|\hat{h}(\xi-\ell\phi'_{k,1}(t))|$ do not overlap. 

{To evaluate $|V^{(h)}_f(t,\xi)|^\gamma$, we} need the following bounds to control the influence of taking the $\gamma$ power.

\begin{lem}\label{Lemma:SmallGammaExpansion}
Suppose $x\geq y\geq 0$.
For $0<\gamma\leq 1$, we have
\begin{align}
(x+y)^\gamma\leq x^\gamma+\gamma y^\gamma\,.
\end{align}
\end{lem}

\begin{proof}
When $x=y=0$, this is the trivial case.  Suppose $x\geq y> 0$ or $x>y\geq0$. By Taylor's expansion, we have
\begin{equation}
(x+y)^\gamma=x^\gamma(1+\frac{y}{x})^\gamma\leq x^\gamma+\gamma\frac{y}{x}x^\gamma=x^\gamma+\gamma\big(\frac{y}{x}\big)^{1-\gamma}y^\gamma.
\end{equation}
Since $y/x\leq 1$, we obtain the  bound.
\end{proof}

\begin{lem}\label{Lemma:STFTabsGammaPower}
Suppose Assumption \ref{Assumption:GeneralAssumption} holds and take $0<\gamma\leq 1$. Then we have
\begin{align}
|V^{(h)}_f(t,\xi)|^\gamma&\,=\frac{1}{2^\gamma}\sum_{k=1}^K\sum_{\ell=-N_k}^{N_k}B^\gamma_{k,\ell}(t)|\hat{h}(\xi-\ell\phi'_{k,1}(t))|^\gamma+\delta_3(t,\xi)+{\epsilon_3}(t,\xi)\,,\label{Equation:STFTabsGammaPower}
\end{align}
where $\delta_3(t,\xi)$ is defined in (\ref{Proof:Lemma:Definition:delta3}) { and $\epsilon_3(t,\xi)$} is defined in (\ref{Theorem:Statement:STCT:epsilon4}). Moreover, $\delta_3(t,\xi)=0$ when $K=1$. When $K>1$, $\delta_3(t,\xi)$ is supported on $Y_{\text{with-OL}}(t)$ and is bounded by {$\frac{I_0^\gamma}{2^\gamma}\sum_{k=2}^K\sum_{\ell\in S_k}B^\gamma_{k,\ell}(t)\chi_{Z_{k,\ell}}(\xi)$}. $\epsilon_3(t,\xi)$ satisfies $|{\epsilon_3}(t,\xi)|\leq |\epsilon_{2}(t,\xi)|^\gamma$.
\end{lem}

{
\begin{proof}
Let $\delta_3(t,\xi)$ and ${\epsilon_3}(t,\xi)$ be defined as
\begin{equation}
\label{Proof:Lemma:Definition:delta3}
\delta_3(t,\xi) := \Big|\frac{1}{2} \sum_{k=1}^K \sum_{\ell = -N_k}^{N_k} B_{k,\ell}(t) \hat{h}(\xi - \ell \phi'_{k,1}(t))e^{i2\pi\phi_{k,\ell}(t)}\Big|^{\gamma}  - \frac{1}{2^{\gamma}} \sum_{k=1}^K \sum_{\ell = -N_k}^{N_k} B_{k,\ell}^{\gamma}(t) |\hat{h}(\xi - \ell \phi'_{k,1}(t))|^{\gamma}
\end{equation}
and
\begin{align}
\label{Theorem:Statement:STCT:epsilon4}
{\epsilon_3}(t,\xi) := |V^{(h)}_f(t,\xi)|^\gamma  -\Big|\frac{1}{2} \sum_{k=1}^K \sum_{\ell = -N_k}^{N_k} B_{k,\ell}(t) \hat{h}(\xi - \ell \phi'_{k,1}(t))e^{i2\pi\phi_{k,\ell}(t)}\Big|^{\gamma}.
\end{align}
That is,
\begin{equation}\label{Proof:Theorem:VfgammaPower:expansion}
 |V^{(h)}_f(t,\xi)|^\gamma  =  \frac{1}{2^{\gamma}} \sum_{k=1}^K \sum_{\ell = -N_k}^{N_k} B_{k,\ell}^{\gamma}(t) |\hat{h}(\xi - \ell \phi'_{k,1}(t))|^{\gamma} + \delta_3(t,\xi) + {\epsilon_3}(t,\xi).
\end{equation}
According to Lemmas \ref{Lemma:SmallGammaExpansion} and \ref{Lemma:STFTexpansion2}, when $\epsilon$ is small enough, by the triangular inequality that $\big||V^{(h)}_f(t,\xi)|- |\frac{1}{2} \sum_{k=1}^K \sum_{\ell = -N_k}^{N_k} B_{k,\ell}(t) \hat{h}(\xi - \ell \phi'_{k,1}(t))e^{i2\pi\phi_{k,\ell}(t)}|\big|\leq|\epsilon_2(t,\xi)|$, we have
\begin{equation}
|{\epsilon_3}(t,\xi)| \leq |\epsilon_2(t,\xi)|^{\gamma}.
\end{equation}

Note that when  $\xi \in Y_{\text{no-OL}}(t)$, $\delta_3(t,\xi) =0$ since the supports of all summands in $\sum_{k=1}^K\sum_{\ell =-N_k}^{N_k}B_{k,\ell}(t)|\hat{h}(\xi-\ell\phi'_{k,1}(t))|$ do not overlap for each $\xi\in Y_{\text{no-OL}}(t)$. 
Therefore, we have
\begin{equation}
\label{Proof:Lemma:Definition:delta3part2}
\delta_3(t,\xi) = \frac{1}{2^\gamma}  \left(\Big|\sum_{k=2}^K\sum_{\ell \in S_k} B_{k,\ell}(t) \hat{h}(\xi - \ell \phi'_{k,1}(t))e^{i2\pi\phi_{k,\ell}(t)}\Big|^{\gamma}  - \sum_{k=2}^K\sum_{\ell \in S_k}B_{k,\ell}^{\gamma}(t) |\hat{h}(\xi - \ell \phi'_{k,1}(t))|^{\gamma}\right)\,.
\end{equation}
Hence,
\begin{align*}
|\delta_3(t,\xi) |
&= \left| \Big|\frac{1}{2} \sum_{k=1}^K \sum_{\ell \in S_k(t)} B_{k,\ell}(t) \hat{h}(\xi - \ell \phi'_{k,1}(t))e^{i2\pi\phi_{k,\ell}(t)}\Big|^{\gamma}  - \frac{1}{2^{\gamma}} \sum_{k=1}^K \sum_{\ell \in S_k(t)} B_{k,\ell}^{\gamma}(t) |\hat{h}(\xi - \ell \phi'_{k,1}(t))|^{\gamma} \right| \\
& = \frac{1}{2^{\gamma}} \sum_{k=1}^K \sum_{\ell \in S_k(t)} B_{k,\ell}^{\gamma}(t) |\hat{h}(\xi - \ell \phi'_{k,1}(t))|^{\gamma}
- \left|\frac{1}{2} \sum_{k=1}^K \sum_{\ell \in S_k(t)} B_{k,\ell}(t) \hat{h}(\xi - \ell \phi'_{k,1}(t))e^{i2\pi\phi_{k,\ell}(t)}\right|^{\gamma}\,,
\end{align*}
since $\big|\frac{1}{2} \sum_{k=1}^K \sum_{\ell \in S_k(t) } B_{k,\ell}(t) \hat{h}(\xi - \ell \phi'_{k,1}(t))e^{i2\pi\phi_{k,\ell}(t)}\big|^{\gamma}  \leq \frac{1}{2^{\gamma}} \sum_{k=1}^K \sum_{\ell \in S_k(t) } B_{k,\ell}^{\gamma}(t) |\hat{h}(\xi - \ell \phi'_{k,1}(t))|^{\gamma}$  by Lemma \ref{Lemma:SmallGammaExpansion}.
Note that when $K=1$, $S_1(t) = \emptyset$. Putting these together, we have
\begin{equation}
\label{Proof:Lemma:Bound:delta3}
|\delta_3(t,\xi)| \leq \frac{1}{2^\gamma}\sum_{k=2}^K\sum_{\ell\in S_k}B^\gamma_{k,\ell}(t)\|\hat{h}\|^\gamma_{L^\infty}\chi_{Z_{k,\ell}}(\xi)\leq \frac{I^\gamma_0}{2^\gamma} \sum_{k=2}^K\sum_{\ell\in S_k}B^\gamma_{k,\ell}(t)\chi_{Z_{k,\ell}}(\xi)
\end{equation}
which completes the proof.
\end{proof}
}

Before finishing the proof, we need to control the error introduced by $\delta_3(t,\xi)$ in Lemma \ref{Lemma:STFTabsGammaPower} {when $K\geq 2$}. Note that $\delta_3(t,\xi)$ is supported on $Y_{\text{with-OL}}(t)$. We now control this set.

\begin{lem}\label{Lemma:ErdosTuranBound}
Suppose Assumption \ref{Assumption:GeneralAssumption} holds and $K>1$. For each $t\in\RR$, we have for each $k\in\{2,\ldots,K\}$ the following bound:
\begin{equation}
\frac{\#S_{k}(t)}{N_k}\leq { \sum_{\ell=1}^{k-1}}\Big[\frac{4\Delta}{\phi'_{\ell,1}(t)} +E^{(\ell)}(N_k)\Big],
\end{equation}
where $\#S_{k}(t)$ is the cardinal number of the set $S_{k}(t)$ and $E^{(\ell)}(N_k)\geq0$ is defined in (\ref{Theorem:Statement:STCT:Eell}). Clearly $\frac{\#S_{1}(t)}{N_1}=0$.
\end{lem}

This Lemma gives a bound of the set $S_k(t)$, which indicates that only a small fraction of the multiples of the $k$-th ANH function has the danger of overlapping with other ANH function.

\begin{proof}
Fix $k\in\{2,3,\ldots,K\}$ and $\ell\in\{1,\ldots,k-1\}$. Define a set
\begin{equation}
S_{k,\ell}(t):=\{m|\, { m\in \NN\cup\{0\}},\,Z_{k,m}(t)\cap Z_{\ell,j}\neq \emptyset,\,j\in { \NN\cup\{0\}}\},
\end{equation}
which is the set of multiples of $\phi'_{k,1}(t)$ that {overlap} some multiples of $\phi'_{\ell,1}(t)$.
Clearly, { $ S_{k}(t) \subset \cup_{\ell=1}^{k-1}S_{k,\ell}(t)$} and $S_{k,\ell_1}(t)$ and $S_{k,\ell_2}(t)$ might overlap when $\ell_1\neq \ell_2$. Thus, $\#S_{k}(t)\leq { \sum_{\ell=1}^{k-1}}\#S_{k,\ell}(t)$.
To evaluate the cardinality of the set $S_{k,\ell}(t)$, denote a sequence $s_{k,\ell}(m)$, $m\in \NN$, so that
\begin{equation}
s_{k,\ell}(m)=m\phi'_{k,1}(t)\,\, (\text{mod }\phi'_{\ell,1}(t)).
\end{equation}
By the compactly supported assumption of $\hat{h}$, when $s_{k,\ell}(m)$ lands in
\begin{equation*}
\mathcal{Z}_{k,\ell}:=[0,2\Delta]\cup [\phi'_{\ell,1}(t)-2\Delta,\phi'_{\ell,1}(t)),
\end{equation*}
we know that $Z_{k,m}(t)\cap Z_{\ell,j}\neq \emptyset$ for some $j$; that is,
\begin{equation*}
S_{k,\ell}(t)=\{ {0} \leq m \leq N_k \, | \,s_{k,\ell}(m)\in \mathcal{Z}_{k,\ell}  \}.
\end{equation*}

When $\phi'_{k,1}(t)/\phi'_{\ell,1}(t)$ is a rational number, that is, $\phi'_{k,1}(t)/\phi'_{\ell,1}(t)=a/b$, where $a,b\in \NN$ and are co-prime numbers, then the sequence $\{s_{k,\ell}(m)\}_{m\in \NN}$ only lands on $\{0,\phi'_{\ell,1}(t)/b,\ldots,(b-1)\phi'_{\ell,1}(t)/b\}$ uniformly on $[0,\phi'_{\ell,1}(t))$ since the integer $a$ has a multiplicative inverse modulo $b$; that is, there exists $n_0$ such that $an_0\,\, (\text{mod } b)=1$. Thus the claim holds with the worst bound
\begin{equation}
\frac{\#S_{k}(t)}{N_k}\leq \sum_{\ell=1}^{k-1}\frac{4\Delta}{\phi'_{\ell,1}(t)}.
\end{equation}

When $\phi'_{k,1}(t)/\phi'_{\ell,1}(t)$ is an irrational number, the sequence $\{s_{k,\ell}(m)\}$ is equidistributed on $[0,\phi'_{\ell,1}(t)]$ by Weyl's criterion.
We apply the following well-known Erd\"os-Tur\'an inequality \cite[Corollary 1.1]{Montgomery:1994} to {bound}  $\frac{\#S_{k,\ell}(t)}{N_k}$:
\begin{align}
& \Big|  \frac{\#S_{k,\ell}(t)}{N_k} - \frac{4\Delta}{\phi'_{\ell,1}(t)} \Big|
 \leq \frac{1}{J+1} + \frac{3}{N_k}\sum_{n=1}^{J} \frac{1}{n} \left| \sum_{m={0}}^{N_k} e^{i2 \pi n s_{k,\ell}(m) } \right| \label{Proof:Overlap:Bound:ErdosTuran}
\end{align}
for all positive $J$. Denote $E^{(\ell)}_J(N_k)$ to be the right hand side of (\ref{Proof:Overlap:Bound:ErdosTuran}). Then the best upper bound we could obtain from Erd\"os-Tur\'an inequality is
\begin{equation}
E^{(\ell)}(N_k):=\min_{J\in\NN}E^{(\ell)}_J(N_k), \label{Theorem:Statement:STCT:Eell}
\end{equation}
which goes to zero when $N_k\to \infty$; that is, when $N_k\to \infty$, the chance that $s_{k,\ell}(m)$ would land in $\mathcal{Z}_{k,\ell}$ is $\frac{4\Delta}{\phi'_{\ell,1}(t)}$.
Thus, in general we know that for the pair $(k,\ell)$,  we have
\begin{equation*}
\frac{\#S_{k,\ell}(t)}{N_k}\leq \frac{4\Delta}{\phi'_{\ell,1}(t)}+E^{(\ell)}(N_k)
\end{equation*}
and hence
\begin{equation}
\#S_{k,\ell}(t)\leq N_k\Big[\frac{4\Delta}{\phi'_{\ell,1}(t)}+E^{(\ell)}(N_k)\Big],
\end{equation}
which is the number of multiples of $\phi'_{k,1}(t)$ that are close to some multiples of $\phi'_{\ell,1}(t)$. In conclusion, we have
\begin{equation}
\frac{\#S_{k}(t)}{N_k}\leq \sum_{\ell=1}^{k-1}\Big[\frac{4\Delta}{\phi'_{\ell,1}(t)} +E^{(\ell)}(N_k)\Big].
\end{equation}

\end{proof}

By putting the above Lemmas together, we can prove Theorem \ref{Theorem:TimeVaryingCepstrum}, which shows that the STCT does provide the necessary information for the fundamental IF of the ANH function, even when there are more than one component.

\begin{proof}[Proof of Theorem \ref{Theorem:TimeVaryingCepstrum}]
Note that in general $|V^{(h)}_f(t,\cdot)|^\gamma$ is a tempered distribution, so we can define the Fourier transform in the distribution sense. {Define a $\ell^1$ sequence $b_k$, where $b_k(\ell)=B^\gamma_{k,\ell}(t)$ for all $\ell\in\{0,\ldots,N_k\}$,  $b_k(\ell)=0$ for all $\ell>N_k$, and $b_k(-\ell)=b_k(\ell)$ for all $\ell\in \NN\cup\{0\}$. By a direct calculation, for $q>0$, we have
\begin{align}
&\mathcal{F}(\sum_{\ell=-N_k}^{N_k}B^\gamma_{k,\ell}(t)\delta_{\ell\phi_{k,1}'(t)}\star |\hat{h}|^\gamma)(q)\nonumber\\
=&\widehat{|\hat{h}|^\gamma}(q)\sum_{\ell=-N_k}^{N_k}b_k(\ell)e^{i2\pi\ell\phi_{k,1}'(t)q}=\widehat{|\hat{h}|^\gamma}(q)\sum_{\ell=-\infty}^{\infty}b_k(\ell)e^{i2\pi\ell\phi_{k,1}'(t)q}=\widehat{|\hat{h}|^\gamma}(q){\hat{b}_k(q)},\label{Theorem:Statement:STCT:bk}
\end{align}
where $\hat{b}_k$ is the discrete-time Fourier transform of the $\ell^1$ sequence $b_k$, which is a continuous and real.} 

For the term $\delta_3$, since $\delta_3(t,\cdot)$ is compactly supported, continuous by (\ref{Proof:Lemma:Definition:delta3part2}) and is bounded by (\ref{Proof:Lemma:Bound:delta3}), $\delta_3(t,\cdot)\in L^1$ and its Fourier transform could be well defined as a function. {Since the support of $\delta_3$, which is determined by the overlapped multiples of different ANH functions, could not be controlled, we apply the Riemann-Lebesgue theorem to evaluate a simple} bound:
\begin{align}
&|\int  \delta_3(t,\xi)e^{-i2\pi\xi q}d \xi|\leq  { \frac{I_0^{\gamma}}{2^{\gamma}}}\sum_{k=2}^K \sum_{\ell\in S_k(t)}B^\gamma_{k,\ell}(t)\int \chi_{{Z}_{k,\ell}}(\xi)  d\xi \nonumber\\
\leq \,
& { 2\Delta I_0^{\gamma} \sum_{k=2}^K  \sum_{\ell\in S_k(t)}B^\gamma_{k,\ell}(t) \leq 2 \Delta I_0^{\gamma}\sum_{k=2}^K  B^\gamma_{k,1}(t)\sum_{\ell\in S_k(t)}c^\gamma_k(\ell) }\label{Proof:Lemma:Definition:delta3Fourier}
\end{align}
since {$|Z_{k,\ell}| =2\Delta $}.
To control $\sum_{\ell\in S_k(t)}c_k^\gamma(\ell)$, we apply the simple bound $c_k(\ell)\leq \|c_k\|_{\ell^\infty}$ for all $\ell=0,1,\ldots,N_k$. This leads to
\begin{align}
\sum_{\ell\in S_k(t)}c^\gamma_k(\ell)&\leq \#S_{k}(t)\|c_k^\gamma\|_{\ell^\infty}\leq \|c_k^\gamma\|_{\ell^\infty}N_k\sum_{\ell=1}^{k-1}\Big[\frac{4\Delta}{\phi'_{\ell,1}(t)}+E^{(\ell)}(N_k)\Big],\nonumber
\end{align}
where the last inequality holds by Lemma \ref{Lemma:ErdosTuranBound}. Thus, the first term
\begin{equation}
E_1:=\mathcal{F}[\delta_3(t,\cdot)]\label{Theorem:Statement:STCT:E1}
\end{equation}
is bounded by
{
\begin{equation*}
|E_1|\leq 2\Delta I_0^\gamma \sum_{k=2}^K  B^\gamma_{k,1}(t)\|c_k^\gamma\|_{\ell^\infty}N_k\sum_{\ell=1}^{k-1}\Big[\frac{4\Delta}{\phi'_{\ell,1}(t)}+E^{(\ell)}(N_k)\Big].
\end{equation*}
}
{ Note that} $K=1$, since {$\delta_3(t,\xi)=0$}, we know that $E_1=0$ and the bound holds trivially.

The error term ${\epsilon_3}(t,\xi)$ is of order $\epsilon^\gamma$ but in general it decays at the rate of $|\xi|^{-\gamma}$ as $|\xi|\to \infty$, so its Fourier transform is evaluated in the distribution sense. Denote $E_2:=\mathcal{F}[{\epsilon_3}(t,\cdot)]$. We have
\begin{equation}
|E_2(\psi)|=\big|\int {\epsilon_3}(t,\xi) \hat{\psi}(\xi)d\xi\big|\leq \|{\epsilon_3}(t,\cdot)\|_{L^\infty} \|\hat{\psi}\|_{L^1}\label{Theorem:Statement:STCT:E2}
\end{equation}
for all $\psi\in\mathcal{S}$. We have thus {obtained} the claim.

\end{proof}

\begin{rem}\label{Remark:BoundE1}
{Note that the bound for $E_1$, which is the Fourier transform of $\delta_3$, is the worst bound, since we could not control the locations of the overlaps between those multiples of different ANH components in the STFT. 
The problem we encounter could be simplified to the following analytic number theory problem: given an irrational number $\alpha$. Denote $\beta_n=n\alpha-[n\alpha]$, where $n\in\NN\cup\{0\}$ and $[x]$ means the integer part of $x$. Denote the set $I=\{n,-n| n\in\NN\cup\{0\},\,0\leq \beta_n<\zeta\}\cup\{n,-n|\beta_n>1-\zeta\}$, where $\zeta>0$ is a small number. Then, what is the spectral distribution of $\sum_{n\in I}\delta_n\star g$, where $g$ is a smooth and compact function supported on $[-\zeta/2,\zeta/2]$? }
\end{rem}

{
\begin{proof}[Proof of Corollary \ref{Corollary:TimeVaryingCepstrum}]

By (\ref{Assumption:ANH:B_ell}), $b_k(\ell)$ is non-zero for $\ell\in\{-N_k,\ldots,0,\ldots,N_k\}$. Thus $\hat{b}_k$ is  a continuous, real, and periodic function with the period equal to $1/\phi_{k,1}'(t)$. 
By (\ref{Proof:Bound:E1k}), (\ref{Proof:Bound:E2k}), and (\ref{Proof:Bound:E2}), $\epsilon_2(t,\xi)$ is bounded by $Q\epsilon$, where
\begin{align*}
Q:=\sum_{k=1}^K\Big[ &\big(\|c_k\|_{\ell^1}[\phi_{k,1}'(t)I_1+\frac{1}{2}\epsilon M_kI_2]+\pi C_k {\sqrt{\frac{1}{4}B^2_{k,0}(t)+\frac{1}{2}\sum_{\ell=1}^\infty B^2_{k,\ell}(t)}} (\phi_{k,1}'(t)I_2+\frac{1}{3}\epsilon M_kI_3) \big) \\
+& I_0\sqrt{\frac{1}{4}B^2_{k,0}(t)+\frac{1}{2}\sum_{\ell=1}^\infty B^2_{k,\ell}(t)} +
 2\pi I_1  \phi_{k,1}'(t)\sum_{\ell=-N_k}^{N_k} B_{k,\ell}(t)\chi_{\tilde{Z}_{k,\ell}}\Big].
\end{align*}
Thus, when $\sqrt{\frac{1}{4}B^2_{k,0}(t)+\frac{1}{2}\sum_{\ell=1}^\infty B^2_{k,\ell}(t)}$ is sufficiently large and $\epsilon$ is sufficiently small, $\frac{1}{2^{\gamma}} \sum_{k=1}^K \sum_{\ell = -N_k}^{N_k} B_{k,\ell}^{\gamma}(t) |\hat{h}(\xi - \ell \phi'_{k,1}(t))|^{\gamma}$ dominantes $|\epsilon_2(t,\xi)|^\gamma$, since $B^\gamma_{k,\ell}(t)>\epsilon^{\gamma/2}\big(\frac{1}{4}B^2_{k,0}(t)+\frac{1}{2}\sum_{\ell=1}^\infty B^2_{k,\ell}(t)\big)^{\gamma/2}$ and $\epsilon_3(t,\xi)$ is bounded by $Q^\gamma\epsilon^\gamma$.
Moreover, when $\Delta N_k$ is sufficiently small, $\frac{1}{2^{\gamma}} \sum_{k=1}^K \sum_{\ell = -N_k}^{N_k} B_{k,\ell}^{\gamma}(t) |\hat{h}(\xi - \ell \phi'_{k,1}(t))|^{\gamma}$ also dominates $\delta_3(t,\xi)$, and hence we finish the proof.
\end{proof}
}

\section{Acknowledgement}
Hau-tieng Wu's research is partially supported by Sloan Research Fellow FR-2015-65363. Part of this work was done during Hau-tieng Wu's visit to National Center for Theoretical Sciences, Taiwan, and he would like to thank NCTS for its hospitality. Hau-tieng Wu also thanks Dr. Ilya Vinogradov for the discussion of equidistribution sequences. The authors thank Professor Stephen W. Porges for sharing the non-contact PPG signal. {The authors acknowledge the anonymous reviewers for their valuable recommendations to improve the manuscript.}

\bibliographystyle{amsplain}
\bibliography{cepstrumBIB}

\end{document}